\newcommand{\isPreprint}{1}
\if\isPreprint1
\newcommand{\isNotPreprint}{0}
\else
\newcommand{\isNotPreprint}{1}
\fi
\newcommand{\AAAIdoesNotKnowHowToUseLaTeX}{\isNotPreprint}
\def\ConferencedoesNotKnowHowToUseLaTeX{\AAAIdoesNotKnowHowToUseLaTeX}
\newcommand{\boolHardCodedAppendixRef}{\isNotPreprint}%
\newcommand{\showAppendix}{\isPreprint}%
\if\AAAIdoesNotKnowHowToUseLaTeX 1
\documentclass[letterpaper]{article} %
\else
\documentclass[letterpaper,backref=page]{article}
\fi
\if\isPreprint1
\usepackage{arXiv_aaai24} 
\else
\usepackage{aaai24}  %
\fi
\usepackage{times}  %
\usepackage{helvet}  %
\usepackage{courier}  %
\usepackage[hyphens]{url}  %
\usepackage{graphicx} %
\usepackage{natbib}  %
\usepackage{caption} %
\title{Machine Learning-Powered Combinatorial Clock Auction}

\author{Ermis Nikiforos Soumalias\textsuperscript{\rm 1,\rm 3}\footnotemark[2]{\normalfont,}
	Jakob Weissteiner\textsuperscript{\rm 1,\rm 3}\footnotemark[2]{\normalfont,}
	Jakob Heiss\textsuperscript{\rm 2,\rm 3}{\normalfont,}
	Sven Seuken\textsuperscript{\rm 1,\rm 3}
}
\affiliations{
\textsuperscript{\rm 1}University of Zurich\\
\textsuperscript{\rm 2}ETH Zurich\\
\textsuperscript{\rm 3}ETH AI Center\\
ermis@ifi.uzh.ch,
weissteiner@ifi.uzh.ch,
jakob.heiss@math.ethz.ch,
seuken@ifi.uzh.ch
}

\usepackage{our_commands} %

\begin{document}

\maketitle

\begin{abstract}
We study the design of \emph{iterative combinatorial auctions (ICAs)}.
The main challenge in this domain is that the bundle space grows exponentially in the number of items. To address this, several papers have recently proposed machine learning (ML)-based preference elicitation algorithms that aim to elicit only the most important information from bidders. However, from a practical point of view, the main shortcoming of this prior work is that those designs elicit bidders' preferences via \emph{value queries} (i.e., ``What is your value for the bundle $\{A,B\}$?''). In most real-world ICA domains, value queries are considered impractical, since they impose an unrealistically high cognitive burden on bidders, which is why they are not used in practice. In this paper, we address this shortcoming by designing an \emph{ML-powered combinatorial clock auction} that elicits information from the bidders only via \emph{demand queries} (i.e., ``At prices $p$, what is your most preferred bundle of items?''). We make two key technical contributions: First, we present a novel method for training an ML model on demand queries. Second, based on those trained ML models, we introduce an efficient method for determining the demand query with the highest clearing potential, for which we also provide a theoretical foundation. We experimentally evaluate our ML-based demand query mechanism in several spectrum auction domains and compare it against the most established real-world ICA: the \emph{combinatorial clock auction (CCA)}.
Our mechanism significantly outperforms the CCA in terms of efficiency in all domains, it achieves higher efficiency in a significantly reduced number of rounds, and, using linear prices, it exhibits vastly higher clearing potential. Thus, with this paper we bridge the gap between research and practice and propose the first practical ML-powered ICA. 
\end{abstract}

\renewcommand{\thefootnote}{\fnsymbol{footnote}}
\footnotetext[1]{\if\isPreprint1
This paper is the full version of \cited{Soumalias2024MLCCA} published at AAAI'24 including the appendix.
\else
The full paper including appendix is available on arXiv via: \url{https://arxiv.org/abs/2308.10226}.%
\fi}
\footnotetext[2]{These authors contributed equally.}
\renewcommand{\thefootnote}{\arabic{footnote}}

\section{Introduction}\label{sec:Introduction}

\textit{Combinatorial auctions (CAs)} are used to allocate multiple items among several bidders who may view those items as complements or substitutes. In a CA, bidders are allowed to submit bids over \textit{bundles} of items. CAs have enjoyed widespread adoption in practice, with their applications ranging from allocating spectrum licences \citep{cramton2013spectrumauctions} to TV ad slots \citep{tvadauctions} and airport landing/take-off slots \citep{rassenti1982combinatorial}. 

One of the key challenges in CAs is that the bundle space grows exponentially in the number of items, making it infeasible for bidders to report their full value function in all but the smallest domains. 
Moreover, \citet{nisan2006communication} showed that for general value functions, CAs require an exponential number of bids in order to achieve full efficiency in the worst case. 
Thus, practical CA designs cannot provide efficiency guarantees in real world settings with more than a modest number of items. 
Instead, the focus has shifted towards \emph{iterative combinatorial auctions (ICAs)}, where bidders interact with the auctioneer over a series of rounds, providing a limited amount of information, and the aim of the auctioneer is to find a highly efficient allocation. 

The most established mechanism following this interaction paradigm is the \emph{combinatorial clock auction (CCA)} \citep{ausubel2006clock}. The CCA has been used extensively for spectrum allocation, generating over $\$ 20$ Billion in revenue between $2012$ and $2014$ alone \citep{ausubel2017practical}. 
Speed of convergence is a critical consideration for any ICA since each round can entail costly computations and business modelling for the bidders 
\citep{kwasnica2005iterative, milgrom2017Designing, bichler2017coalition}. 
Large spectrum auctions following the CCA format can take more than $100$ bidding rounds. In order to decrease the number of rounds, many CAs in practice use aggressive price update rules (e.g., increasing prices by up to $10$\% each round), which can harm efficiency \cite{ausubel2017practical}. 
Thus, it remains a challenging problem to design a practical ICA that elicits information via demand queries, is  efficient, and converges in a small number of rounds. Specifically, given the value of resources allocated in such real-world ICAs, increasing their efficiency by even one percentage point already translates into monetary gains of hundreds of millions of dollars. 

\subsection{ML-Powered Preference Elicitation}\label{subsec:Machine Learning-powered Assignemt Mechanims}

To address this challenge, researchers have proposed various ways of using machine learning (ML) to improve the efficiency of CAs. 
The seminal works by \citet{blum2004preference} and \citet{lahaie2004applying} were the first to frame preference elicitation in CAs as a learning problem.
In the same strand of research, \citet{brero2018combinatorial, brero2021workingpaper}, \citet{weissteiner2020deep} and \citet{weissteiner2022fourier} proposed ML-powered ICAs. 
At the heart of those approaches lies an ML-powered preference elicitation algorithm that uses an ML model to approximate each bidder's value function and to generate the next value query, which in turn refines that bidder's model.  
\citet{weissteiner2022monotone} designed a special network architecture for this framework while \citet{weissteiner2023bayesian} incorporated a notion of uncertainty \citep{heiss2022nomu} into the framework, further increasing its efficiency. 
Despite their great efficiency gains compared to traditional CA designs, those approaches suffer from one common limitation: they fundamentally rely on value queries of the form ``What is your value for bundle $\{A,B\}$''. 
Prior research in auction design has identified demand queries (DQs) as the best way to run an auction \cite{cramton2013spectrumauctions}. 
Their advantages compared to value queries include elimination of tacit collusion and bid signaling, as well as simplified bidder decision-making that keeps the bidders focused on what is most relevant: the relationship between prices and aggregate demand.
Additionally, value queries are cognitively more complex, and thus typically impractical for real-world ICAs.
For these reasons, DQs are the most prominent interaction paradigm for auctions in practice.

Despite the prominence of DQs in real-world applications, the only prior work on ML-based DQs that we are aware of is that of \citet{brero2018bayesian} and \citet{brero2019fast}, who proposed integrating ML in a price-based ICA to generate the next price vector in order to achieve faster convergence. Similar to our design, in these works the auctioneer maintains a model of each agent's value function, which are updated as the agents bid in the auction and reveal more information about their values. 
Then, those models are used in each round to compute new prices and drive the bidding process. 
Unlike our approach, the design of this prior work focuses solely on clearing potential, as the authors do not report efficiency results. 
Additionally, their design suffers from some significant limitations: 
(i) it does not exploit any notion of similarity between bundles that contain overlapping items, 
(ii) it only incorporates a fraction of the information revealed by the agents' bidding. Specifically, it only makes use of the fact that for the bundle an agent bids on, her value for that bundle must be larger than its price, and  
(iii) their approach is computationally intractable already in medium-sized auction domains,
as their price update rule requires a large number $l$ of posterior samples 
for \emph{expectation maximization} and then solving a linear program whose number of constraints for each bidder is proportional to $l$ times the number of bids by that agent.
These limitations are significant, as they can lead to large efficiency decreases in complex combinatorial domains. Moreover, their design cannot be easily modified to alleviate these limitations. 
In contrast, our approach effectively addresses all of these limitations. 

\subsection{Our Contributions}\label{subsec:Contribution}
In this paper, we address the main shortcomings of prior work by designing an ML-powered combinatorial clock auction. Our auction elicits information from bidders via \emph{demand queries  (DQs)} instead of value queries, while simultaneously, unlike prior work on ML-based DQs, being computationally feasible for large domains and incorporating the \textit{complete information} the demand query observations provide into the training of our ML models.
Concretely, we use \emph{Monotone-Value Neural Networks (MVNNs)} \cite{weissteiner2022monotone} as ML models, which are tailored to model monotone and non-linear combinatorial value functions in CAs.

The main two technical challenges are (i) training those MVNNs only on \emph{demand query} observations and (ii) efficiently determining the next demand query that is most likely to clear the market based on the trained MVNNs.
In detail, we make the following contributions: 
\begin{enumerate}[leftmargin=*,topsep=0pt,partopsep=0pt, parsep=0pt]
    \item We first propose an adjusted MVNN architecture, which we call \emph{multiset MVNNs (mMVNNs)}. mMVNNs can be used more generally in \textit{multiset domains} (i.e., if multiple indistinguishable copies of the same good exist)  (\Cref{subsec:mvnns}) and we prove the universality property of mMVNNs in such multiset domains (\Cref{app:thm:Universality}).
    \item\label{itm:TrainOnDQContribution} We introduce a novel method for training 
    \emph{any} MIP-formalizable and gradient descent (GD)-compatible ML model (e.g., mMVNNs)
    on \textit{demand query} observations (\Cref{subsec:training_algorithm}).\footnote{Namely, this includes neural networks with any piecewise linear activation function.}
    Unlike prior work, our training method provably makes use of the \emph{complete} information provided by the demand query observations. 
    \item\label{itm:NextPriceContribution} We introduce an efficient method for determining the price vector that is most likely to clear the market based on the trained ML models (\Cref{sec:ML-powered Demand Query Generation}). For this, we derive a simple and intuitive price update rule that results from performing GD on an objective function which is minimized exactly at clearing prices (\Cref{thm:GD_on_W}).
    \item  Based on \Cref{itm:TrainOnDQContribution,itm:NextPriceContribution}, we propose a practical ML-powered clock auction (\Cref{sec:ML-powered Clock Phase}). 
    \item We experimentally show that compared to the CCA, our ML-powered clock auction can achieve substantially higher efficiency on the order of 9\% points.
    Furthermore, using linear prices, our ML-powered clock auction exhibits significantly higher clearing potential compared to the CCA (\Cref{sec:experiments}).
\end{enumerate}

\paragraph{GitHub} Our source code is publicly available on GitHub at \url{https://github.com/marketdesignresearch/ML-CCA}.

\subsection{Further Related Work}\label{subsec:Further Related Work}

In the field of \textit{automated mechanism design}, \citet{dutting2015payment,dutting2019optimal}, \citet{golowich2018deep} and \citet{narasimhan2016automated} used ML to learn new mechanisms from data, 
while \citet{cole_rougharden2014samplecomplexity, morgensternRoughgarden2015pseudodimension} and \citet{balcan2023generalization}
bounded the sample complexity of learning approximately optimal mechanisms. 
In contrast to this prior work, our design incorporates an ML algorithm into the mechanism itself, i.e., the ML algorithm is part of the mechanism. \citet{lahaielubin2019adaptive} suggest an adaptive price update rule that increases price expressivity as the rounds progress in order to improve efficiency and speed of convergence. 
Unlike that work, we aim to improve preference elicitation 
while still using linear prices.
Preference elicitation is a key market design challenge outside of CAs too. \citet{soumalias2023machine} introduce an ML-powered mechanism for course allocation that improves preference elicitation by asking students comparison queries.

\subsection{Practical Considerations and Incentives} \label{subsec:incentives}
Our ML-powered clock phase can be viewed as an alternative to the clock phase of the CCA. In a real-world application, many other considerations (beyond the price update rule) are also important. For example, the careful design of \textit{activity rules} is vital to induce truthful bidding in the clock phase of the CCA \citep{ausubel2017practical}.
The payment rule used in the supplementary round is also important, and it has been argued that the use of the VCG-nearest payment rule, while not strategy-proof, induces good incentives in practice \citep{cramton2013spectrumauctions}. 
Similar to the clock phase of the CCA, our ML-powered clock phase is not strategyproof. 
If our design were to be fielded in a real-world environment, we envision that one would combine it with carefully designed activity and payment rules in order to induce good incentives. 
Thus, we consider the incentive problem orthogonal to the price update problem and in the rest of the paper, we follow prior work \citep{brero2019fast, parkes2000iterative} and assume that bidders follow myopic best-response (truthful) bidding throughout all auction mechanisms tested.

\section{Preliminaries}\label{sec:Preliminaries}
\subsection{Formal Model for ICAs}\label{subsec:Formal Model for ICAs}
We consider \textit{multiset} CA domains with a set $N=\{1,\ldots,n\}$ of bidders and a set $M=\{1,\ldots,m\}$ of distinct items with corresponding \textit{capacities}, i.e., number of available copies, $c=(c_1,\ldots,c_m)\in\N^m$. 
We denote by $x\in \X=\{0,\ldots,c_1\}\times\ldots\times\{0,\ldots,c_m\}$ a bundle of items represented as a positive integer vector, where $x_{j}=k$ iff item $j \in M$ is contained $k$-times in $x$. 
The bidders' true preferences over bundles are represented by their (private) value functions $v_i: \X\to \R_{\geq0},\,\, i \in N$, i.e., $v_i(x)$ represents bidder $i$'s true value for bundle $x\in \X$. We collect the value functions $v_i$ in the vector $v=(v_i)_{i \in N}$. By $a=(a_1,\ldots,a_n) \in \X^n$ we denote an allocation of bundles to bidders, where $a_i$ is the bundle bidder $i$ obtains. We denote the set of \emph{feasible} allocations by $\F=\left\{a \in \X^n:\sum_{i \in N}a_{ij} \le c_j, \,\,\forall j \in M\right\}$. 
We assume that bidders have quasilinear utility functions $u_i$ of the form $u_i(a_i) = v_i(a_i) - \pi_i$ where $v_i$ can be highly non-linear and $\pi_i\in \R_{\geq0}$ denotes the bidder's payment. This implies that the (true) \emph{social welfare} $V(a)$ of an allocation $a$ is equal to the sum of all bidders' values $\sum_{i\in N} v_i(a_i)$.
We let $a^* \in \argmax_{a \in {\F}}V(a)$ denote a social-welfare maximizing, i.e., \textit{efficient}, allocation. The \emph{efficiency} of any allocation $a \in \F$ is determined as $V(a)/V(a^*)$.

An ICA \textit{mechanism} defines how the bidders interact with the auctioneer and how the allocation and payments are determined. 
In this paper, we consider ICAs that iteratively ask bidders \emph{linear demand queries}. In such a
query, the auctioneer presents a vector of item prices $p \in \R^m_{\ge 0}$ and each bidder $i$ responds with her utility-maximizing bundle, i.e.,
\begin{align}\label{eq:utility_maximizing_bundle}
x_i^*(p) \in \argmax_{x\in \X}\left\{v_i(x)-\iprod{p}{x}\right\}\, i\in N,
\end{align}
where $\iprod{\cdot}{\cdot}$ denotes the Euclidean scalar product in $\R^m$.

Even though our approach could conceptually incorporate any kind of (non-linear) price function $p: \X \to \R_{\ge 0}$, our concrete implementation will only use linear prices (i.e., prices over items). Linear prices are most established in practice since they are intuitive and simple for the bidders to understand (e.g., \citep{ausubel2006clock}).

For bidder $i \in N$, we denote a set of $K\in \N$ such elicited utility-maximizing bundles and price pairs as $R_i=\left\{\left(x_i^*(p^{r}),p^{r}\right)\right\}_{r=1}^{K}$. Let $R=(R_1,\ldots,R_n)$ be the tuple of elicited demand query data from all bidders.
The ICA's (inferred) optimal feasible allocation $a^*(R)\in \F$ and payments $\pi_i\coloneqq\pi_i(R)\in \R^n_+$ are computed based on the elicited reports $R$ \emph{only}. Concretely, $a^*_{R}\in \F$ is defined as 
\begin{align}\label{WDPFiniteReports}
a^*(R) \in \argmax_{\substack{\F\, \ni \, \left(x_i^*\left(p^{r_i}\right)\right)_{i=1}^{n}\\:\, r_i\in \{1,\ldots,K\} \,\forall i \in N}}\sum_{i\in N} \iprod{p^{r_i}}{x_i^*(p^{r_i})}.
\end{align}
In words, a bidder's response to a demand query provides a lower bound on that bidder's value for the bundle she requested. That lower bound is equal to the bundle's price in the round the bundle was requested. The ICA's optimal (inferred) feasible allocation $a^*(R)\in \F$ is the one that maximizes the corresponding lower bound on social welfare, based on all elicited demand query data $R$ from the bidders. 
As payment rule $\pi_i(R)$ one could use any reasonable choice (e.g., VCG payments, see \Appendixref{app:payment_methods}{Appendix~A}).
As the auctioneer can only ask a limited number of demand queries $|R_i| \leq \Qmax$ (e.g., $\Qmax=100$), an ICA needs a practically feasible and smart preference elicitation algorithm.

\subsection{The Combinatorial Clock Auction (CCA)}\label{subsec:prelims_CCA}
We consider the CCA \citep{ausubel2006clock} as the main benchmark auction. 
The CCA consists of two phases. The initial \textit{clock phase} proceeds in rounds. In each round, the auctioneer presents anonymous item prices $p \in \mathbb{R}_{\ge 0}^m$, and each bidder is asked to respond to a \textit{demand query}, declaring her utility-maximizing bundle at $p$. 
The clock phase of the CCA is parametrized by the \textit{reserve prices} employed in its first round, and the way prices are updated. 
An item $j$ is \emph{over-demanded} at prices $p$, if, for those prices, its total demand based on the bidders' responses to the demand query exceeds its capacity, i.e., $\sum_{i \in N} (x_i^*(p))_j > c_j$. The most common price update rule is to increase the price of all over-demanded items by a fixed percentage, which we set to $5\%$ for our experiments, as in many real-world applications (e.g., \cite{canadianCCA}).

The second phase of the CCA is \textit{the supplementary round}. In this phase, each bidder can submit a finite number of additional bids for bundles of items, which are called \textit{push bids}. 
Then, the final allocation is determined based on the combined set of all inferred bids of the clock phase, plus all submitted push bids of the supplementary round. 
This design aims to combine good price discovery in the clock phase with good expressiveness in the supplementary round. In simulations, the supplementary round is parametrized by the assumed bidder behaviour in this phase, i.e., which bundle-value pairs they choose to report. As in
\citep{brero2021workingpaper}, we consider the following heuristics when simulating bidder behaviour: 
\begin{itemize}[leftmargin=*,topsep=0pt,partopsep=0pt, parsep=0pt]
    \item \textbf{Clock Bids:} Corresponds to having no supplementary round. Thus, the final allocation is determined based only on the inferred bids of the clock phase (\Cref{WDPFiniteReports}).
    \item \textbf{Raised Clock Bids}: The bidders also provide their true value for all bundles they bid on during the clock phase. 
    \item \textbf{Profit Max:} Bidders provide their true value for all bundles that they bid on in the clock phase, and additionally submit their true value for the $\QPmax$ bundles earning them the highest utility at the prices of the final clock phase.
\end{itemize}

\section{Training on Demand Query Observations}\label{sec: Training on Demand Query Observation}
In this section, we first propose a new version of MVNNs that are applicable to multiset domains $\X$ and extend the universality proof of classical MVNNs. Finally, we present our demand-query training algorithm.

\begin{algorithm}[t!]
    \DontPrintSemicolon
    \SetKwInOut{Input}{Input}
    \SetKwInOut{Output}{Output}
    \Input{Demand query data $R_i=\left\{\left(x_i^*(p^{r}),p^{r}\right)\right\}_{r=1}^{K}$, Epochs $T\in \N$, Learning Rate $\gamma>0$.}
    $\theta_0 \gets $ init mMVNN \hspace{-0.00cm} \Comment*[f]{\color{CommentColor} \citet[S.3.2]{weissteiner2023bayesian}\!}\; 
    \For{$t = 0 \text{ to } T - 1$}{
        \For(\Comment*[f]{\color{CommentColor}Demand responses for prices}){$r = 1 \text{ to } K$}{
         Solve $\hat{x}^*_i(p^r) \in \argmax_{x\in \X}\mathcal{M}_i^{\theta_t}(x) - \iprod{p^r}{x}$\label{alg_dq:line5}\;
          \If(\Comment*[f]{\color{CommentColor}mMVNN is wrong}){$\hat{x}^*_i(p^r) \neq x_i^*(p^{r})$\label{alg_dq:line6}}
          {
          $L(\theta_t) \gets 
       (\mathcal{M}_i^{\theta_t}(\hat{x}^*_i(p^r)) - \iprod{p^r}{\hat{x}^*_i(p^r)}) - (\mathcal{M}_i^{\theta_t}(x_i^*(p^{r})) - \iprod{p^r}{x_i^*(p^{r})})$\label{alg_dq:line7}
       \Comment*[r]{\color{CommentColor}Add predicted utility difference to loss}
       $\theta_{t+1} \gets \theta_{t} - \gamma(\nabla_{\theta} L(\theta))_{\theta=\theta_t}$\Comment*[r]{\color{CommentColor}SGD step}
        }
        }
    }
    \Return{Trained parameters $\theta_T$ of the mMVNN $\mathcal{M}_i^{\theta_T}$}
    \caption{\textsc{TrainOnDQs}}
    \label{alg:train_on_dqs}
\end{algorithm}

\subsection{Multiset MVNNs}\label{subsec:mvnns}
MVNNs \cite{weissteiner2022monotone} are a recently introduced class of NNs specifically designed to represent \emph{monotone} \emph{combinatorial} valuations. We introduce an adapted version of MVNNs, which we call \emph{multiset MVNNs (mMVNNs)}. 
Compared to MVNNs,  mMVNNs have an added linear normalization layer $\D$ after the input layer. 
We add this normalization since the input (i.e., a bundle) $x \in \X$ is a positive integer vector instead of a binary vector as in the classic case of indivisible items with capacities $c_j=1$ for all $j \in M$. 
This normalization ensures that $\D x \in [0,1]$ and thus we can use the weight initialization scheme from \cite{weissteiner2023bayesian}. 
Unlike MVNNs, mMVNNs incorporate at a structural level the prior information that some items are identical and consequently significantly reduce the dimensionality of the input space. 
This improves the sample efficiency of mMVNNs, which is especially important in applications with a limited number of samples such as auctions. 
For more details on mMVNNs and their advantages, please see \Appendixref{sec:app_mMVNNs}{Appendix~B}.

\begin{definition}[Multiset MVNN]\label{def:MVNN}
		An mMVNN $\MVNNi{}:\X \to \R_{\ge0}$  for bidder $i\in N$ is defined as
            {\small
		\begin{equation}\label{eq:MVNN}
		\MVNNi{x}\coloneqq  W^{i,K_i}\varphi_{0,t^{i, K_i-1}}\left(\ldots\varphi_{0,t^{i, 1}}(W^{i,1}\left(\D x\right)+b^{i,1})\ldots\right)
		\end{equation}
             }
		\begin{itemize}[leftmargin=*,topsep=0pt,partopsep=0pt, parsep=0pt]
		\item $K_i+2\in\mathbb{N}$ is the number of layers ($K_i$ hidden layers),
		\item $\{\varphi_{0,t^{i, k}}{}\}_{k=1}^{K_i-1}$ are the MVNN-specific activation functions with cutoff $t^{i, k}>0$, called \emph{bounded ReLU (bReLU)}:
		\begin{align}\label{itm:MVNNactivation}
		\varphi_{0,t^{i, k}}(\cdot)\coloneqq\min(t^{i, k}, \max(0,\cdot))
		\end{align}
		\item $W^i\coloneqq (W^{i,k})_{k=1}^{K_i}$ with $W^{i,k}\ge0$ and $b^i\coloneqq (b^{i,k})_{k=1}^{K_i-1}$ with $b^{i,k}\le0$ are the \emph{non-negative} weights and \emph{non-positive} biases of dimensions $d^{i,k}\times d^{i,k-1}$ and $d^{i,k}$, whose parameters are stored in $\theta=(W^i,b^i)$.
        \item\label{itm:D} $\D\coloneqq \diag\left(\nicefrac{1}{c_1},\ldots,\nicefrac{1}{c_m}\right)$ is the linear normalization layer that ensures $\D x\in [0,1]$ and is not trainable.
		\end{itemize}
\end{definition}
In \Cref{app:thm:Universality}, we extend the proof from \citet{weissteiner2022monotone} and show that mMVNNs
 are \textit{also universal} in the set of monotone value functions defined on a multiset domain $\X$. For this, we first define the following properties:
\begin{enumerate}[align=left, leftmargin=*,topsep=2pt]\label{monotonicity_and_normalization}
\item[\textbf{(M)}]\label{itm:monotonicity}\textbf{Monotonicity}~(\emph{``more items weakly increase value''}):\\ For $a,b \in \X$: if $a\leq b$, i.e.~$\forall k \in M: a_i\leq b_i$, %
    it holds that $\hvi{a}\le \hvi{b}$,
    \item[\textbf{(N)}]\label{itm:normalization}\textbf{Normalization}~(\emph{''no value for empty bundle''}):\\ $\hvi{\emptyset}=\hvi{(0,\ldots,0)}\coloneqq 0$,
\end{enumerate}
These properties are common assumptions and are satisfied in many market domains.
We can now present the following universality result:
\begin{theorem}[Multiset Universality]\label{app:thm:Universality}\emph{Any} value function $\hvi{}:\X\to\Rp$ that satisfies  \textbf{\monoton{}} and \textbf{\normalized{}} can be represented exactly as an mMVNN~$\MVNNi{}$ from \Cref{def:MVNN}, i.e., for 
$\Vmon:=\{\hvi{}:\X \to \Rp|\, \text{satisfy \textbf{(M)} and \textbf{(N)}}\}$ it holds that
{%
\begin{align}
\Vmon=\left\{\MVNNi[(W^i,b^i)]{}: W^{i}\ge0, b^{i}\le0\} \right\}.
\end{align}}
\end{theorem}
\begin{proof}
Please, see \Appendixref{subsec:MVNN Universality}{Appendix~B.2} for the proof.
\end{proof}

Furthermore, we can formulate maximization over mMVNNs, i.e., $\max_{x\in\X}\MVNNi{x}-\iprod{p}{x}$, as a \emph{mixed integer linear program (MILP)} %
analogously to \citet{weissteiner2022monotone}, which will be key for our ML-powered clock phase. 

\subsection{Training Algorithm}\label{subsec:training_algorithm}

In \Cref{alg:train_on_dqs}, we describe how we train, for each bidder $i\in N$, a distinct mMVNN $\MVNNi{}$ on demand query data $R_i$.
\begin{figure}[t!]
    \begin{center}
    \resizebox{0.7\columnwidth}{!}{
\includegraphics[trim=10 9 40 55, clip]{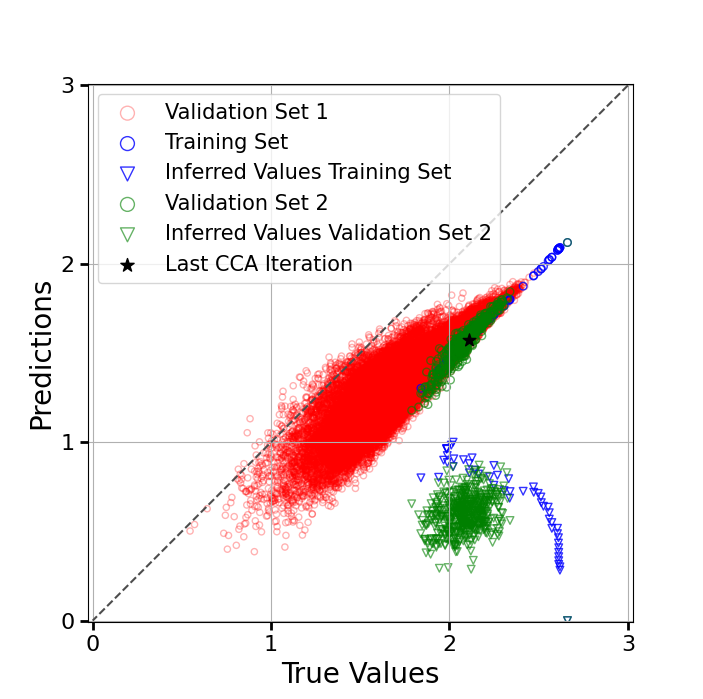}
{}
}
    \caption{Scaled prediction vs. true plot of a trained mMVNN via \Cref{alg:train_on_dqs} for the national bidder in the MRVM domain (see \Cref{sec:experiments}).}
    \label{fig:pred_vs_true_national_MRVM}
    \end{center}
    \vskip -0.65cm
\end{figure}
Our design choices regarding this training algorithm are motivated by the information that responses to demand queries provide.  According to myopic best response bidding, at each round $r$, bidder $i$ reports a utility-maximizing bundle $x_i^*\left(p^{r}\right) \in \X$ at current prices $p^r$. Formally, for all $x \in \X$: 
\begin{gather}
    v_i(x_i^*\left(p^{r}\right)) - \iprod{p^r}{x_i^*\left(p^{r}\right)} \ge v_i(x) - \iprod{p^r}{x}.  \label{ineq:dqs}
\end{gather}
Notice that for any epoch $t$ and round $r$, the loss $L(\theta_t)$ for that round calculated in \Cref{alg_dq:line5,alg_dq:line6,alg_dq:line7} is always non-negative, and can only be zero if the mMVNN $\MVNNi{}$ (instead of $v_i$) satisfies \Cref{ineq:dqs}. Thus, the loss for an epoch is zero iff the mMVNN $\MVNNi{}$ satisfies \Cref{ineq:dqs} for all rounds, and in that case the model has captured the full information provided by the demand query responses $R_i$ of that bidder.
Finally, note that \Cref{alg:train_on_dqs} can be applied to any 
MILP-formalizable ML model whose parameters can be efficiently updated via GD, such as MVNNs or ReLU-NNs.

In \Cref{fig:pred_vs_true_national_MRVM}, we present a prediction vs. true plot of an mMVNN, which we trained via \Cref{alg:train_on_dqs}. We present the training set of $50$ demand query data points $R_i$ in blue circles, where the prices $\{p^r\}_{r=1}^{50}$ are generated according to the same rule as in CCA.
Additionally, we mark the bundle $x^{\text{CCA}}\in \X$ from this last CCA iteration (i.e., the one resulting from $p^{50}$) with a black star. Moreover, we present two different validation sets on which we evaluate mMVNN configurations in our hyperparameter optimization (HPO): \emph{Validation set 1} (red circles), which are $50,000$ uniformly at random sampled bundles $x\in \X$, and \emph{validation set 2} (green circles), 
where we first sample $500$ price vectors $\{p^r\}_{r=1}^{500}$ where the price of each item is drawn uniformly at random from the range of $0$ to $3$ times the average maximum value of an agent of that type for a single item,
and then determine utility-maximizing bundles $x_i^*(p^{r})$ (w.r.t. $v_i$) at those prices (cp. \Cref{eq:utility_maximizing_bundle}). While validation set 1 measures generalization performance in a classic sense over the whole bundle space, validation set 2 focuses on utility-maximizing bundles. 
We additionally demonstrate the inferred values of the bundles of the training set and validation set 2 using triangles of the same colour, i.e., $\{\iprod{p^r
}{x_i^*(p^{r})}\}_{r=1}^{50/500}$. These triangles highlight the only cardinal information that our mMVNNs have access to during training and are a lower bound of the true value. In \Cref{fig:pred_vs_true_national_MRVM}, we see that our mMVNN is able to learn at the training points (blue circles) the true value functions almost perfectly up to a constant shift $\kappa$, i.e., $\mathcal{M}_i^{\theta_T}(x)\approx v_i(x)+\kappa$. This is true even though the corresponding inferred values (blue triangles) are very far off from the true values $\iprod{p^r}{x_i^*(p^{r})}\ll v_i(x_i^*(p^{r}))$. %
Moreover, the mMVNN generalizes well (up to the constant shift $\kappa$) on validation sets 1 and 2. Overall, this shows that \Cref{alg:train_on_dqs} indeed leads to mMVNNs $\mathcal{M}_i^{\theta_T}$ which are a good approximation of $v_i+\kappa$. 
Note that learning the true value function up to a constant shift suffices for our proposed demand query generation procedure presented in \Cref{sec:ML-powered Demand Query Generation}.

\section{ML-powered Demand Query Generation}\label{sec:ML-powered Demand Query Generation}
In this section, we show how we generate ML-powered demand queries and provide the theoretical foundation for our approach by extending a well-known connection between \emph{clearing prices}, \emph{efficiency} and a \emph{clearing objective function}.
First, we define indirect utility, revenue and clearing prices.

\begin{definition}[Indirect Utility and Revenue]\label{def:Indirect_Utility_and_Revenue}
    For linear prices $p \in \R_{\ge0}^m$, a bidder's indirect utility $U$ and the seller's indirect revenue $R$ are defined as
    \begin{align}
    &U(p,v_i)\coloneqq \max\limits_{x\in \X}\left\{v_i(x)-\iprod{p}{x}\right\} \text{ and } \label{eq:indirect_utility}\\
    &R(p)\coloneqq\max\limits_{a\in \F}\left\{\sum\limits_{i\in N}\iprod{p}{a_i}\right\}\stackrel{\hyperref[foot:Rlinear]{\footnotemark[1]}}{\hyperlink{proof:app:eq:indirect_revenue}{=}}\sum\limits_{j\in M}c_j p_j, \label{eq:indirect_revenue}
    \end{align}
    i.e., at prices $p$, \Cref{eq:indirect_utility,eq:indirect_revenue} are the maximum utility a bidder can achieve for all $x\in \X$ and the maximum revenue the seller can achieve among all feasible allocations.
\end{definition}
\begin{definition}[Clearing Prices]\label{def:linear_clearing_prices}
    Prices $p \in \R_{\ge 0}^m$ are \emph{clearing prices} if there exists an allocation $a(p)\in \mathcal{F}$ such that
    \begin{enumerate}
    \item \label{itm:clearing_bidder}  for each bidder $i$, the bundle $a_i(p)$ maximizes her utility, i.e., $v_i(a_i(p))-\iprod{p}{a_i(p)}=U(p,v_i),\forall i \in N$, and
    \item \label{itm:clearing_seller} the allocation $a(p)\in \F$ maximizes the sellers revenue, i.e., $\sum_{i\in N}\iprod{p}{a_i(p)}=R(p)$.\footnote{\label{foot:Rlinear}For linear prices, %
    this maximum is achieved by selling every item, i.e., $\forall j\in M: \sum_{i\in N}(a_{i})_j=c_j$ (see \Appendixref{subsec:app:GD_on_W}{Appendix~C.2}).}
    \end{enumerate}
\end{definition}
Next, we provide an important connection between \emph{clearing prices}, \emph{efficiency} and a \emph{clearing objective} $W$. \Cref{thm:app:connection_clearing_prices_efficiency_constrainedVersion} extends \citet[Theorem 3.1]{bikhchandani2002package}.

\begin{theorem}\label{thm:app:connection_clearing_prices_efficiency_constrainedVersion}
    Consider the notation from \Cref{def:Indirect_Utility_and_Revenue,def:linear_clearing_prices}
    and the objective function $W(p,v)\coloneqq R(p) + \sum_{i \in N}U(p,v_i)$.
    Then it holds that, if a linear clearing price vector exists, every price vector
    \begin{subequations}\label{eqs:ConstraintWmin}
    \begin{align}\label{eq:constrained_clearing_objective}
    p^{\prime} \in &\argmin_{\tilde{p}\in \R_{\ge 0}^m}  &&W(\tilde{p},v)\\
    & \text{such that}&&(x^*_i(\tilde{p}))_{i\in N}\in \F\label{eq:app:corollary_constraint}
    \end{align}
    \end{subequations}
    is a clearing price vector and the corresponding allocation $a(p^{\prime})\in \F$ is \emph{efficient}.\footnotemark
\end{theorem}
\begin{proof}
Please, see \Appendixref{subsec:app:proof_of_corollarly1connection_clearing_prices_efficiency}{Appendix~C.1} for the proof.
\end{proof}
\footnotetext{{\label{foot:PrecisFormulationConstraint}More precisely, constraint~\eqref{eq:app:corollary_constraint} should be reformulated as \[\exists \left(x^*_i(\tilde{p})\right)_{i\in N}\in \bigtimes_{i\in N}\X^*_i(\tilde{p}) : \left(x^*_i(\tilde{p})\right)_{i\in N} \in \F
    ,\] where $\X^*_i(\tilde{p}):=\argmax_{x\in \X}\left\{\hvi{x}-\iprod{\tilde{p}}{x}\right\}$, since in theory, $x^*_i(\tilde{p})$ does not always have to be unique.}}
\Cref{thm:app:connection_clearing_prices_efficiency_constrainedVersion} does not claim the existence of \emph{linear clearing prices (LCPs)} $p \in \R_{\ge 0}^m$. For general
value functions $v$, LCPs may not exist \citep{bikhchandani2002package}. 
However, in the case that LCPs do exist, \Cref{thm:app:connection_clearing_prices_efficiency_constrainedVersion} shows that \textit{all} minimizers of \eqref{eqs:ConstraintWmin} are LCPs and their corresponding allocation is efficient. This is at the core of our ML-powered demand query generation approach, which we discuss next.

The key idea to generate ML-powered demand queries is as follows: As an approximation for the true value function $v_i$, we use for each bidder a distinct mMVNN $\MVNNi{}:\X \to \R_{\ge0}$ that has been trained on the bidder's elicited demand query data $R_i$ (see \Cref{sec: Training on Demand Query Observation}). Motivated by \Cref{thm:app:connection_clearing_prices_efficiency_constrainedVersion}, we then try to find the demand query $p\in \R_{\ge 0}^m$ minimizing $W(p,\left(\MVNNi{}\right)_{i=1}^n)$ subject to the feasibility constraint~\eqref{eq:app:corollary_constraint}. This way, we find demand queries $p\in \R_{\ge 0}^m$ which, given the already observed demand responses $R$, have high clearing potential.
Note that unlike the CCA, this process does not result in monotone prices.\footnote{
We see no reason why non-monotone prices would introduce additional complexity for the bidders.   
With our approach, the prices quickly converge to the final prices, and then only change very little, as shown in \Appendixref{fig:GSVM_CE_and_LPs,fig:LSVM_CE_and_LPs,fig:SRVM_CE_and_LPs_logscale,fig:MRVM_CE_and_LPs}{Figures~5\crefrangeconjunction{}8} of \Appendixref{subsec:Details Results}{Appendix~D.7}.
For this reason, one could even argue that round-over-round optimizations for the bidders may be \textit{easier} in our auction: given that prices are close to each other round-over-round, the optimal bundle from the last round is still close to optimal (in terms of utility) in the next round.} 

\begin{remark}[Constraint~\eqref{eq:app:corollary_constraint}]\label{rem:ConstraintImportant}
An important economic insight is that minimizing $W(\cdot, (\mathcal{M}_i^\theta)_{i = 1}^{n})$ is optimal, when LCPs exist (also without constraint~\eqref{eq:app:corollary_constraint} as shown in \Appendixref{thm:connection_clearing_prices_efficiency}{Lemma~2} in \Appendixref{subsec:app:proof_of_corollarly1connection_clearing_prices_efficiency}{Appendix~C.1}). 
If however LCPs do not exist, it is favourable to minimize $W$ under the constraint of having no predicted over-demand for any items (see \Appendixref{subsec:app:Results_unconstrained_W}{Appendix~D.9} for an empirical comparison of minimizing $W$ with and without constraint \eqref{eq:app:corollary_constraint}). %
This is because in case the market does not clear, our ML-CCA (see \Cref{sec:ML-powered Clock Phase}), just like the CCA, will have to 
combine the clock bids of the agents to produce a \emph{feasible} allocation with the highest inferred social welfare according to \Cref{WDPFiniteReports}. 
See \Appendixref{subsec:app_constrained_W_minimization}{Appendix~D.6} for details.
\end{remark}

Note that \eqref{eqs:ConstraintWmin} is a hard, bi-level optimization problem.
We minimize \eqref{eqs:ConstraintWmin} via gradient descent (GD), since \Cref{thm:GD_on_W} gives us the gradient and convexity of $W(\cdot,\left(\MVNNi{}\right)_{i=1}^n)$.
\begin{theorem}\label{thm:GD_on_W}
Let $\left(\MVNNi{}\right)_{i=1}^n$ be a tuple of trained mMVNNs and let $\hat{x}^*_i(p)\in \argmax_{x\in \X}\left\{\MVNNi{x}-\iprod{p}{x}\right\}$ denote each bidder's predicted utility maximizing bundle w.r.t. $\MVNNi{}$. Then it holds that $p\mapsto W(p,\left(\MVNNi{}\right)_{i=1}^n)$ is \emph{convex}, \emph{Lipschitz-continuous} and \emph{a.e.
differentiable}. Moreover, 
\begin{equation}
    c-\sum_{i\in N}\hat{x}^*_i(p)\in \subgrad_p W(p,\left(\MVNNi{}\right)_{i=1}^n)
\end{equation}
is always a sub-gradient and a.e.\ a classical gradient.
\end{theorem}
\begin{proof}
In \Appendixref{subsec:app:GD_on_W}{Appendix~C.2} we provide the full proof. Concretely, \Appendixref{le:Wcontinous,le:Wconvex}{Lemmas~3 and 4} prove the Lipschitz-continuity and the convexity. 
In the following, we provide a sketch of how the (sub-)gradients are derived.
First, since $\X$ is finite, it is intuitive that $\hat{x}^*_i(p)$ is a piece-wise constant function and thus $\partial_p \hat{x}^*_i(p)\stackrel{\mathrm{a.e.}}{=}0$ (as intuitively argued by \citet{poganvcic2020differentiation} and proven by us in \Appendixref{le:WaeDiff}{Lemma~6}). Then we can compute the gradient a.e.\ as if $\hat{x}^*_i(p)$ was a constant:
\begin{align*}
&\nabla_p W\left(p,\left(\MVNNi{}\right)_{i=1}^n\right)=\nabla_p \left ( R(p) + \sum_{i \in N}U(p,\MVNNi{}) \right )\\
&=%
\nabla_p \left ( %
\sum_{j\in M}c_jp_j
+ \sum_{i \in N}\left(\MVNNi{}(\hat{x}^*_i(p))-\iprod{p}{\hat{x}^*_i(p)} \right )\right)%
\\%
&\stackrel{\mathclap{\mathrm{a.e.}}}{=}\,%
c
+ \sum_{i \in N}(0-\nabla_p\iprod{p}{\hat{x}^*_i(p)})%
=c-\sum_{i\in N}\hat{x}^*_i(p)%
.
\end{align*}
For a mathematically rigorous derivation of sub-gradients and a.e.\ differentiability see \Appendixref{le:WsubGradient,le:WaeDiff}{Lemmas~5 and 6}.
\end{proof}
With \Cref{thm:GD_on_W}, we obtain the following update rule of classical GD $
p^{\textnormal{new}}_j \stackrel{a.e.}{=} p_j-\gamma (c_j-\sum_{i\in N}(\hat{x}^*_{i}(p))_j),\, \forall j \in M$.
Interestingly, this equation has an intuitive economic interpretation. 
If the $j$\textsuperscript{th} item  is over/under-demanded based on the predicted utility-maximizing bundles $\hat{x}^*_{i}(p)$, then its new price $p^{\textnormal{new}}_j$ is increased/decreased by the learning rate times its over/under-demand. 
However, to enforce constraint~\eqref{eq:app:corollary_constraint} in GD, we asymmetrically increase the prices $1 + \mu \in \mathbb{R}_{\ge 0}$ times more in case of over-demand than we decrease them in case of under-demand. This leads to our final update rule (see \Appendixref{itm:Nextprice:Asymmetry}{Item~1} in \Appendixref{subsec:app_constrained_W_minimization}{Appendix~D.6} for more details):
\begin{subequations}\label{eqs:finalUpdateRule}
\begin{equation}\label{eq:simplified_update_rule_GD}
p^{\textnormal{new}}_j \stackrel{a.e.}{=} p_j-\tilde{\gamma_j} (c_j-\sum_{i\in N}(\hat{x}^*_{i}(p))_j),\, \forall j \in M, 
\end{equation}
\begin{equation}\label{eq:gamma_tilde}
\tilde{\gamma}_j\coloneqq\begin{cases}
\gamma\cdot(1+\mu)& ,c_j < \sum_{i\in N}(\hat{x}^*_{i}(p))_j\\
\gamma& ,\text{else}\\
\end{cases}
\end{equation}
\end{subequations}

To turn this soft constraint into a hard constraint, we increase this asymmetry via $\mu$ iteratively until we achieve feasibility and in the end we select the GD step with the lowest $W$ value \emph{within} those steps that were feasible%
. Based on the final update rule from \eqref{eqs:finalUpdateRule}, we propose \textsc{NextPrice} (\Appendixref{alg:Constrained_W_minimization}{Algorithm~3} in \Appendixref{subsec:app_constrained_W_minimization}{Appendix~D.6}), an algorithm that generates demand queries with high clearing potential, which additionally induce utility-maximizing bundles that are predicted to be feasible (see \Appendixref{subsec:app_constrained_W_minimization}{Appendix~D.6} for all details).

\begin{table*}[ht]
    \renewcommand\arraystretch{1.2}
    \setlength\tabcolsep{2pt}
	\robustify\bfseries
	\centering
	\begin{sc}
 \begin{adjustbox}{max width=\textwidth}
	\small
    \begin{tabular}{lcccccccccccccccc}
    \toprule
    &  \multicolumn{4}{c}{\textbf{GSVM}}  &\multicolumn{4}{c}{\textbf{LSVM}}  &   \multicolumn{4}{c}{\textbf{SRVM}}  & \multicolumn{4}{c}{\textbf{MRVM}}\\
        
        \cmidrule(l{2pt}r{2pt}){2-5}
        \cmidrule(l{2pt}r{2pt}){6-9}
        \cmidrule(l{2pt}r{2pt}){10-13}
        \cmidrule(l{2pt}r{2pt}){14-17}
     \textbf{Mechanism}& \textbf{E\textsubscript{clock}} & \textbf{E\textsubscript{raise}} & \textbf{E\textsubscript{profit}} & \textbf{\tableTextClear{}} & \textbf{E\textsubscript{clock}} & \textbf{E\textsubscript{raise}} & \textbf{E\textsubscript{profit}} & \textbf{\tableTextClear{}} & \textbf{E\textsubscript{clock}} & \textbf{E\textsubscript{raised}} & \textbf{E\textsubscript{profit}} & \textbf{\tableTextClear{}} & \textbf{E\textsubscript{clock}} & \textbf{E\textsubscript{raise}} & \textbf{E\textsubscript{profit}} & \textbf{\tableTextClear{}}\\
    \midrule
                              \multirow{1}{*}{\textbf{ML-CCA}}   & \ccell 98.23 & \ccell98.93 & \ccell100.00 & \ccell56 & \ccell91.64 &  \ccell96.39 & \ccell99.95  &  \ccell26 & \ccell99.59 & \ccell99.93& \ccell100.00&  \ccell13& \ccell93.04&  \ccell93.31 & \ccell93.68 &\ccell0\\
    \midrule
                      \textbf{CCA} & 90.40 & 93.59 & \ccell100.00 & 3 & 82.56 & 91.60 & 99.76  & 0 & \ccell99.63& 99.81& \ccell100.00& 8& 92.44 & 92.62 & 93.18 & \ccell0 \\
    \bottomrule
    \end{tabular}
\end{adjustbox}
    \end{sc}
    \vskip -0.1 in
    \caption{ML-CCA vs CCA. Shown are averages over a test set of $100$ synthetic CA instances of the following metrics: 
    efficiency in \% for clock bids (\textsc{\tablecaptionbf{E\textsubscript{clock}}}), raised clock bids (\textsc{\tablecaptionbf{E\textsubscript{raise}}}) and raised clock bids plus $100$ profit-max bids (\textsc{\tablecaptionbf{E\textsubscript{profit}}}) and percentage of instances where clearing prices were found (\textsc{\tablecaptionbf{\tableTextClear{}}}).
    Winners based on a paired t-test with $\alpha=5\%$ are marked in grey.
    }
\label{tab:efficiency_loss_mlca}
    \vskip -0.3cm
\end{table*}

\section{ML-powered Combinatorial Clock Auction}\label{sec:ML-powered Clock Phase}
In this section, we describe our \emph{ML-powered combinatorial clock auction (ML-CCA)}, which is based on our proposed new training algorithm from \Cref{sec: Training on Demand Query Observation} as well as our new demand query generation procedure from \Cref{sec:ML-powered Demand Query Generation}.

\begin{algorithm}[t!]
        \DontPrintSemicolon
        \SetKwInOut{parameters}{Parameters}
        \parameters{$\Qinit,\Qmax$ with $\Qinit\le \Qmax$ and $F_{\text{init}}$}
    $R \gets (\{\})_{i=1}^N$ \;
    \For(\Comment*[f]{\color{CommentColor}{Draw $\Qinit$ initial prices}}){$r=1,...,\Qinit$}{ \label{alg_line:qinit_start}
    $p^r \gets F_{\text{init}}(R)$ \; 
        \ForEach(\Comment*[f]{\color{CommentColor}{Initial demand query responses}}){$i \in N$}
        {
        $R_i \gets R_i\cup\{(x^*_{i}(p^r),p^{r})\}$ \label{alg_line:qinit_p_response}
        }
        }
    \For(\Comment*[f]{\color{CommentColor}{ML-powered rounds}}){$r=\Qinit+1,...,\Qmax$}{
        \ForEach{$i \in N$}
        {{$\MVNNi{} \gets $}  \textsc{TrainOnDQs}$(R_i)$ \label{alg_line:train_mvnns}\Comment*[r]{\color{CommentColor}{\Cref{alg:train_on_dqs}}}
        } 
        $p^{r} \gets$ \textsc{NextPrice$(\left(\MVNNi{}\right)_{i=1}^n)$}\Comment*[r]{\color{CommentColor}{
        \Appendixref{alg:Constrained_W_minimization}{Algorithm~3} 
        }}\label{alg_line:next_pv} 
        \ForEach(\Comment*[f]{{\color{CommentColor}Demand query responses for $p^r$}}){$i \in N$}{
         $R_i \gets R_i\cup\{(x^*_{i}(p^r),p^{r})\}$ \label{alg_line:dq_responses}
        }
        \If(\Comment*[f]{\color{CommentColor}Market-clearing}){$\sum\limits_{i=1}^n (x^*_{i}(p^k))_j= c_j\, \forall j\in M$}{
        Set final allocation $a^*(R) \gets (x^*_i(p^r))_{i=1}^n$\;
        Calculate payments $\pi(R)\gets(\pi_i(R))_{i=1}^n$ \;
        \Return{$a^*(R)$ and $\pi(R)$} \label{alg_line:return_clearing_allocation}
        }}
    \ForEach{$i \in N$}
    {$R_i \gets R_i \cup B_i$ \label{alg_line:push_bids}  
    \Comment*[r]{\color{CommentColor}{Optional Push bids }}
    }
    Calculate final allocation $a^*(R)$ as in \Cref{WDPFiniteReports}\; \label{alg_line:final_allocation}
    Calculate payments $\pi(R)$ \Comment*[r]{{\color{CommentColor}E.g., VCG (\Appendixref{app:payment_methods}{Appendix~A})}} \label{alg_line:final_payments}
    \Return{$a^*(R)$ and $\pi(R)$} \label{alg_line:return_final_result}
    \caption{\small \textsc{ML-CCA}($\Qinit,\Qmax, F_{\text{init}}$)}
    \label{ML-CCA}
\end{algorithm}

We present ML-CCA in \Cref{ML-CCA}. 
In Lines \ref{alg_line:qinit_start} to \ref{alg_line:qinit_p_response}, we draw the first $\Qinit$ price vectors using some initial demand query method $F_{\text{init}}$ and receive the bidders' demand responses to those price vectors. Concretely, in \Cref{sec:experiments}, we report results using the same price update rule as the CCA for $F_{\text{init}}$. 
In each of the next up to $\Qmax - \Qinit$ ML-powered rounds, we first train, for each bidder, an mMVNN on her demand responses using \Cref{alg:train_on_dqs} (\Cref{alg_line:train_mvnns}). Next, in \Cref{alg_line:next_pv}, we call \textsc{NextPrice} to generate the next demand query $p$ based on the agents' trained mMVNNs (see \Cref{sec:ML-powered Demand Query Generation}).
If, based on the agents' responses to the demand query (\Cref{alg_line:dq_responses}), our algorithm has found market-clearing prices, then the corresponding allocation is efficient and is returned, along with payments $\pi(R)$ according to the deployed payment rule  (\Cref{alg_line:return_clearing_allocation}).
If, by the end of the ML-powered rounds, the market has not cleared, we 
optionally allow bidders to submit push bids, analogously to the supplementary round of the CCA (\Cref{alg_line:push_bids}) and 
calculate the optimal allocation $a^*(R)$ and the payments $\pi(R)$ (\Cref{alg_line:final_allocation,alg_line:final_payments}). Note that 
ML-CCA can be combined with various possible payment rules $\pi(R)$, such as VCG or VCG-nearest.

\section{Experiments}\label{sec:experiments}
In this section, we experimentally evaluate the performance of our proposed ML-CCA from \Cref{ML-CCA}. 

\subsection{Experiment Setup.}
To generate synthetic CA instances, we use the GSVM, LSVM, SRVM, and MRVM domains from the spectrum auction test suite (SATS) \cite{weiss2017sats}  (see \Appendixref{subsec:appendix_SATS_domains}{Appendix~D.1} for details).
We compare our ML-CCA with the original CCA. For both mechanisms, we allow a maximum of $100$ clock rounds per instance, i.e., we set $\Qmax = 100$.
For CCA, we set the price increment to $5$\% as in \citep{canadianCCA} and optimized the initial reserve prices to maximize its efficiency. 
For ML-CCA, we create $\Qinit$ price vectors to generate the initial demand query data using the same price update rule as the CCA, with the price increment adjusted to accommodate for the reduced number of rounds following this price update rule.
In GSVM, LSVM and SRVM we set $\Qinit = 20$ for ML-CCA, while in MRVM we set $\Qinit = 50$. 
After each clock round, we report efficiency according to the clock bids up to that round, as well as efficiency if those clock bids were raised (see \Cref{subsec:prelims_CCA}). Finally, we report the efficiency if the last clock round was supplemented with $\QPmax = 100$ bids using the profit max heuristic.
Note that this is a very unrealistic and cognitively expensive bidding heuristic in practice, as it requires the agents to both discover their top $100$ most profitable bundles as well as report their exact values for them, and thus only adds theoretical value to gauge the difficulty of each domain.

\subsection{Hyperparameter Optimization (HPO).}
We optimized the hyperparameters (HPs) of the mMVNNs for each bidder type of each domain. Specifically, for each bidder type we trained an mMVNN on the demand responses of a bidder of that type on $50$ CCA clock rounds
and selected the HPs that resulted in the highest $R^{2}$ on validation set $2$ as described in \Cref{subsec:mvnns}. For more details 
please see \Appendixref{subsec:appendix_hpo}{Appendix~D.3}.

\subsection{Results.}
All efficiency results are presented in \Cref{tab:efficiency_loss_mlca}, while in \Cref{fig:efficiency_per_clock_round} we present the efficiency after each clock round, as well as the efficiency if those clock bids were enhanced with the clock bids raised heuristic (for $95$\% CIs and  $p$-values see \Appendixref{subsec:Details Results}{Appendix~D.7}).

In GSVM, ML-CCA's clock phase exhibits over $7.8$\% points higher efficiency compared to the CCA, while if we add the clock bids raised heuristic to both mechanisms, ML-CCA still exhibits over $5.3$\% points higher efficiency. At the same time, ML-CCA is able to find clearing prices in $56$\% of the instances, as opposed to only $3$\% for the CCA.

The results for LSVM are qualitatively very similar; ML-CCA's clock phase increases efficiency compared to the CCA by over $9$\% points, while clearing the market in $26$\% of the cases as opposed to $0$\%.
If we add the clock bids raised heuristic to both mechanisms, ML-CCA still increases efficiency by over $4.7$\% points.

The SRVM domain, as suggested by the existence of only $3$ unique goods, is quite easy to solve. Thus, both mechanisms can achieve almost $100$\% efficiency after their clock phase.
For the clock bids raised heuristic our method reduces the efficiency loss by a factor of more than two (from $0.19$\% to $0.07$\%), and additionally, one can see from \Cref{fig:efficiency_per_clock_round} that our method reaches over $99$\% in less than $30$ rounds.  

In MRVM, ML-CCA again achieves statistically significantly better results for all 3 bidding heuristics. 
Notably, the CCA needs both the clock bids raised heuristic \emph{and} $38$ profit max bids to reach the same efficiency as our ML-CCA clock phase, i.e., it needs up to $138$ additional \emph{value} queries per bidder (see \Appendixref{subsec:Details Results}{Appendix~D.7}).
In MRVM, LCPs never exist, thus neither ML-CCA nor the CCA can ever clear the market.

To put our efficiency improvements in perspective, 
in the GSVM, LSVM and MRVM domains, ML-CCA's clock phase achieves higher efficiency than the CCA enhanced with the clock bids raised heuristic, i.e., the CCA, even if it uses up to an additional $100$ value queries per bidder, cannot match the efficiency of our ML-powered clock phase.
In \Cref{fig:efficiency_per_clock_round}, we see that our ML-CCA can (almost) reach the efficiency numbers of \Cref{tab:efficiency_loss_mlca} in a significantly reduced number of clock rounds compared to the CCA, while if we attempt to ``speed up'' the CCA, then its efficiency can substantially drop, see \Appendixref{subsec:app:Results_reduced_Qmax}{Appendix~D.8}. 
In particular, in GSVM and LSVM, using $50$ clock rounds, our ML-CCA can achieve higher efficiency than the CCA can in $100$ clock rounds.

\begin{figure}[t!]
    \begin{center}
\includegraphics[width=1\columnwidth,trim=0 0 0 0, clip]{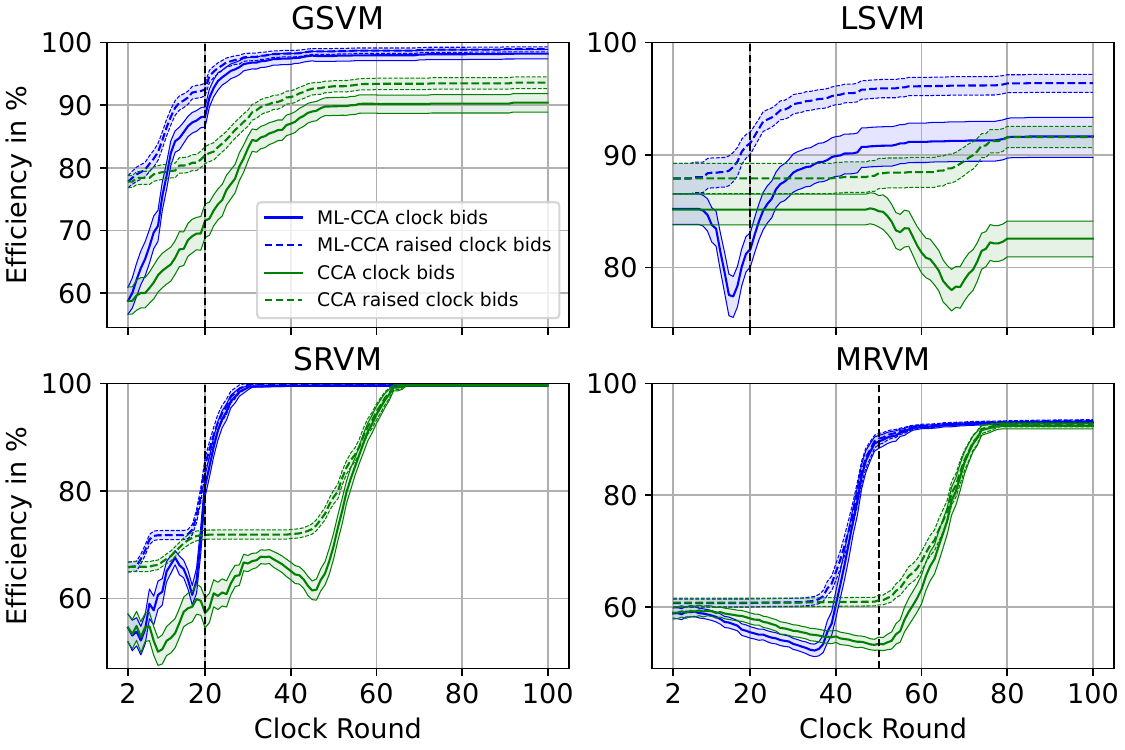}
    \caption{Efficiency path plots in SATS for ML-CCA and CCA both after clock bids (solid lines) and raised clock bids (dashed lines). Averaged over 100 runs including a 95\% CI. The dashed black vertical line indicates the value of $\Qinit$.}
    \label{fig:efficiency_per_clock_round}
    \end{center}
    \vskip -0.35cm
\end{figure}

\subsection{Computational Efficiency.}
For our choice of hyperparameters, the computation time when using $8$ CPU cores\footnote{Note that \Cref{alg:train_on_dqs} is not GPU-implementable, as it requires solving a MIP in each iteration, for every demand response by an agent} (see \Appendixref{subsec:app:compute_infrastructure}{Appendix~D.2} for details on our compute infrastructure) for a single round of the ML-CCA averages under $45$ minutes for all domains tested. Notably, for three out of four domains, it averages less than $10$ minutes (see \Appendixref{tab:app:compute_time}{Table~7} in \Appendixref{subsec:Details Results}{Appendix~D.7}). 
The overwhelming majority of this time is devoted to training the mMVNNs of all bidders using our \Cref{alg:train_on_dqs} 
and generating the next DQ using our \Appendixref{alg:Constrained_W_minimization}{Algorithm~3} detailed in \Cref{sec:ML-powered Demand Query Generation}. 
It is important to note that both of these algorithms can always be parallelized by up to the number of bidders $N$, further reducing the time required. 
In our implementation, while we parallelized the training of the $N$ mMVNNs, we did not do so for the generation of the next DQ.
In spectrum auctions, typically no more than $2$ rounds are conducted per day. 
For the results presented in this paper, we ran the full auction for $400$ instances. 
Given the estimated welfare improvements of over $25$ million USD per auction, attributed to ML-CCA's efficiency gains, we consider ML-CCA's computational and time requirements to be negligible.

\section{Conclusion}
We have proposed a novel method for training MVNNs
to approximate the bidders' value functions based on demand query observations.
Additionally, we have framed the task of determining the price vector with the highest clearing potential as minimization of an objective function that we prove is convex, Lipschitz-continuous, a.e.
differentiable, and whose gradient for linear prices has an intuitive economic interpretation: change the price of every good proportionally to its predicted under/over-demand at the current prices. 
The resulting mechanism (ML-CCA) from combining these two components exhibits significantly higher clearing potential than the CCA and can
increase efficiency by up to $9$\% points while at the same time converging in a much smaller number of rounds. 
Thus, we have designed the first \textit{practical} ML-powered auction that employs the same interaction paradigm as the CCA, i.e., demand queries instead of cognitively too complex value queries, yet is able to significantly outperform the CCA in terms of both efficiency and clearing potential in realistic domains.

\section*{Acknowledgments}
This paper is part of a project that has received funding from the European Research Council (ERC)
under the European Union’s Horizon 2020 research and innovation program (Grant agreement No. 805542).

\if\showAppendix 1
\clearpage
\appendix
\section*{Appendix}

\counterwithin{definition}{section}
\counterwithin{corollary}{section}%
\counterwithin{problem}{section}
\counterwithin{example}{section}
\counterwithin{remark}{section}
\counterwithin{fact}{section}

\section{Payment and Activity Rules}\label{app:payment_methods}

In this section, we reprint the VCG and VCG-nearest payment rules, as well as give an overview of activity rules for the CCA, and argue why the most prominent choices are also applicable to our ML-CCA.

\subsection{VCG Payments from Demand Query Data}
\begin{definition}{\textsc{(VCG Payments from Demand Query Data)}}\label{def:vcg_payments}
Let $R=(R_1,\ldots,R_n)$ denote an elicited set of demand query data from each bidder and let $R_{-i}\coloneqq(R_1,\ldots,R_{i-1},R_{i+1},\ldots,R_n)$. We then calculate the VCG payments $\pi^{\text{\tiny VCG}}(R)=(\pi^{\text{\tiny VCG}}_1(R)\ldots,\pi^{\text{\tiny VCG}}_n(R)) \in \R_{\ge0}^n$ as follows:
{\small
\begin{align}\label{VCGPayments}
&\pi^{\text{\tiny VCG}}_i(R) \coloneqq\\
&\sum_{j \in N \setminus \{i\}} \iprod{\left(p^*_{R_{-i}}\right)_j}{\left(a^*_{R_{-i}}\right)_j} - \sum_{j \in N \setminus \{i\}}\iprod{\left(p^*_{R}\right)_j}{\left(a^*_{R}\right)_j}.
\end{align}
}%
where $a^*_{R_{-i}}$ is the allocation that maximizes the inferred social welfare (SW) when excluding bidder $i$, i.e.,
\begin{align}
&a^*_{R_{-i}}\in \argmax_{\substack{\F\, \ni \, \left(x_j^*\left(p^{(r_j)}\right)\right)_{j\in N\setminus\{i\}}\\:\, r_j\in \{1,\ldots,K\} \,\forall j \in N\setminus\{i\}}}\sum_{j\in N\setminus\{i\}} \iprod{p^{(r_j)}}{x_j^*(p^{(r_j)})},
\end{align}
and $p^*_{R_{-i}}=\left((p^*_{R_{-i}})_1,\ldots,(p^*_{R_{-i}})_n\right) \in \R^{mn}_{\ge 0}$ denote the corresponding price vectors that lead to $a^*_{R_{-i}}$, and $a^*_R$ is a inferred-social-welfare-maximizing allocation (see \Cref{WDPFiniteReports}) with corresponding prices $p^*_{R}=\left((p^*_{R})_1,\ldots,(p^*_{R})_n\right) \in \R^{mn}_{\ge 0}$.
\end{definition}

Thus, when using VCG, bidder $i$'s utility is:
{\small
\begin{align*}
u_i &= v_i((a_R^*)_i)-\pi^{\text{\tiny VCG}}_i(R)\\
&=  v_i((a_R^*)_i) + \hspace{-0.3cm}\sum_{j \in N \setminus \{i\}}\hspace{-0.25cm} \iprod{\left(p^*_{R}\right)_j}{\left(a^*_{R}\right)_j} - \hspace{-0.1cm} \hspace{-0.1cm}\sum_{j \in N \setminus \{i\}}\hspace{-0.25cm} \iprod{\left(p^*_{R_{-i}}\right)_j}{\left(a^*_{R_{-i}}\right)_j}.
\end{align*}
}
\begin{remark}
Note that \Cref{def:vcg_payments} defines VCG payments for a set of elicited \emph{demand query} data $R$ using the bidders' inferred values from this data set. Specifically, those would be the VCG payments after the CCA's clock phase for example (i.e., if there was no supplementary round). However, one can analogously define VCG-payments for any given set of reported bundle-value pairs (see \citet[Definition B.1.]{weissteiner2023bayesian}). For example, in the case of additional value bids, such as supplementary round push bids, one would use \citet[Definition B.1.]{weissteiner2023bayesian} and set $\hat{v}_i(x^r)=\max \{\iprod{x^r}{p^r}, b_i^s(x^r) \}$, where $b^s$ is the bidder's supplementary round bid for that bundle (or zero, if she did not bid on it in the supplementary round), i.e., bidder $i$'s bid for any bundle is the maximum of her largest inferred value for that bundle based on the clock round bids and the supplementary round bids. 
\end{remark}

\subsection{VCG-Nearest Payments}
\label{subsec:app_vickrey_nearest}
To define the VCG-nearest payments, we must first introduce the core:
\begin{definition}{\textsc{(The Core)}}\label{def:core}
An outcome $(a,\pi)\in \F\times \Rpz^n$ (i.e., a tuple of a feasible allocation $a$ and payments $\pi$) is in the core if it satisfies the following two properties:
\begin{enumerate}
\item The outcome is \emph{individual rational}, i.e, $u_i=v_i(a_i)-\pi_i\ge0$ for all $i\in N$
\item The core constraints
\begin{equation}
    \forall  \; L \subseteq N \; \sum_{i \in N \setminus L} \pi_i(R) 
    \ge \max_{a^{\prime} \in \F} \sum_{i \in L} v_i(a^{\prime}_i) - \sum_{i \in L} v_i(a_i)  
\end{equation}
where $v_i(a_i)$ is bidder $i$'s value for bundle $a_i$ and $\mathcal{F}$ is the set of feasible allocations.
\end{enumerate}
\end{definition}
In words, a payment vector $\pi$ (together with a feasible allocation $a$) is in the core if no coalition of bidders $L\subset N$ is willing to pay more for the items than the mechanism is charging the winners. Note that by replacing the true values $v_i(a_i)$ with the bidders' (possibly untruthful) bids $b_i(a_i)$ in \Cref{def:core} one can equivalently define the \emph{revealed core}. 

Now, we can define 

\begin{definition}{\textsc{(Minimum Revenue Core)}}
Among all payment vectors in the (revealed) core, the (revealed) minimum revenue core is the set of payment vectors with smallest $L_1$-norm, i.e., which minimize the sum of the payments of all bidders.
\end{definition}

We can now define VCG-nearest payments:
\begin{definition}{\textsc{(VCG-Nearest Payments)}} \label{def:vcg_nearest_payments}
Given an allocation $a_R$ for bidder reports $R$, the VCG-nearest payments $\pi^{\text{\tiny VCG-nearest}}(R)$ are defined as the vector of payments in the (revealed) minimum revenue core that minimizes the $L_2$-norm to the VCG payment vector $\pi^{\text{\tiny VCG}}(R)$.
\end{definition}

\subsection{On the Importance of Activity Rules to Align Incentives} \label{subsec:app_acticity_rules}
In the CCA, activity rules serve multiple purposes. First, they can help speed up the auction process. Second, they reduce ''bid-sniping'' opportunities, i.e., bidders concealing their true intentions until the very last rounds of the auction.\footnote{The notion of ``bid-sniping'' first originated in eBay auctions with predetermined ending times, where the high-value bidder can sometimes reduce her payments by submitting her bid at the very last moment.}
Third, they can limit surprise bids in the supplementary round of the CCA and significantly reduce a bidder's ability to drive up her opponents payments by overbidding on bundles that she can no longer win \citep{ausubel2017practical}.
There are two types of activity rules that are implemented in a CCA: 
\begin{enumerate}
    \item \emph{Clock phase activity rules}, that limit the bundles that an agent can bid on during the clock phase, based on her bids in previous clock rounds.
    \item  \emph{Supplementary round activity rules}, that restrict the amount that an agent can bid on for various sets of items during the supplementary round.
\end{enumerate}
Most of the activity rules that were traditionally used for the clock phase of the CCA were based on either revealed-preference considerations or some \textit{points-based system}, where the main idea is to assign points to each item prior to the auction, and only allow bidders to submit monotonically non-increasing in points bids, i.e., as the rounds progress and the prices increase, the bidders cannot submit bids for larger sets of items. 
Both of these approaches, as well as hybrid combinations thereof, were shown to actually further interfere with truthful bidding in some cases \citep{ausubel2014, ausubel2020}.

However, \citet{ausubel2019iterative} showed that basing the clock phase activity rule not on the above but instead entirely upon the \emph{generalized axiom of revealed preference (GARP)} can dynamically approximate VCG payoffs and thus improve the bidding incentives of the CCA.
GARP imposes revealed-preference constraints (see \Cref{def:revealed_preference_constr}) to the bidder’s demand responses, i.e., the GARP activity rule requires the bidder to exhibit rational behaviour in her demand choices.
Importantly, the GARP activity rule \emph{does not} require a monotonic price trajectory. Thus, it can also be applied in our ML-powered clock phase, allowing the clock phase of our ML-CCA to enjoy the same improvement in bidding incentives.

For the supplementary round, the CCA's most prominent activity rules are again based on a combination of the same points-based system and revealed-preference ideas. For this, we need to define the following constraint:

\begin{definition}{\textsc{(Revealed-preference constraint)}} \label{def:revealed_preference_constr}
The revealed-preference constraint for bundle $x\in X$ with respect to clock round $r$ is 
\begin{equation}
b_i(x) \le b_i(x^r) + \iprod{p^r}{x - x^r},
\end{equation}
where $b_i(x)\in \Rpz$ is bidder $i$'s bid for bundle $x\in X$ in the supplementary round, $x^r\in \X$ is the bundle demanded by the agent at clock round $r$, $b_i(x^r)\in \Rpz$ is the final bid for bundle $x^r\in \X$ and $p^r\in \Rpz^m$ is the linear price vector of clock round $r$.
\end{definition}
Intuitively, the revealed-preference constraint states that a bidder is not allowed to claim a high value for bundle $x\in \X$ relative to bundle $x^r\in \X$, given that she claimed to prefer bundle $x^r\in \X$ at clock round $r$ (see Inequality~\eqref{ineq:dqs}).
The difference between the three most prominent supplementary round activity rules is with respect to \textit{which clock rounds} the revealed-preference constraint should be satisfied. Specifically: 
\begin{enumerate}
    \item \emph{Final Cap:} A bid for bundle $x\in \X$ should satisfy the \emph{revealed-preference constraint (\Cref{def:revealed_preference_constr})} with respect to the \emph{final} clock round's price $p^{\Qmax}\in \Rpz$ and bundle $x^{\Qmax}\in \X$.
    \item \emph{Relative Cap:} A bid for bundle $x\in \X$ should satisfy the \emph{revealed-preference constraint (\Cref{def:revealed_preference_constr})} with respect to the last clock round for which the bidder was eligible for that bundle $x\in \X$, based on the points-based system.
    \item \emph{Intermediate Cap:} A bid for bundle $x \in \X$ should satisfy the \emph{revealed-preference constraint (\Cref{def:revealed_preference_constr})} with respect to all eligibility-reducing rounds, starting from the last clock round for which the bidder was eligible for $x\in \X$ based on the point system.
\end{enumerate}

\citet{ausubel2017practical} showed that combining the \emph{Final Cap} and \emph{Relative Cap} activity rules leads to the largest amount of reduction in bid-sniping opportunities for the UK 4G auction, as measured by the theoretical bid amount that each bidder would need to increase her bid by in the supplementary round in order to protect her final clock round bundle. 
Finally, note that the \emph{Final-} and \emph{Intermediate Cap} activity rules can also be applied to our ML-CCA.\footnote{With the modification for the \emph{Relative Cap} rule that the revealed-preference constraint should hold for the $\Qinit$ rounds that follow the same price update rule as the CCA, and then the ML-powered clock rounds should be treated as corresponding to the same amount of points, since the prices in these rounds on aggregate stay very close to the prices of the last $\Qinit$ round, as shown in \Cref{fig:GSVM_CE_and_LPs,fig:LSVM_CE_and_LPs,fig:SRVM_CE_and_LPs_logscale,fig:MRVM_CE_and_LPs}.}

To conclude, we observe that for both its clock phase and the supplementary round, our ML-CCA (with the same $F_{\text{init}}$ method as the CCA), is compatible with the most prominent activity rules for the corresponding phases of the CCA, while it is also obviously compatible with the most prominent payment rule, VCG-nearest prices (\Cref{def:vcg_nearest_payments}). 
This, combined with the fact that the ML-CCA has the same interaction paradigm for the bidders as the CCA, is a very strong indication that our ML-CCA can reduce the opportunities for the bidders to misreport to a similar extent as the classical CCA.

\section{Multiset MVNNs} \label{sec:app_mMVNNs}
\subsection{Advantages of mMVNNs over MVNNs} 
In \citet{weissteiner2020deep} and  \citet{weissteiner2022fourier,weissteiner2022monotone,weissteiner2023bayesian} items with a capacity $c_k$, were treated as $c_k$ distinct items, without exploiting the prior knowledge that these $c_k$ items are indistinguishable to the bidders. Sufficiently large classical MVNNs that were trained on large enough training sets would at some point learn that these items are indistinguishable, but this 
prior information was not incorporated at an architectural level.
For multiset MVNNs (mMVNNs) this prior knowledge is hard-coded directly into the architecture (see \Cref{def:MVNN}). This additional prior knowledge is particularly beneficial for small training data sets in terms of generalization.
At the same time, it can significantly reduce the dimensionality of the input space, as \Cref{ex:BetterGernalizationWithMultisets} illustrates.
\begin{example}\label{ex:BetterGernalizationWithMultisets}
    If we have 30 different items each of capacity 2 (i.e., $\X=\{0,1,2\}^{30}$), then then the bundle $x=(1,\dots,1)\in\X$ containing one of each 30 items would have $2^{30}>1$ billion different representations as sets corresponding to binary vectors in $\tilde{\X}=\{0,1\}^{60}$ (where the items are treated as 60 items of capacity 1). All these $2^{30}$ representations in $\tilde{\X}$ could be mapped to different values by a classical MVNN~$\tilde{\MVNNiw} :\tilde{\X}\to\Rpz$, while we actually have the hard prior knowledge that they should all have exactly the same value due to indistinguishable items. And all these $2^{30}$ representations in $\tilde{\X}$ correspond to the same multiset $x=(1,\dots,1)\in\X$, which gets assigned to exactly one value $\MVNNiw(x)$ by an mMVNN~$\MVNNiw :\X\to\Rpz$.
\end{example}

Furthermore, the solution times of MILPs are likely to benefit from the mMVNNs' reduced variable count compared to traditional MVNNs. In the case of classical MVNNs, it was necessary to introduce one binary variable for each indistinguishable duplicate of an item, resulting in $\tilde{m}=\sum_{k=1}^m c_k$. 
In contrast, mMVNNs require only $m$ integer variables. For an experimental comparison of mMVNNs to MVNNs, see \Cref{subsec:app:results_mvnns_multiset_domains}.

\subsection{Universality of mMVNNs}\label{subsec:MVNN Universality}
Note that \citet{weissteiner2022monotone} have proven universality of MVNNs only for the case of binary input vectors corresponding to classical sets (i.e., $c_k=1 \ \forall k\in M$). Here, we prove in \Cref{appproof:thm:Universality} universality of mMVNNs for arbitrary capacities $c\in\N^m$ corresponding to \href{https://en.wikipedia.org/w/index.php?title=Multiset&oldid=1145794995#}{multiset} domains such as our $\X$.

First, we recall the following definition:
\begin{definition}\label{def:Vmon}
The set $\Vmon$ of all monotonic\footnote{Within this paper, by \enquote{monotonic}, we always refer to weakly monotonically increasing \monoton{}, i.e.,  monotonically non-decreasing. In multiset-notation this reads as: $\forall a,b \in \X : \left(a\subseteq b\implies \hvi{a}\le \hvi{b}\right)$.} and normalized functions from $\X$ to $\R_{\geq0}$ is defined as
\begin{align}\label{eq:Vmon}
    \Vmon:=\{\hvi{}:\X \to \Rp|\, \text{satisfy \textbf{\normalized{}} and \textbf{\monoton{}}}\}.
\end{align}  
\end{definition}

The following \namecref{app_lem:normalized and monoton} says that every mMVNN is monotonic \monoton{} and and normalized \normalized{} in the sense of \Cref{def:Vmon}.
\begin{lemma}\label{app_lem:normalized and monoton}
    Let $\MVNNi{}:\X\to \Rp$ be an mMVNN from \Cref{def:MVNN}. Then it holds that  $\MVNNi[(W^{i},b^{i})]{}\in\Vmon$ for all $W^{i}\ge0$ and $b^{i}\le0$.
\end{lemma}
\begin{proof}\ 
    The proof of this lemma is perfectly analogous to the proof of \cite[Lemma~1]{weissteiner2022monotone}.
\end{proof}
Now we are ready to present a constructive proof for \Cref{app:thm:Universality}.
\paragraph{Proof of \Cref{app:thm:Universality}}
\begin{proof}\label{appproof:thm:Universality}\
This proof follows a similar strategy as the proof of \cite[Theorem~1]{weissteiner2022monotone}.
\begin{enumerate}
    \item {\small $\Vmon \supseteq \left\{\MVNNi[(W^i,b^i)]{}: W^{i,k}\ge0, b^{i,k}\le0\,\, \forall k \in \{1,\ldots,K_i\} \right\}$}\\
    
    This direction follows immediately from \Cref{app_lem:normalized and monoton}.\\
    
    \item {\small$\Vmon \subseteq \left\{\MVNNi[(W^i,b^i)]{}: W^{i,k}\ge0, b^{i,k}\le0\,\, \forall k \in \{1,\ldots,K_i\} \right\}$}\\
    
    Let $(\hvi{x})_{x\in \X} \in \Vmon$.
    For the reverse direction, we give a constructive proof, i.e., we construct an mMVNN $\MVNNi{}$ with $\theta=(W^i_{\hvi{}},b^i_{\hvi{}})_{i\in\{1,\dots,4\}}$ such that $\MVNNi{x}=\hvi{x}$ for all $x\in \X$.
    
    Let $(w_j)_{j=1}^{|\X|}$ denote the values corresponding to $(\hvi{x})_{x\in \X}$ of all $|\X|=\prod_{j=1}^m (c_j+1)$ possible bundles $x\in\X$ sorted by value in increasing order, i.e, let
    $x_1=(0,\ldots,0)$ with 
    \begin{align}\label{eq:w_1}
    w_1:=\hvi{x_1}=0,
    \end{align}
    let $x_{|\X|}=c=(c_1,\ldots,c_m)$ with \begin{align}\label{eq:w_2m}
    w_{|\X|}:=\hvi{x_{|\X|}},
    \end{align}
    and $x_j,x_l \in \X\setminus\{x_1,x_{|\X|}\}$\ for $1< l \le j \le |\X|-1$ with
    \begin{align}\label{eq:w_remaining}
    w_j:=\hvi{x_j}\, \leq\, w_l:=\hvi{x_l}.
    \end{align}
    In the following, for $x_l,x_j\in \X$, we use the notation  $x_l\leq x_j$ iff $x_{l,k}\leq x_{j,k}\ \ \forall k\in M$. Thus, we write $x_l\not\leq x_j$ iff $\exists k\in M:x_{l,k}> x_{j,k}$, i.e., iff $x_l\leq x_j$ is not the case.\footnote{Note, that in the multidimensional case $m>1$, the two symbols $\not\leq$ and $>$ have a different meaning, i.e., $x_l> x_j \iff \left( x_l\not\leq x_j \text{ and } x_l\geq x_j\right)$. Further note that $<$, $\leq$ and $\not\leq$ exactly correspond to $\subset$, $\subseteq$ and $\not\subseteq$ respectively in multiset notation.} Furthermore, we denote by $ \left<\cdot,\cdot\right>$ the Euclidean scalar product on $\R^m$. Then, for all $x\in\X$:
    {\small
    \begin{align}
    &\hvi{\xi}=\hspace{-0.05cm}\sum_{l=1}^{|\X|-1}\hspace{-0.05cm} \left(w_{l+1}-w_{l}\right)\1{\forall j\in\{1,\dots,l\}\,:\,\xi\not\leq x_{j}}\\
    &=\sum_{l=1}^{|\X|-1}\hspace{-0.05cm} \left(w_{l+1}-w_{l}\right)		  \phiu_{0,1}{}\hspace{-0.05cm}\left(\sum_{j=1}^{l}	\phiu_{0,1}{}\left( \sum_{k=1}^{m}\phiu_{0,1}(\xi_k-x_{j,k})\right)-(l-1)\right)\label{eq:last_term},	      
    \end{align}
    }%
    where the second equality follows since
    {\small
    \begin{align}
        \xi\not\leq x_{j} & \iff \exists k \in  M : \xi_k>x_{j,k}\\
        & \iff \exists k \in M: \phiu_{0,1}(\xi_k-x_{j,k})=1\\
        &\iff \phiu_{0,1}{}\left( \sum_{k=1}^{m}\phiu_{0,1}(\xi_k-x_{j,k})\right)=1,
    \end{align}
    }%
    which implies that
    \begin{align}
        &\left(\forall j\in\{1,\dots,l\}: \xi\not\leq x_{j}\right)\notag\\ 
        &\iff \sum_{j=1}^{l}\phiu_{0,1}{}\left( \sum_{k=1}^{m}\phiu_{0,1}(\xi_k-x_{j,k})\right)=l,
    \end{align}
    and
    \begin{align}
        &\1{\forall j\in\{1,\dots,l\}\,:\,x\not\leq x_{j}} \notag\\
        &=\phiu_{0,1}{}\left(\sum_{j=1}^{l}\phiu_{0,1}{}\left( \sum_{k=1}^{m}\phiu_{0,1}(\xi_k-x_{j,k})\right)-(l-1)\right).
    \end{align}

    To match the structure of the mMVNN architecture defined in \Cref{def:MVNN}, we can write $\hvi{\xi}=\hvi{{\D}^{-1}\D\xi}$ and plug in ${\D}^{-1}\D\xi$ instead of $\xi$ in \Cref{eq:last_term}.\footnote{The diagonal matrix $\D\coloneqq \diag\left(\nicefrac{1}{c_1},\ldots,\nicefrac{1}{c_m}\right)$ is invertible with $\D^{-1}= \diag\left(c_1,\ldots,c_m\right)$, since all capacities $c_i$ are strictly larger than 0.}

\Cref{eq:last_term} can be rewritten as the following mMVNN
{\small
\begin{align}&\MVNNi{x}=\notag\\
&W^{i,4}_{\hvi{}} \phiu_{0,1}{}\left(W^{i,3}_{\hvi{}} \phiu_{0,1}{}\left(W^{i,2}_{\hvi{}} \phiu_{0,1}{}\left(W^{i,1}_{\hvi{}} \D\xi+ b^{i,1}_{\hvi{}}\right) + b^{i,2}_{\hvi{}}\right) b^{i,3}_{\hvi{}}\right)\notag\\
&+b^{i,4}_{\hvi{}}
\end{align}
}

in the matrix-notation of \Cref{def:MVNN} with weight matrices and bias-vectors be given as:
\begin{align}\tiny
    W_{\hvi{}}^{i,1}
    &:=
        \begin{bmatrix}
    \D^{-1}  \\
    \vdots     \\
 \D^{-1}
    \end{bmatrix}
    \in\R^{m(|\X|-1)\times m},
    \\
    b_{\hvi{}}^{i,1}
    &:=
    -\begin{bmatrix}
    x_1\\
    x_2\\
    \vdots\\
    x_{|\X|-1}
    \end{bmatrix}\in\R^{m(|\X|-1)},
    \\
    \begin{split}
    W_{\hvi{}}^{i,2}
    &:=
    \begin{bmatrix}
    \overbrace{1,\dots,1}^{m}       &  \overbrace{0,\dots,0}^{m} &  \ldots &\overbrace{0,\dots,0}^{m}      \\
    0,\dots,0 &1,\dots,1  & \ddots  &\vdots     \\
    \vdots   & \ddots  &  \ddots &0,\dots,0     \\
    0,\dots,0      &   \ldots     & 0,\dots,0 & 1,\dots,1 
    \end{bmatrix}\\
    &\in\R^{(|\X|-1)\times m(|\X|-1)},\label{eq:W2}
    \end{split}
    \\
    b_{\hvi{}}^{i,2}
    &:=0
    \in\R^{|\X|-1},
    \\
    W^{i,3}_{\hvi{}}
    &:=
    \begin{bmatrix}
    1       &  0 &  \ldots &0      \\
    \vdots  &\ddots  & \ddots  &\vdots     \\
    \vdots   &  &  \ddots &0     \\
    1       &   \ldots     & \ldots & 1
    \end{bmatrix}
    \in\R^{(|\X|-1)\times (|\X|-1)},
    \\
    b_{\hvi{}}^{i,3}
    &:=
    \begin{bmatrix}  0\\-1\\ \vdots \\ -(|\X| - 2)\end{bmatrix}
    \in\R^{|\X|-1},
    \\
    W_{\hvi{}}^{i,4}
    &:=
    \begin{bmatrix}
    w_2-w_1\\
    w_3-w_2     \\
    \vdots\\
    w_{|\X|}-w_{|\X|-1} 
    \end{bmatrix}^{\top}
    \in\R^{1\times(|\X|-1)},
    \\
    b_{\hvi{}}^{i,4}
    &:=0
    \in\R,
\end{align}
where $W_{\hvi{}}^{i,2}=\left(\1{(j-1)m<k\leq jm}\right)_{1\leq j\leq|\X|-1,1\leq k\leq m(|\X|-1)}$ is an alternative notation to describe \Cref{eq:W2}.
 Thus, $\MVNNi{x}$ is an mMVNN from \Cref{def:MVNN} with 6 layers in total (i.e., 1 input layer, 1 linear normalization layer~$\D$, 3 non-linear hidden layers and 1 output layer) and respective dimensions $[m,m,m(|\X|-1),|\X|-1,|\X|-1,1]$.
\end{enumerate}
\end{proof}

\begin{remark}[Normalization Layer]\label{rem:NormLayerOptional}
    Note that our \hyperref[appproof:thm:Universality]{proof} of \Cref{app:thm:Universality} would also work without the linear normalization layer~$\D$. The normalization layer has the advantage that we can use similar hyperparameters as for classical MVNNs. E.g., we can use the same values of the parameters in the initialization scheme as provided in \citet[Section~3.2 and Appendix~E]{weissteiner2023bayesian}.
\end{remark}

\begin{remark}[Number of Hidden Layers]
    Note that for binary input vectors (corresponding to classical sets, i.e., $c=(1,\dots,1)$), 2 non-linear hidden layers were sufficient \citep[Proof of Theorem~1]{weissteiner2022monotone}, while for integer-valued input vectors (corresponding to multisets with capacities~$c\in\N^m$), we used 3 non-linear hidden layers for our \hyperref[appproof:thm:Universality]{proof} of \Cref{app:thm:Universality}. It is an interesting open question if 2 non-linear hidden layers would already be sufficient also for our multiset setting with capacities $c\in\N^m$.\footnote{Our \hyperref[appproof:thm:Universality]{proof} of \Cref{app:thm:Universality} works with 1  input layer, 1 linear normalization layer~$\D$ (which is optional, see \Cref{rem:NormLayerOptional}), 3 non-linear hidden layers and 1 output layer (i.e., 5 or 6 layers in total). Whereas the \citep[Proof of Theorem~1]{weissteiner2022monotone} only required 1 input layer, 2 non-linear hidden layers and 1 output layer (i.e., 4 layers in total).}
\end{remark}

\begin{remark}[Number of Hidden Neurons]
    The dimension $[m,m,m(|\X|-1),|\X|-1,|\X|-1,1]$ of the mMVNN used in the \hyperref[appproof:thm:Universality]{proof} of \Cref{app:thm:Universality} is just an upper bound. This proof does not imply that such large networks are actually needed. From a theoretical perspective it is interesting that this upper bound is finite, while for NNs on continuous domains no finite upper bound for perfect approximation exits. In practice much smaller networks are usually sufficient. All the results reported in this paper were achieved by much smaller networks with at most 30 neurons per layer (see \Cref{table_HPO_ranges} in \Cref{subsec:appendix_hpo}).
\end{remark}

\section{ML-powered Demand Query Generation:Theoretical Results}\label{sec:app:ML-powered Demand Query Generation:Theoretical Results}
In this section, we prove \Cref{thm:app:connection_clearing_prices_efficiency_constrainedVersion,thm:GD_on_W}.

\subsection{Proof of \Cref{thm:app:connection_clearing_prices_efficiency_constrainedVersion}}\label{subsec:app:proof_of_corollarly1connection_clearing_prices_efficiency}
In this subsection, we first present and prove in \Cref{thm:connection_clearing_prices_efficiency} an unconstrained version of \Cref{thm:app:connection_clearing_prices_efficiency_constrainedVersion}, which we later use to prove our main statement from \Cref{thm:app:connection_clearing_prices_efficiency_constrainedVersion}.

\Cref{thm:connection_clearing_prices_efficiency} extends \citep[Theorem 3.1.]{bikhchandani2002package}. Concretely, we additionally show that if linear clearing prices exist, \emph{every} minimizer of $W$ is a clearing price  (while \citet[Theorem 3.1.]{bikhchandani2002package} only showed that every clearing price $p$ minimizes $W$ not specifying if there could be other minimizers of $W$ which are not clearing prices).

\begin{lemma}\label{thm:connection_clearing_prices_efficiency}
    Consider the notation from \Cref{def:Indirect_Utility_and_Revenue,def:linear_clearing_prices}
    and the objective function $W(p,v)\coloneqq R(p) + \sum_{i \in N}U(p,v_i)$. Then it holds that, if a linear clearing price exists, every price vector
    \begin{align}\label{eq:clearing_objective}
    p' \in \argmin_{\tilde{p}\in \R_{\ge 0}^m}W(\tilde{p},v)
    \end{align} 
    is a clearing price and the corresponding allocation $a(p')\in \F$ is \emph{efficient}.
\end{lemma}

\begin{proof}[Proof of \Cref{thm:connection_clearing_prices_efficiency}]\label{proof:connectionClearingPricesEfficiency}
We first show in \Cref{itm1:app:proof_thm:connection_clearing_prices_efficiency} that for clearing prices $p \in \Rpz^m$, the corresponding allocation $a(p)\in \F$ is \emph{efficient}. Next, in \Cref{itm3:app:proof_thm:connection_clearing_prices_efficiency} we show that every linear clearing price $p\in \R_{\ge 0}^m$ minimizes $W$, i.e., formally $p \in \argmin_{\tilde{p}\in \R_{\ge 0}^m}W(\tilde{p},v)$. Finally, in \Cref{itm4:app:proof_thm:connection_clearing_prices_efficiency} we show that if linear clearing prices exist, \emph{any} minimizer of $W$ is in fact a linear clearing price.
\begin{enumerate}
\item \label{itm1:app:proof_thm:connection_clearing_prices_efficiency} Let $(p,a(p)) \in \R^m_{\ge 0}\times \F$ be clearing prices and the corresponding supported allocation (see \Cref{def:linear_clearing_prices}). Furthermore, let $\tilde{a}\in \F$ be any other feasible allocation. Then it holds that:
\begin{align}
&\sum_{i=1}^{n}v_i(a_i(p))=\\
&\sum_{i=1}^{n} [v_i(a_i(p))-\iprod{p}{a_i(p)}] + \sum_{i=1}^{n}\iprod{p}{a_i(p)}= \\
& \left(\sum_{i=1}^{n} U(p,v_i)\right) + R(p)\ge \\
& \sum_{i=1}^{n} [v_i(\tilde{a}_i)-\iprod{p}{\tilde{a}_i}] + \sum_{i=1}^{n}\iprod{p}{\tilde{a}_i} = \\ 
&\sum_{i=1}^{n}v_i(\tilde{a}_i),
\end{align}
where the first inequality follows since $(p,a(p))$ are clearing prices and fulfill \Cref{itm:clearing_bidder,itm:clearing_seller} in \Cref{def:linear_clearing_prices} and because $\tilde{a}\in \F$ is another feasible allocation. This shows that a supported allocation $a(p)$ is efficient.

\item \label{itm3:app:proof_thm:connection_clearing_prices_efficiency} Let $(p,a(p))\in \R^m_{\ge 0}\times \F$ be clearing prices and the corresponding supported allocation (see \Cref{def:linear_clearing_prices}). Furthermore,
let $\tilde{p}\in \R_{\ge 0}^m$ be any other linear prices. Then it holds that:
\begin{align}
&W(p,v)=\label{eq:app:proof:thm1:itm3:1}\\
&\max_{a\in \F}\left\{\sum_{i=1}^{n}\iprod{p}{a_i}\right\} + \sum_{i=1}^{n}\max_{x\in \X} \left\{v_i(x)-\iprod{p}{x}\right\}=\label{eq:app:proof:thm1:itm3:2}\\
&\sum_{i=1}^{n}\iprod{p}{a_i(p)} + \sum_{i=1}^{n}[v_i(a_i(p))-\iprod{p}{a_i(p)}]=\label{eq:app:proof:thm1:itm3:3}\\
&\sum_{i=1}^{n}v_i(a_i(p))=\\%
&\sum_{i=1}^{n}\iprod{\tilde{p}}{a_i(p)} + \sum_{i=1}^{n}[v_i(a_i(p))-\iprod{\tilde{p}}{a_i(p)}]\le\label{eq:app:proof:thm1:itm3:5}\\
&\max_{a\in \F}\left\{\sum_{i=1}^{n}\iprod{\tilde{p}}{a_i}\right\} + \sum_{i=1}^{n}\max_{x\in \X} \left\{v_i(x)-\iprod{\tilde{p}}{x}\right\}=\label{eq:app:proof:thm1:itm3:6}\\
&W(\tilde{p},v),
\end{align}
where \Cref{eq:app:proof:thm1:itm3:2} follows by definition of $W$ and \Cref{eq:app:proof:thm1:itm3:3} follows since $(p,a(p))$ are clearing prices and the corresponding supported allocation, respectively
Thus, we get that
\begin{equation}
W(p,v)\le W(\tilde{p},v),
\end{equation}
for all linear prices $\tilde{p}\in \R_{\ge 0}^m$, which concludes the proof.

\item\label{itm4:app:proof_thm:connection_clearing_prices_efficiency} Let $p^{\prime} \in \argmin_{\tilde{p}\in \R_{\ge 0}^m}W(\tilde{p},v)$ be a minimizer of $W$. Moreover,  let $p \in  \R_{\ge 0}^m$ denote a linear clearing price and
let $a(p)\in \F$ be the corresponding supported allocation (see \Cref{def:linear_clearing_prices}). Furthermore, let $x_i^*(p^{\prime}) \in \argmax_{x\in \X}\left\{v_i(x)-\iprod{p^{\prime}}{x}\right\}\, \forall i\in N$. We know from \Cref{itm3:app:proof_thm:connection_clearing_prices_efficiency} that $W(p^{\prime},v)=W(p,v)$. Furthermore, we get that
\begin{subequations}\label{eqs:WpPrimeVsWp}
\begin{align}
     &\sum_{i=1}^{n}\hvi{x_i^*(p^{\prime})}
    -\sum_{i=1}^{n} \iprod{p^{\prime}}{x_i^*(p^{\prime})}+\iprod{p^{\prime}}{c}=\\
    &W(p^{\prime},v)=W(p,v)=\\
    &\sum_{i=1}^{n}\hvi{a_i(p)}-\sum_{i=1}^{n} \iprod{p}{a_i(p)}+\iprod{p}{c}\stackrel{ \eqref{eq:app:proof:thm1:itm3:3}}{=}\\
    &\sum_{i=1}^{n}\hvi{a_i(p)}=\\
    &\sum_{i=1}^{n}\hvi{a_i(p)}-\sum_{i=1}^{n} \iprod{p^{\prime}}{a_i(p)}+\iprod{p^{\prime}}{c},
\end{align}
\end{subequations}
where the last equation follows since $a(p)$ is a clearing allocation supported by $p$ and thus $\sum_{i=1}^{n} \iprod{p^{\prime}}{a_i(p)}=\iprod{p^{\prime}}{c}$. Next, we can subtract $\iprod{p^{\prime}}{c}$ from both sides of \Cref{eqs:WpPrimeVsWp} to obtain
\begin{align}
    &\sum_{i=1}^{n}\hvi{x_i^*(p^{\prime})}
    - \iprod{p^{\prime}}{x_i^*(p^{\prime})}
    =  \nonumber \\ 
    &
    \sum_{i=1}^{n}\hvi{a_i(p)}- \iprod{p^{\prime}}{a_i(p)}. \label{eq:proof:cor1:itm4:sum_equal}
\end{align}
Moreover, %
from the optimality
of $x_i^*(p^{\prime})$ it follows that 
\begin{align}\label{eq:proof:cor1:itm4:util_maximizer}
v_{i}(x_i^*(p^{\prime}))-
     \iprod{p^{\prime}}{x_i^*(p^{\prime})}
    \geq
    v_{i}(a_i{(p)}
)- \iprod{p^{\prime}}{a_i{}(p)},
\end{align}
holds for every $i \in \{1,\dots,n\}$.

Using now  \Cref{eq:proof:cor1:itm4:sum_equal,eq:proof:cor1:itm4:util_maximizer}, it follows that for every $i \in \{1,\dots,n\}$
\begin{align}\label{eq:summands_equal}
&U(p^{\prime},v_i)=v_{i}(x_i^*(p^{\prime}))-
     \iprod{p^{\prime}}{x_i^*(p^{\prime})}=\\
    &v_{i}(a_{i}(p))- \iprod{p^{\prime}}{a_{i}(p)}.
\end{align}
Taken all together, since $\sum_{i \in N}\iprod{p^{\prime}}{a_i(p)}=\sum_{j \in M} \iprod{p^{\prime}_j}{c_j}=R(p^{\prime})$ and $U(p^{\prime},v_i)=v_{i}(a_{i}(p))- \iprod{p^{\prime}}{a_{i}(p)}$ for all $i\in N$ %
, we can conclude that $(p^{\prime},a(p))\in \R^m_{\ge 0}\times \F $  by \Cref{def:linear_clearing_prices} is a linear clearing price with corresponding supported allocation $a(p)\in \F$. This finalizes the proof of \Cref{thm:connection_clearing_prices_efficiency}.\end{enumerate}
\end{proof}

Now, we are ready to prove  \Cref{thm:app:connection_clearing_prices_efficiency_constrainedVersion}, which follows almost immediately from \Cref{thm:connection_clearing_prices_efficiency}.

\begin{proof}[Proof of \Cref{thm:app:connection_clearing_prices_efficiency_constrainedVersion}]
Let $p^{\prime}$ be a solution to the constrained minimization problem defined by \Cref{eqs:ConstraintWmin}. Moreover,  let $p \in  \R_{\ge 0}^m$ denote a linear clearing price vector and
let $a(p)\in \F$ be the corresponding supported allocation (see \Cref{def:linear_clearing_prices}).

In \Cref{itm3:app:proof_thm:connection_clearing_prices_efficiency} in the proof of \Cref{thm:connection_clearing_prices_efficiency}, we have seen that $p$ minimizes $W$ without any constraints.
The clearing price vector $p$ obviously satisfies constraint~\eqref{eq:app:corollary_constraint}, thus $p$ is also a solution to the constrained optimization problem~\eqref{eqs:ConstraintWmin}.
Therefore, in the case that linear clearing prices exist the minimal objective value of the constrained optimization problem from \eqref{eqs:ConstraintWmin} is equal to the minimal objective value of the unconstrained minimization problem~\eqref{eq:clearing_objective}.
From this we can conclude that in the case that linear clearing prices exist, every solution $p^{\prime}$ of \eqref{eqs:ConstraintWmin} is also a solution of the optimization problem~\eqref{eq:clearing_objective}.
Finally, \Cref{thm:connection_clearing_prices_efficiency} tells us that every solution $p^{\prime}$ of the optimization problem~\eqref{eq:clearing_objective} is a clearing price and $a(p^{\prime})$ is efficient.
\end{proof}
In \Cref{thm:app:connection_clearing_prices_efficiency_constrainedVersion} we proved that the constraint minimizer of $W$ has the same economical favourable properties as the unconstrained minimizer of $W$ if linear clearing prices exist. Thus, we do not lose anything in cases where linear clearing prices exist, while due to the constraint~\eqref{eq:app:corollary_constraint} we have the advantage of receiving feasible allocations in the case that no linear clearing price exist.

\subsection{Details: Proof of \Cref{thm:GD_on_W}}\label{subsec:app:GD_on_W}
Before we start with the proof, we will quickly provide a proof for \Cref{eq:indirect_revenue}, since this equation (that we have not explicitly proven in the main paper) is essential in the proof of \Cref{thm:GD_on_W}.
\hypertarget{proof:app:eq:indirect_revenue}{}\begin{proof}[Proof of \Cref{eq:indirect_revenue}]%
\begin{align}
    R(p)&\coloneqq\max\limits_{a\in \F}\left\{\sum\limits_{i\in N}\iprod{p}{a_{i}}\right\}\\
    &=\max\limits_{a\in \F}\left\{\sum\limits_{i\in N}\sum_{j \in M}p_j a_{i,j}\right\}\\
    &=\max\limits_{a\in \F}\left\{\sum_{j\in M}\sum_{i \in N}p_j a_{i,j}\right\}\\
    &=\max\limits_{a\in \F}\left\{\sum_{j\in M}p_j\sum_{i \in N} a_{i,j}\right\}\\
    &=\max\left\{\sum_{j\in M}p_j\sum_{i \in N} a_{i,j} : \sum_{i \in N} a_{i,j}\leq c_j\ \forall j\in M\right\}\\
    &=\sum\limits_{j\in M}p_j c_j
    =\sum\limits_{j\in M}c_j p_j
    =\iprod{c}{p}.
\end{align}
    
\end{proof}

In the first lemma we prove that $p\mapsto W(p,\left(\MVNNi{}\right)_{i=1}^n)$ is Lipschitz-continuous. While neither $p\mapsto x^*_i(p)$, nor $p\mapsto \MVNNi{x^*_i(p)}$ nor $p\mapsto\iprod{p}{x^*_i(p)}$ are continuous, (surprisingly) $p\mapsto \MVNNi{x^*_i(p)} + \iprod{p}{x^*_i(p)}$ is continuous.
\begin{lemma}[Continuity]\label{le:Wcontinous}
The map $p\mapsto W\left(p,\left(\MVNNi{}\right)_{i=1}^n\right)$ from $\Rpz^m$ to $\Rpz$ is Lipschitz-continuous with Lipschitz-constant $(n+1)\twonorm[c]$.\footnote{Note that Lipschitz-continuity implies (uniform) continuity. The Lipschitz-constant $(n+1)\twonorm[c]$, given in the proof, only depends on the capacities $c$ of $\X$ and on the number of bidders $n$. The proof also works for any other (possibly non-monotonic) value function $v_i:\X\to\R$ instead of $\MVNNi{}$.}
\end{lemma}
\begin{proof}
    Since $ W\left(p,\left(\MVNNi{}\right)_{i=1}^n\right)\coloneqq R(p) + \sum_{i \in N}U\left(p,\MVNNi{}\right)$ is the sum of $1+n$ functions, we first quickly show that $R$ is Lipschitz-continuous in $p$ and afterwards we will show that also $U$ is Lipschitz-continuous in $p$.
    
    From \Cref{eq:indirect_revenue} it follows that $R(p)=\sum\limits_{j\in M}c_j p_j$ is linear in $p$ and thus Lipschitz-continuous with Lipschitz-constant $\twonorm[c]$.

    In the remainder of the proof we are going to show that $U\left(p,\MVNNi{}\right)\coloneqq \max\limits_{x\in \X}\left\{\MVNNi{}(x)-\iprod{p}{x}\right\}$ is Lipschitz-continuous in $p$.
    Let $p,\tilde{p}\in\Rpz^m$ be two price vectors, then we have
    \begin{align}
U\left(\tilde{p},\MVNNi{}\right)
&\geq\MVNNi{\hat{x}^*_i(p)}+\iprod{\tilde{p}}{\hat{x}^*_i(p)}\\
&=\MVNNi{\hat{x}^*_i(p)}+\iprod{\tilde{p}+p-p}{\hat{x}^*_i(p)}\\
&=\MVNNi{\hat{x}^*_i(p)}+\iprod{p}{\hat{x}^*(p)}+\iprod{\tilde{p}-p}{\hat{x}^*_i(p)}\\
&=U\left(p,\MVNNi{}\right) +\iprod{\tilde{p}-p}{\hat{x}^*_i(p)}\\
&\geq U\left(p,\MVNNi{}\right) -\twonorm[\tilde{p}-p]\twonorm[\hat{x}^*_i(p)]\\
&\geq U\left(p,\MVNNi{}\right) -\twonorm[\tilde{p}-p]\twonorm[c],
    \end{align}
    and thus
    \begin{align}\label{eq:cont_lemma1}
    U\left(\tilde{p},\MVNNi{}\right) -U\left(p,\MVNNi{}\right)\geq -\twonorm[\tilde{p}-p]\twonorm[c].
 \end{align}
    By exchanging the roles of $p$ and $\tilde{p}$, we also obtain $
     U\left(p,\MVNNi{}\right) \geq U\left(\tilde{p},\MVNNi{}\right) -\twonorm[\tilde{p}-p]\twonorm[c]$
 and thus,
 \begin{align}\label{eq:cont_lemma2}
 U\left(\tilde{p},\MVNNi{}\right)-U\left(p,\MVNNi{}\right) \leq \twonorm[\tilde{p}-p]\twonorm[c].
 \end{align}
 Finally, \Cref{eq:cont_lemma1,eq:cont_lemma2} together imply
 \begin{align}\label{eq:LipschitzU}
    \left|U\left(\tilde{p},\MVNNi{}\right) - U\left(p,\MVNNi{}\right) \right| \leq \twonorm[\tilde{p}-p]\twonorm[c].   
 \end{align}
 \Cref{eq:LipschitzU}
 by definition says that $U$ is Lipschitz-continuous in $p$ with Lipschitz-constant $\twonorm[c]$.
 
 So, we finally obtain, that $W$ is Lipschitz-continuous in $p$ with Lipschitz-constant $(n+1)\twonorm[c]$. %
\end{proof}

\begin{lemma}[Convexity]\label{le:Wconvex}
The map $p\mapsto W\left(p,\left(\MVNNi{}\right)_{i=1}^n\right)$ from $\Rpz^m$ to $\Rpz$ is convex.\footnote{The proof also works for any other (possibly non-monotonic) value function $v_i:\X\to\R$ instead of $\MVNNi{}$.}
\end{lemma}
\begin{proof}
    Since $ W\left(p,\left(\MVNNi{}\right)_{i=1}^n\right)\coloneqq R(p) + \sum_{i \in N}U\left(p,\MVNNi{}\right)$ is the sum of $1+n$ functions, we first quickly show that $R$ is obviously convex in $p$ and afterwards we will show that also $U$ is convex in $p$.%

    From \Cref{eq:indirect_revenue} it follows that $R(p)=\sum\limits_{j\in M}c_j p_j$ is linear in $p$ and thus convex in $p$.

    In the remainder of the proof we are going to show that $U\left(p,\MVNNi{}\right)\coloneqq \max\limits_{x\in \X}\left\{\MVNNi{}(x)-\iprod{p}{x}\right\}$ is convex in $p$.
    For every $i\in N$ and for every $x\in \X$, the map $p\mapsto \MVNNi{}(x)-\iprod{p}{x}$ is (affine-)linear\footnote{In the following, we will call affine-linear functions \enquote{linear} too as commonly done in the literature.} in $p$ and thus convex in $p$. As the maximum over convex functions is always convex \citep{Danskin1967nla.cat-vn2047560,Bertsekas1971ControlOU,bertsekas1999nonlinear}, $U\left(p,\MVNNi{}\right)\coloneqq \max\limits_{x\in \X}\left\{\MVNNi{}(x)-\iprod{p}{x}\right\}$ is convex in $p$.

So finally we get that $W$ is convex in $p$, since it is the sum of $(n+1)$ convex functions.
\end{proof}

\begin{lemma}[Sub-gradients]\label{le:WsubGradient} 
Let $\hat{\X}^*_i(p):=\argmax_{x\in \X}\left\{\MVNNi{x}-\iprod{p}{x}\right\}$ denote the set of utility maximizing bundles w.r.t.~$\MVNNi{}$  for the $i$-th bidder. Then the sub-gradients of the map $p\mapsto W\left(p,\left(\MVNNi{}\right)_{i=1}^n\right)$ from $\Rpz^m$ to $\Rpz$ are given by
\begin{multline}\label{eq:subgradWconv}
    \subgrad_p W(p,\left(\MVNNi{}\right)_{i=1}^n)\\
    :=
    \conv{c-\sum_{i\in N}\hat{x}^*_i(p) : %
    \left(\hat{x}^*_i(p)\right)_{i\in N}\in \bigtimes_{i\in N}\hat{\X}^*_i(p)
    },
\end{multline}
where $\convwo$ denotes the convex hull and $\bigtimes$ denotes the Cartesian product. In particular, for each $\left(\hat{x}^*_i(p)\right)_{i\in N}\in \bigtimes_{i\in N}\hat{\X}^*_i(p)$,
\begin{equation}
    c-\sum_{i\in N}\hat{x}^*_i(p) 
    \in
    \subgrad_p W(p,\left(\MVNNi{}\right)_{i=1}^n)
\end{equation}
is \emph{a} subgradient of $W$ with respect to $p$.\footnote{The proof also works for any other (possibly non-monotonic) value function $v_i:\X\to\R$ instead of $\MVNNi{}$.}
\end{lemma}
\begin{proof}
    Since $ W\left(p,\left(\MVNNi{}\right)_{i=1}^n\right)\coloneqq R(p) + \sum_{i \in N}U\left(p,\MVNNi{}\right)$ is the sum of $1+n$ functions, we first quickly compute the (sub\nobreakdash-)gradient of $R$ with respect to $p$ and afterwards we will compute the sub-gradients of $U$ with respect to $p$.

 From \Cref{eq:indirect_revenue} it follows that $R(p)=\sum\limits_{j\in M}c_j p_j$ is linear in $p$ and thus its sub-gradient is uniquely defined as its gradient $\nabla_p R(p)=c$, i.e., $\subgrad_p R(p)=\{c\}$.

    In the remainder of the proof we are going to compute the set of sub-gradients of $U\left(p,\MVNNi{}\right)\coloneqq \max\limits_{x\in \X}\left\{\MVNNi{}(x)-\iprod{p}{x}\right\}$ with respect to $p$ with the help of \href{https://en.wikipedia.org/w/index.php?title=Danskin\%27s_theorem&oldid=1159930476}{Danskin's theorem} (originally proven by \citet{Danskin1967nla.cat-vn2047560} and later refined by \citet{Bertsekas1971ControlOU,bertsekas1999nonlinear}).
Let's define $\phi_i(p,x):=\MVNNi{}(x)-\iprod{p}{x}$. Then $U\left(p,\MVNNi{}\right)= \max\limits_{x\in \X}\phi_i(p,x)$ and $\hat{\X}^*_i(p)=\argmax\limits_{x\in \X}\phi_i(p,x)$. Next, we show that $\phi_i:\Rpz^m\times\X\to\R$ fulfills all the assumptions of \href{https://en.wikipedia.org/w/index.php?title=Danskin\%27s_theorem&oldid=1159930476}{Danskin's theorem}:
\begin{itemize}
\item For every $x\in \X$ the map $p\mapsto\phi_i(p,x)$ is convex and differentiable, since it is (affine-)linear.
\item For every $p\in \Rpz^m$ the map $x\mapsto\phi_i(p,x)$ is continuous, since $\X$ is discrete and every function is continuous with respect to the discrete topology.
\item Furthermore, the set~$\X$ is compact (since it is finite).
\end{itemize}
Under these three assumptions \href{https://en.wikipedia.org/w/index.php?title=Danskin\%27s_theorem&oldid=1159930476}{Danskin's theorem} tells us that 
\begin{align}\label{eq:subUgeneral}
    \subgrad_p U\left(p,\MVNNi{}\right)
    = \conv{\nabla_p\phi_i(p,x) : x\in\hat{\X}^*_i(p)}.
\end{align}
Now we can simply compute $\nabla_p\phi_i(p,x)$ for any constant $x\in\X$, i.e.,
\begin{subequations}\label{eq:gradphi}
    \begin{align}
    \nabla_p\phi_i(p,x)
    &=\nabla_p\left(\MVNNi{}(x)-\iprod{p}{x}\right)\\
    &=\nabla_p\MVNNi{}(x)-\nabla_p\iprod{p}{x}\\
    &=0-x
    =-x.
\end{align}
\end{subequations}
Plugging in \Cref{eq:gradphi} into \Cref{eq:subUgeneral} results in
\begin{align}\label{eq:subU}
    \subgrad_p U\left(p,\MVNNi{}\right)
    = \conv{-x : x\in\hat{\X}^*_i(p)}.
\end{align}
Note that set of sub-gradients of a sum of real-valued\footnote{Note that \citet[Theorem~23.8]{Rockafellar1970} is formulated for proper convex functions with overlapping effective domains that can potentially attain $+\infty$ as value. In our case of real-valued function (which implies that they cannot attain $+\infty$) all convex functions are proper and their effective domain is simply their domain.} convex functions, is equal to the sum\footnote{The (Minkowski) sum of two sets $A$ and $B$ is defined as the set $A+B=\{a+b:(a, b)\in A\times B\}$ of all possible sums of any element of the first set $A$ and any element of the second set $B$.} of the sets of sub-gradients of the individual functions \citep[Theorem~23.8]{Rockafellar1970}. Further note that the sum of convex hulls of sets is equal to the convex hull of their sum. We use these two insights and \Cref{eq:subU} to finally compute the set of sub-gradients
\begin{subequations}
\begin{align}
    \MoveEqLeft[1] \subgrad_pW(p,\left(\MVNNi{}\right)_{i=1}^n)=%
    \\
    &=
    \subgrad\left(R(p) + \sum_{i \in N}U\left(p,\MVNNi{}\right)\right)\\
    &=
    \subgrad R(p) + \sum_{i \in N}\subgrad U\left(p,\MVNNi{}\right)\\
    &=
    \{c\}+\sum_{i \in N}\conv{-x : x\in \hat{\X}^*_i(p) }\\
    &=
    \{c\}+\convwo{\sum_{i \in N}\left\{-x : x\in \hat{\X}^*_i(p) \right\}}\\
    &=
     \{c\}-\convwo{\sum_{i \in N}\hat{\X}^*_i(p)}\\
    &=
    \{c\}-\convwo{\left\{\sum_{i \in N} \hat{x}^*_i(p) : %
    \left(\hat{x}^*_i(p)\right)_{i\in N}\in \bigtimes_{i\in N}\hat{\X}^*_i(p)
    \right\}}\\
    &=
    \convwo{\left\{c-\sum_{i \in N} \hat{x}^*_i(p) : %
    \left(\hat{x}^*_i(p)\right)_{i\in N}\in \bigtimes_{i\in N}\hat{\X}^*_i(p)
    \right\}}
    ,\label{eq:proof:subgradW}
\end{align}
\end{subequations}
which proves the main statement~\eqref{eq:subgradWconv} of \Cref{le:WsubGradient}. This directly implies that for each $\left(\hat{x}^*_i(p)\right)_{i\in N}\in \bigtimes_{i\in N}\hat{\X}^*_i(p)$,
\begin{equation}
    c-\sum_{i\in N}\hat{x}^*_i(p) 
    \in
    \subgrad_p W(p,\left(\MVNNi{}\right)_{i=1}^n)
\end{equation}
is an element of the convex hull in \eqref{eq:proof:subgradW} and thus a sub-gradient of $W$ with respect to $p$.
\end{proof}

Next, \Cref{le:WaeDiff} shows that the map $p\mapsto W\left(p,\left(\MVNNi{}\right)_{i=1}^n\right)$ from $\Rpz^m$ to $\Rpz$ is Lebesgue-almost everywhere differentiable with (proper) gradient\footnote{Note that the gradient in \Cref{eq:app:gradWae,eq:app:gradW} actually denotes a classical gradient and not just a sub-gradient.}
\begin{equation}\label{eq:app:gradWae}
    \nabla_p W(p,\left(\MVNNi{}\right)_{i=1}^n)\stackrel{\mathrm{a.e.}}{=}c-\sum_{i\in N}\hat{x}^*_i(p),
\end{equation} 
where $\hat{x}^*_i(p)\in\hat{\X}^*_i(p):=\argmax_{x\in \X}\left\{\MVNNi{x}-\iprod{p}{x}\right\}$ denotes the utility maximizing bundle w.r.t.~$\MVNNi{}$ for the $i$-th bidder, which is Lebesgue-almost everywhere unique.
\begin{lemma}[A.e.~Differentiable]\label{le:WaeDiff} %
Let $\left(\MVNNi{}\right)_{i=1}^n$ be a tuple of mMVNNs.
Then there exists a dense\footnote{The set $P$ being \emph{dense} means that its topological closure $\overline{P}$ covers the whole space, i.e., $\overline{P}=\Rpz^m$.} subset $P\subseteq \Rpz^m$ such that $\Rpz^m\setminus P$ is a Lebesgue null set\footnote{The set $\Rpz^m\setminus P$ being a \emph{(Lebesgue) null set} means that its $m$-dimensional Lebesgue measure (i.e., the $m$-dimensional volume)~$\lambda^m\left(\Rpz^m\setminus P\right)=0$. Within this work \enquote{a.e.} always corresponds to \enquote{Lebesgue almost everywhere}.}, and that $\forall p\in P:$
\begin{enumerate}
\item\label{itm:uniqueness} A unique optimizer $\hat{x}^*_i(p)\in\hat{\X}^*_i(p)$ exists, i.e.,
\begin{align}
    \left\{x^*_i(p)\right\}=\hat{\X}^*_i(p):=\argmax_{x\in \X}\left\{\MVNNi{x}-\iprod{p}{x}\right\}
\end{align}
and
\item\label{itm:app:gradW} the map $p \mapsto W\left(p,\left(\MVNNi{}\right)_{i=1}^n\right)$ is differentiable with gradient 
\begin{equation} \label{eq:app:gradW}
    \nabla_p W(p,\left(\MVNNi{}\right)_{i=1}^n)=c-\sum_{i\in N}\hat{x}^*_i(p),
\end{equation}
which is also the unique sub-gradient, i.e. $\left\{c-\sum_{i\in N}\hat{x}^*_i(p)\right\}=\subgrad_p W(p,\left(\MVNNi{}\right)_{i=1}^n)$.
\item\label{itm:app:locConstW} The unique optimizer $\hat{x}^*_i(p)$ is constant in a local neighborhood of $p$, and thus $\partial_p \hat{x}^*_i(p)=0$.\footnote{Note that one could even prove that the gradient is continuous on $P$, while the gradient can have jumps at the null set $\Rpz^m\setminus P$. The proof also works for any other (possibly non-monotonic) value function $v_i:\X\to\R$ instead of $\MVNNi{}$.}
\end{enumerate}
\end{lemma}
\begin{proof}
We start this proof by defining $P$ and showing a.e.\ differentiability.
 Since $ W\left(p,\left(\MVNNi{}\right)_{i=1}^n\right)\coloneqq R(p) + \sum_{i \in N}U\left(p,\MVNNi{}\right)$ is the sum of $1+n$ functions, we first quickly show that $R$ is obviously differentiable and afterwards we show that each $U\left(p,\MVNNi{}\right)$ is differentiable on a set $P_i$. Afterwards we will define $P$ as their intersection.

 From \Cref{eq:indirect_revenue} it follows that $R(p)=\sum\limits_{j\in M}c_j p_j$ is linear in $p$ and thus differentiable.

 Next, we show that $U\left(p,\MVNNi{}\right)\coloneqq \max\limits_{x\in \X}\left\{\MVNNi{}(x)-\iprod{p}{x}\right\}$ is differentiable a.e.\ and define $P_i$. We know already from the proof of \Cref{le:Wconvex} that $U\left(p,\MVNNi{}\right)$ is convex and \citet[Theorem~25.5 on p.~246]{Rockafellar1970} tells us that for any\footnote{There are some very mild technical assumptions in \citet[Theorem~25.5 on p.~246]{Rockafellar1970} that do not matter in our case. \citet{Rockafellar1970} assumes that the function is defined on $\R^m$. In our case we could easily extend the domain of our function from $\Rpz^m$ to $\R^m$. As \citet[Theorem~25.5 on p.~246]{Rockafellar1970} is formulated for proper convex functions that can also attain $+\infty$ one could extend any convex function from a convex domain (such as $\Rpz^m$) to be defined to be $+\infty$ outside of that convex domain. Note that if $P_i$ is dense in $\R_{>0}^m$ it is also dense in $\Rpz^m$ and that $\lambda^m\left(\Rpz^m\setminus \R_{>0}^m\right)=0$.} real-valued convex function there exists a dense set $P_i$ on which the functions is differentiable with the complement of $P_i$ being a null set. %
 Thus, for each $p\mapsto U\left(p,\MVNNi{}\right), \, i\in N$, we obtain such a $P_i$.
 If $1+n\in \N$ functions are differentiable at a point $p$, then their sum is too. Thus, $p\mapsto W\left(p,\left(\MVNNi{}\right)_{i=1}^n\right)$ is differentiable on $P:=\bigcap_{i=1}^n P_i$. Since each $P_i$ is dense in $\Rpz^m$, $P$ is also dense in $\Rpz^m$. Moreover, it holds that
 $\lambda^m\left(\Rpz^m\setminus P\right)
 =\lambda^m\left(\bigcup_{i=1}^n( \Rpz^m \setminus P_i)\right)
 \leq \sum_{i=1}^n\lambda^m\left( \Rpz^m \setminus P_i\right)=\sum_{i=1}^n 0=0$, i.e., the $m$-dimensional Lebesgue measure of $\Rpz^m\setminus P$ vanishes. Putting everything together, we get that $p\mapsto W\left(p,\left(\MVNNi{}\right)_{i=1}^n\right)$ is a.e. differentiable (concretely differentiable for all $p \in P$).

 Next, we prove the uniqueness of $\hat{x}^*_i(p)\in\hat{\X}^*_i(p)$ for every $p\in P_i\supseteq P$. We know that $p\mapsto U\left(p,\MVNNi{}\right)$ is differentiable at $p\in P_i$, and thus we know that the sub-gradient is unique, i.e., $\left|\subgrad_p U\left(p,\MVNNi{}\right)\right|=\left|\left\{\nabla_p U\left(p,\MVNNi{}\right)\right\}\right|=1$. We have already computed this set of sub-gradients $\subgrad_p U\left(p,\MVNNi{}\right)
    = \conv{-x : x\in\hat{\X}^*_i(p)}=\convwo -\hat{\X}^*(p)\supseteq-\hat{\X}^*(p)$ in \Cref{eq:subU} in the proof of \Cref{le:WsubGradient}. This set can only be a singleton, if $\hat{\X}^*(p)$ is a singleton. Finally, $\left|\hat{\X}^*(p)\right|=1$ immediately implies \Cref{itm:uniqueness}, i.e., the uniqueness of $\hat{x}^*_i(p)\in\hat{\X}^*_i(p)$.

    Using that $p\mapsto U\left(p,\MVNNi{}\right)$ is differentiable for every $p\in P$ together with \Cref{itm:uniqueness} and \Cref{le:WsubGradient}, we finally obtain \Cref{itm:app:gradW} for every $p\in P$.

    For \Cref{itm:app:locConstW} it is crucial that $\X$ is finite. For every $p\in P$, we know from \Cref{itm:uniqueness} that there is a unique maximizer $\hat{x}^*_i(p)\in\hat{\X}^*_i(p)=\argmax\limits_{x\in \X}\phi_i(p,x)$ with
    $\phi_i(p,x):=\MVNNi{}(x)-\iprod{p}{x}$. Since $\X$ is finite, we can define \begin{align}\label{eq:app:def:eps}
    \epsilon_i(p):=%
    \phi_i(p,\hat{x}^*_i(p))-
    \max\limits_{x\in \X\setminus\{\hat{x}^*_i(p)\}}\phi_i(p,x),   
    \end{align}
    which
    has to be strictly larger than 0 for every $p\in P$ because of the uniqueness of the maximizer.
    \Cref{eq:app:def:eps} implies that $\hat{x}^*_i(p)$ outperforms every other $x\in\X$ by at least a margin of $\epsilon_i(p)$. Since $\phi_i$ is continuous, $\hat{x}^*_i(p)$ cannot be \enquote{overtaken} by any other $x\in\X$ within a small neighbourhood of $p$ as we will calculate explicitly in the following using the Lipschiz constant $\twonorm[c]$ of $\phi_i$ that we derived in the proof of \Cref{le:Wcontinous}:
    For any $x\in\X\setminus \{\hat{x}^*_i(p)\}, p\in P, \tilde{p}\in \Rpz^m$, we have
\begin{subequations}
    \begin{align}       \phi_i(\tilde{p},x)
    &\leq
    \phi_i(p,x) +\twonorm[c]\twonorm[\tilde{p}-p]\\
    &\leq
    \max\limits_{x\in \X\setminus\{\hat{x}^*_i(p)\}}\phi_i(p,x)+\twonorm[c]\twonorm[\tilde{p}-p]\\
    &\leq
    \phi_i(p,\hat{x}^*_i(p)) -\epsilon_i(p) +\twonorm[c]\twonorm[\tilde{p}-p]\\
     &\leq
     \phi_i(\tilde{p},\hat{x}^*_i(p)) -\epsilon_i(p) +2\twonorm[c]\twonorm[\tilde{p}-p].
    \end{align}
\end{subequations}
From this inequality we obtain, that within an open ball with radius $\min_{i\in N}\frac{\epsilon_i(p)}{2\twonorm[c]}>0$ around $p\in P$, the optimizers $\hat{x}^*_i(p)$ stay constant for every $i\in N$. Therefore, the differential $\partial_p \hat{x}^*_i(p)$ is zero for all $p\in P$, which concludes the proof of \Cref{itm:app:locConstW}.
\end{proof}
Putting everything together, we can finally prove \Cref{thm:GD_on_W}.
\begin{proof}[Proof of \Cref{thm:GD_on_W}]
    Combining \Cref{le:Wcontinous,le:Wconvex,le:WsubGradient,le:WaeDiff} proves \Cref{thm:GD_on_W} (and even slightly stronger statements, e.g., \Cref{le:WsubGradient} fully specifies the set of \emph{all} sub-gradients).
\end{proof}

\Cref{thm:GD_on_W} and most of the statements from \Cref{le:Wcontinous,le:Wconvex,le:WsubGradient,le:WaeDiff} can be proven under even less assumptions as we will discuss in the following \Cref{rem:GeneralValueFunctions,rem:ContinuousInputSpace} which generalize the theory to further settings.
\begin{remark}[General Value Functions]\label{rem:GeneralValueFunctions}
     \Cref{thm:GD_on_W} and \Cref{le:Wcontinous,le:Wconvex,le:WsubGradient,le:WaeDiff} are also true for the map $p\mapsto W(w,(g_i)_{i=1}^{n})$, with any (possibly non-monotonic) value function $g_i:\X\to\R$ instead of $g_i=\MVNNi{}$, since we never used any specific properties of MVNNs $\MVNNi{}$ in the proofs of these statements. In particular, this includes bidders' true value functions $g_i=v_i,\, v_i:\X\to\R$.
\end{remark}
\begin{remark}[Continuous Input Space $\tilde{\X}$]\label{rem:ContinuousInputSpace}
    \Cref{thm:GD_on_W} and all of the statements from \Cref{le:Wcontinous,le:Wconvex,le:WsubGradient,le:WaeDiff} except \Cref{itm:app:locConstW} from \Cref{le:WaeDiff} are also true if one replaces the finite set $\X$ by any (possibly non-discrete) compact set $\tilde{\X}$ (if one extends the definition of $\MVNNi{}$ in the natural way from $\X$ to $\tilde{\X}\subseteq\R^m$). When combining \Cref{rem:GeneralValueFunctions,rem:ContinuousInputSpace} one has to assume that the $g_i:\tilde{\X}\to\R$ are continuous for our proof.\footnote{Note that if $\tilde{\X}$ is finite, every function $g_i:\tilde{\X}\to\R$ is continuous by definition.}
    \Cref{itm:app:locConstW} from \Cref{le:WaeDiff} can be violated for compact $\tilde{\X}\subseteq\R^m$ if $\left|\tilde{\X}\right|=\infty$.\footnote{Note that, while \Cref{itm:app:locConstW} from \Cref{le:WaeDiff} was used for the intuitive sketch of the proof of \Cref{thm:GD_on_W} in the main paper, \Cref{itm:app:locConstW} from \Cref{le:WaeDiff} is not necessary at all for the mathematical rigorous proof of \Cref{thm:GD_on_W} given in \Cref{subsec:app:GD_on_W}.}
\end{remark}

\begin{remark}[Piece-wise Linear]
    In our case of finite $\X$, one can intuitively see with similar arguments as in the proof of \Cref{le:Wconvex}, that $p\mapsto W\left(p,\left(\MVNNi{}\right)_{i=1}^n\right)$ is piece-wise linear as a sum over a linear function and $n$ functions which are the maxima over finitely many linear functions.
\end{remark}

\section{Experiment Details}\label{sec:Experiment Details}
In this section, we present all details of our experiments from \Cref{sec:experiments}.

\subsection{SATS Domains}\label{subsec:appendix_SATS_domains}
In this section, we provide a more detailed overview of the four SATS domains, which we use to experimentally evaluate ML-CCA:
\begin{itemize}[leftmargin=*,topsep=0pt,partopsep=0pt, parsep=0pt]
\item \textbf{Global Synergy Value Model (GSVM)} \cite{goeree2010hierarchical} has 18 items with capacities $c_j=1$ for all 
$j\in \{1,\ldots,18\}$, $6$ \emph{regional} and $1$ \emph{national bidder}. In GSVM the value of a package increases by a certain percentage with every additional item of interest. Thus, the value of a bundle only depends on the total number of items contained in a bundle which makes it one of the
simplest models in SATS. In fact, bidders’ valuations exhibit at most two-way(i.e., pairwise) interactions between items.
\item \textbf{Local Synergy Value Model (LSVM)} \cite{scheffel2012impact} has $18$ items with capacities $c_j=1$ for all 
 $j\in \{1,\ldots,18\}$, $5$ \emph{regional} and $1$ \emph{national bidder}. Complementarities arise from spatial proximity of items.
\item \textbf{Single-Region Value Model (SRVM)} \cite{weiss2017sats} has $3$ items with capacities $c_1=6,c_2=14,c_3=9$ and $7$ bidders (categorized as  \emph{local}, \emph{high frequency}, \emph{regional}, or \emph{national}) and models UK 4G spectrum auctions.
\item \textbf{Multi-Region Value Model (MRVM)} \cite{weiss2017sats} has $42$ items with capacities $c_j\in \{2,3\}$ for all 
 $j\in \{1,\ldots,42\}$ and $10$ bidders (\emph{local}, \emph{regional}, or \emph{national}) and models large Canadian 4G spectrum auctions.
\end{itemize}
In the efficiency experiments in this paper, we instantiated for each SATS domain the $100$ synthetic CA instances with the seeds $\{101,\ldots,200\}$. We used \href{https://github.com/spectrumauctions/sats/releases/}{SATS version 0.8.1}. 

\subsection{Compute Infrastructure}\label{subsec:app:compute_infrastructure}
All experiments were conducted on a compute cluster running Debian GNU/Linux 10 with Intel Xeon E5-2650 v4 2.20GHz processors with 24 cores and 128GB RAM and Intel E5 v2 2.80GHz processors with 20 cores and 128GB RAM and Python 3.8.10.

\subsection{Hyperparameter Optimization}\label{subsec:appendix_hpo}
In this section, we provide details on our exact HPO methodology and the ranges that we used. \

We separately optimized the HPs of the mMVNNs for each bidder type of each domain, using a different set of SATS seeds than for all other experiments in the paper. Specifically, for each bidder type, we first trained an mMVNN using as initial data points the demand responses of an agent of that type during $50$ consecutive CCA clock rounds, and then measured the generalization performance of the resulting network on a validation set that was created by drawing $500$ price vectors where the price of each item was drawn uniformly at random from the range of zero to three times the average maximum value of an agent of that type for a single item (which was determined using separate seeds, see validation set 2 in \Cref{fig:pred_vs_true_national_MRVM}). The number of seeds used to evaluate each model was equal for all models and set to 10. 
Finally, for each bidder type we selected the set of HPs that performed the best on this validation set with respect to the coefficient of determination ($R^2$). The full range of HPs tested for all agent types and all domains is shown in \Cref{table_HPO_ranges}, while the winning configurations are shown in \Cref{tab:appenedix_hpo_winners}.

Additionally, we determined the set of HPs with the best generalization performance on validation set 2 using as evaluation metric a shift-invariant variation of $R^2$, defined as: 
\begin{gather}
    R^2_{c} = 1 - \frac{\sum_r ( (v_i(x^r) - \bar{v_i}) - (\mathcal{M}_i(x^r) - \bar{\mathcal{M}}) )^2}
    {\sum_r (v_i(x^r) - \bar{v_i})^2},
\end{gather}
where $v_i(x^r)$ is the true value of the bidder for the $r$-th bundle, $\mathcal{M}_i(x^r)$ is the neural network's predicted value for that bundle, and $\bar{v_i}$ and $\bar{\mathcal{M}_i}$ are their empirical means, respectively.  
The reason that we opted for this shift-invariant version of $R^2$ is that, as explained in \Cref{subsec:mvnns}, learning the true value functions of the agents up to a constant shift suffices for our query generation procedure as described in \Cref{sec:ML-powered Demand Query Generation}.
Surprisingly, in all domains our mechanism performed slightly worse with those HPs, with the maximum efficiency delta between the two configurations being $1.2$\% in LSVM. However, in all domains results were qualitatively identical.
The winning configurations for both metrics are shown in \Cref{tab:appenedix_hpo_winners,tab:appenedix_hpo_winners_r2c}. In all domains we chose the configurations from \Cref{tab:appenedix_hpo_winners} for our efficiency experiments.

\begin{table}
\centering
\resizebox{0.4\textwidth}{!}{
\begin{tabular}{ll}
\toprule
 \multicolumn{1}{l}{\textbf{Hyperparameter}} &                      \multicolumn{1}{l}{\textbf{HPO-Range}} \\
\midrule
Non-linear Hidden Layers            & {[}1,2,3{]}                         \\
         Neurons per Hidden Layer             & {[}8, 10, 20, 30{]} \\
         Learning Rate  & (1e-4, 1e-2)                        \\
         Epochs\footnotemark            & {[}30, 50, 70, 100{]}                            \\
         L2-Regularization & (1e-8, 1e-2)                       \\
         Linear Skip Connections\footnotemark   & {[}True, False{]}           \\
\bottomrule
\end{tabular}
}
\vskip 0.1cm
    \caption{HPO ranges for all domains.}
    \label{table_HPO_ranges}
\end{table}
\setcounter{footnote}{\value{footnote}-1}
\footnotetext{For GSVM and LSVM, the number of epochs was fixed to 30}
\setcounter{footnote}{\value{footnote}+1}
\footnotetext{For the definition of (m)MVNNs with a linear skip connection, please see \citet[Definition~F.1]{weissteiner2023bayesian}}

\begin{table*}[ht]
    \renewcommand\arraystretch{1.2}
    \setlength\tabcolsep{2pt}
	\robustify\bfseries
	\centering
	\begin{sc}
	\resizebox{1\textwidth}{!}{
	\small
\begin{tabular}{llllllll}
\toprule
Domain & Bidder Type & \# Hidden Layers & \# Hidden Units & Lin. Skip & Learning Rate & L2 Regularization & Epochs \\
\midrule
GSVM & Regional & 2 & 20 & False & 0.005 & 0.00001 & 30 \\
     & National & 3 & 30 & True & 0.001 & 0.000001 & 30 \\
\midrule
LSVM & Regional & 1 & 30 & True & 0.01 & 0.000001 & 30 \\
     & National & 3 & 20 & False & 0.005 & 0.0001 & 30 \\
\midrule
SRVM & Local & 2 & 20 & True & 0.01 & 0.0001 & 30 \\
     & Regional & 1 & 20 & True & 0.01 & 0.0001 & 50 \\
     & National & 1 & 30 & False & 0.005 & 0.00001 & 70 \\
     & High Frequency & 2 & 20 & False & 0.01 & 0.00001 & 30 \\
\midrule
MRVM & Local & 3 & 20 & True & 0.005 & 0.000001 & 100 \\
     & Regional & 2 & 20 & True & 0.001 & 0.001 & 100 \\
     & National & 3 & 20 & True & 0.001 & 0.0001 & 50 \\
\bottomrule
\end{tabular}
}
    \end{sc}
    \vskip -0.1 in
    \caption{Winning HPO configurations for $R^2$}
\label{tab:appenedix_hpo_winners}
\end{table*}

\begin{table*}[ht]
    \renewcommand\arraystretch{1.2}
    \setlength\tabcolsep{2pt}
	\robustify\bfseries
	\centering
	\begin{sc}
	\resizebox{1\textwidth}{!}{
	\small
\begin{tabular}{llllllll}
\toprule
Domain & Bidder Type & \# Hidden Layers & \# Hidden Units & Lin. Skip & Learning Rate & L2 Regularization & Epochs \\
\midrule
GSVM & National & 3 & 10 & True & 0.001 & 0.001 & 30 \\
     & Regional & 2 & 10 & True & 0.01 & 0.001 & 30 \\
\midrule
LSVM & National & 1 & 10 & True & 0.005 & 0.01 & 30 \\
     & Regional & 3 & 20 & True & 0.005 & 0.0001 & 30 \\
\midrule
SRVM & Local & 2 & 30 & True & 0.01 & 0.0001 & 30 \\
     & Regional & 2 & 20 & True & 0.01 & 0.000001 & 50 \\
     & National & 1 & 20 & True & 0.01 & 0.0001 & 50 \\
     & High Frequency & 2 & 30 & True & 0.01 & 0.00001 & 70 \\
\midrule
MRVM & National & 1 & 20 & True & 0.001 & 0.000001 & 30 \\
     & Regional & 3 & 20 & True & 0.001 & 0.000001 & 50 \\
     & Local    & 2 & 20 & False & 0.001 & 0.000001 & 30 \\
\bottomrule
\end{tabular}
}
    \end{sc}
    \vskip -0.1 in
    \caption{Winning HPO configurations for $R^2_{c}$}
\label{tab:appenedix_hpo_winners_r2c}
    \vskip -0.3cm
\end{table*}

\subsection{Details on mMVNN Training}\label{subsec:Details MVNN-Training}
\begin{remark}[Other ML-models]
    Note that our training method \textsc{TrainOnDQs} (\Cref{alg:train_on_dqs}) also works for any other ML method that can be trained via GD \emph{and} for which the inner optimization problem $\hat{x}^*_i(p^r) \in \argmax_{x\in \X}\left\{\mathcal{M}_i^{\theta_t}(x) - \iprod{p^r}{x}\right\}$ (see \Cref{alg_dq:line5} of \Cref{alg:train_on_dqs}) can be solved efficiently. This is the case for (m)MVNNs, where \Cref{alg_dq:line5} can be solved as a MILP analogously to \cite{weissteiner2022monotone}. Another example would be classical ReLU-neural networks (NNs)\footnote{Note that in principal for every NN with piece-wise linear activation function (e.g., ReLU, bReLU or Leaky ReLU) a MILP-formulation is possible. However for other activation functions such as the sigmoid activation function an exact MILP-formulation is not possible.}, where such a MILP formulation exists too \citep{weissteiner2020deep}, which are suitable for domains without the free disposal property.
\end{remark}

\begin{remark}[Initialization]
   We use the initialization scheme introduced by \citet{weissteiner2023bayesian}, which offers advantages over the original initialization scheme used by \citet{weissteiner2022monotone} as explained in \citet[Section~3.2 and Appendix~E]{weissteiner2023bayesian}. 
\end{remark}
In the conducted experiments, Python 3.8.10 and PyTorch 2.0.0, were employed as the primary programming language and framework for implementing the mMVNNs. The Adam optimizer was chosen as the optimization algorithm for the training process. To further enhance the training procedure, the cosine annealing scheduler was utilized, dynamically adjusting the learning rate over epochs to facilitate convergence and prevent premature stagnation.

\subsection{Details MILP Parameters}\label{subsec:Details MILP Parameters}
There are three distinct points in which MILPs are solved in our ML-powered combinatorial clock auction (ML-CCA):
\begin{enumerate}
\item in the training of mMVNNs according to \Cref{alg_dq:line5} in \Cref{alg:train_on_dqs},
\item for $W$ minimization according to \Cref{thm:GD_on_W} in order to predict the demand of each agent at a given price vector (see \Cref{line:alg:MILPxhat} in \Cref{alg:Constrained_W_minimization}), and
\item finally to solve the \emph{winner determination problem (WDP)} and determine the resulting allocation based on the elicited bids (see \Cref{alg_line:final_allocation} in \Cref{ML-CCA}).
\end{enumerate}
The first two MILPs are of the same type: given as input an mMVNN $\mathcal{M}_i^{\theta}$ that approximates a bidder's value function and linear item prices $p\in \Rpz^m$, find the utility maximizing bundle for that bidder, i.e., solve 
$\hat{x}^*_i(p) \in \argmax_{x\in\X}\left\{ \mathcal{M}_i^\theta(x) - \iprod{p}{x}\right\}$. 
The third MILP is of a different type: given as inputs a set of bundle-value tuples from each agent, find a feasible allocation that maximizes reported social welfare. This WDP is described in more detail in \Cref{sec:Preliminaries} and \Cref{WDPFiniteReports}. In each clock round, only two WDP MILPs need to be solved, one involving just the clock bids of the agents and one also including the bids that would result from the clock bids raised heuristic. For each agent, on average two thousand MILPs of type 1 need to be solved per clock round. It should be noted that they are very fast, as a single MILP of this type can be solved in under 200 milliseconds in our server architecture as described in \Cref{subsec:app:compute_infrastructure}. 
The MILPs of the first type used a formulation based on the MILP-formulation for MVNNs in \cite[Section~3.2 and Appednix~F]{weissteiner2023bayesian} which is an improved version of the MILP-formulation for MVNNs in \cite[Theorem~2 in Section~3.1 and Appendix~C.5]{weissteiner2022monotone}.
The MILPs of the first type were solved using the Gurobi 10 solver, and the WDP MILPs were solved using CPLEX 20.01.0 . For all MILPs, we set the  feasibility tolerance to $1e-9$, the integrality tolerance to $1e-8$ and the relative MIP optimality gap to $1e-06$. All other parameters were set to their respective default values. 

\subsection{Details on \textsc{NextPrice} Procedure} \label{subsec:app_constrained_W_minimization}
In this section, we describe the details of the \textsc{NextPrice} procedure from \Cref{alg_line:next_pv} in \Cref{ML-CCA}. Given trained mMVNNs, \textsc{NextPrice} generates new demand queries by minimizing $W$ under constraint~\eqref{eq:app:corollary_constraint} via GD which is based on \Cref{thm:app:connection_clearing_prices_efficiency_constrainedVersion,thm:GD_on_W}.

\subsubsection{Detailed motivation of constraint~\eqref{eq:app:corollary_constraint}}

In the case that linear clearing prices (LCPs) exist, both minimizing $W$ with or without constraint~\eqref{eq:app:corollary_constraint} would lead to clearing prices and thus to efficient allocations, if the mMVNNs approximate the value functions $v$ well enough, as shown in  \Cref{thm:app:connection_clearing_prices_efficiency_constrainedVersion} and \Cref{thm:connection_clearing_prices_efficiency}, respectively. In this case, constraint~\eqref{eq:app:corollary_constraint} is neither beneficial nor harmful, because both versions are well motivated by theory (see \Cref{thm:connection_clearing_prices_efficiency} and \Cref{thm:app:connection_clearing_prices_efficiency_constrainedVersion}): 
In this case the set of solutions to both problems \eqref{eqs:ConstraintWmin} and \eqref{eq:clearing_objective} are both exactly equal to the set of all possible LCPs and thus result in efficient allocations with no over-demand and no under-demand (see the proof of \Cref{thm:connection_clearing_prices_efficiency} and \Cref{thm:app:connection_clearing_prices_efficiency_constrainedVersion} in \Cref{subsec:app:proof_of_corollarly1connection_clearing_prices_efficiency}). 
Thus in the case that LCPs do exist, both problems \eqref{eqs:ConstraintWmin} and \eqref{eq:clearing_objective} are exactly equivalent and well supported by theory. Indeed our experiments resulted in similarly good results for both versions of the method in domains where LCPs often exist (see GSVM, LSVM, SRVM in \Cref{tab:details_GSVM_LSVM_efficiency_constrained_vs_unconstrained,tab:details_SRVM_MRVM_efficiency_constrained_vs_unconstrained}).\footnote{In \Cref{tab:details_GSVM_LSVM_efficiency_constrained_vs_unconstrained,tab:details_SRVM_MRVM_efficiency_constrained_vs_unconstrained} one can see a slight tendency that also for GSVM, LSVM, SRVM adding constraint~\eqref{eq:app:corollary_constraint} is rather beneficial on average. The reason for this might be, that we do not always find LCPs even in these domains (note that LCPs do not always exist in these domains).} 

However, when no LCPs exist, for every price vector $p$ we have over-demand or under-demand for some goods and we need to make a choice on how to select a price vector $p$ for the next demand query based on the over- and under-demand.

Minimizing the classical $W$ introduced in \Cref{thm:connection_clearing_prices_efficiency} via a classical GD-update rule using the gradient derived in \Cref{thm:GD_on_W} punishes over- and under-demand symmetrically. This can be directly seen 
from the classical GD-update rule
\begin{align*}
p^{\textnormal{new}}_j \stackrel{a.e.}{=} p_j-\gamma (c_j-\sum_{i\in N}(\hat{x}^*_{i}(p))_j),\, \forall j \in M,
\end{align*}
since the magnitude of change in the price vector would be the same for the same amount of over- or under-demand.

The following example also illustrates this symmetry.
\begin{example}
Suppose there is a single item, and two bidders with a value of $5$ and $5 - \epsilon$ for that item. 
Any price $p \in  [5 - \epsilon, 5]$ is a clearing price,
where the indirect utility of the bidder with the higher value is $5 - p$ and of the bidder with the lower value is $0$, while the seller's indirect revenue is $p$ for a $W$ value of $5$. 
For a price that is $x \in \mathbb{R}_{> 0}$ higher than the largest clearing price, i.e., $5 + x$, no agent buys the item and they have an indirect utility of $0$, while the seller's indirect revenue is $5 + x$, for a $W$ of value $5 + x$. 
For a price that is $x$ lower than the smallest clearing price, i.e., $5 - \epsilon - x$, both agents want to buy the item and they have indirect utilities of $\epsilon + x$ and $x$, while the seller's indirect revenue is $5 - \epsilon - x$, for a total $W$ value of $5 + x$.
\end{example}

However, even though the $W$ objective of \Cref{thm:connection_clearing_prices_efficiency} punishes over- and under-demand equally, our preference between them is highly asymmetric; we strongly prefer under-demand over over-demand in practice, 
since the demand responses of the agents at a price vector with no over-demand constitute a feasible allocation, 
while the demand responses of the agents at a price vector with over-demand do not.
This is important because in case that the market does not clear within the clock round limit, our ML-CCA, just like the CCA, will have to 
combine the clock bids of the agents to produce a \emph{feasible} allocation with the highest inferred social welfare according to \Cref{WDPFiniteReports}.
If the demand responses elicited from the agents constituted feasible solutions, it makes it more likely that they can be effectively combined together in the WDP of \Cref{WDPFiniteReports} to produce highly efficient allocations. This is why in domains where no LCPs exist, adding constraint~\eqref{eq:app:corollary_constraint} leads to significantly increased efficiency (see MRVM in \Cref{tab:details_SRVM_MRVM_efficiency_constrained_vs_unconstrained}).
For more intuition on constraint~\eqref{eq:app:corollary_constraint} see the following example.

\begin{example}\label{example:app:constrained_vs_unconstrained_W_minimization}
    Suppose there are $m=2$ items with capacities $c_1=c_2=10$ and $n=2$ bidders with value functions
    \begin{align*}
       v_1(x)&=\max\left\{10\1{x\ge(7,3)},10\1{x\ge(3,7)},9\1{x\ge(4,4)}\right\},
       \\
       v_2(x)&=\max\left\{10\1{x\ge(8,2)},10\1{x\ge(2,8)},9\1{x\ge(4,4)}\right\}.
    \end{align*}
    In this setting, no LCP exists. This can be seen as follows. First note that we can obviously exclude every price vector $p\in \Rpz^2$ with $p_1=0$ or $p_2=0$ from being a LCP. Furthermore, for any price vector $p\in \R_{>0}^2$ bidder 1's utility maximizing bundle $x^*_1(p)\in \X^*_1(p)=\{(7,3),(3,7),(4,4)\}$ and bidder 2's utility maximizing bundle $x^*_2(p)\in \X^*_2(p)=\{(8,2),(2,8),(4,4)\}$. However, from this we see that $x^*_1(p)$ and $x^*_2(p)$ cannot be combined without over- or under-demand, thus violating \Cref{itm:clearing_seller} in \Cref{def:linear_clearing_prices}.
    
    The clearing potential objective $W$ is minimized for every price vector $p=(p_1,p_2)$ that satisfies $p_1=p_2\in[0,0.5]$.\footnote{The objective function $W$ can be formulated as $W(p)=$
    \begin{align*} &\phantom{+}\max\left\{0,10-\iprod{(7,3)}{p},10-\iprod{(3,7)}{p},9-\iprod{(4,4)}{p})\right\}
    \\
    &{+}\max\left\{0,10-\iprod{(8,2)}{p},10-\iprod{(2,8)}{p},9-\iprod{(4,4)}{p})\right\}
    \\
    &{+}(10x+10y) \quad \forall p \in \R_{\geq0}^2.
    \end{align*}
    The set of minimizers can be seen \href{https://www.wolframalpha.com/input?i=+plot+max\%280\%2C-7*x-3*y\%2B10\%2C-3*x-7*y\%2B10\%2C-4*x-4*y\%2B9\%29\%2Bmax\%280\%2C-8*x-2*y\%2B10\%2C-2x-8y\%2B10\%2C-4*x-4*y\%2B9\%29+\%2B\%2810*x\%2B10*y\%29+from+0+to+1.5}{by plotting $W$}. 
    } For $p_1=p_2\in(0,0.5)$, we have $\X^*_1(p)=\{(7,3),(3,7)\}$ and $\X^*_2(p)=\{(8,2),(2,8)\}$ and thus the there is always positive over-demand for at least one of the items, which violates constraint~\eqref{eq:app:corollary_constraint} (this is also the case for $p=(0,0)$).
    For $p=(0.5,0.5)$, we have $\X^*_1(p)=\{(7,3),(3,7),(4,4)\}$ and $\X^*_2(p)=\{(8,2),(2,8),(4,4)\}$. Thus, $p=(0.5,0.5)$ fulfills constraint~\eqref{eq:app:corollary_constraint}, since  $\left((4,4),(4,4)\right)\in\X^*_1(p)\times\X^*_2(p)$ is feasible (see \Cref{foot:PrecisFormulationConstraint}). Thus, $p=(0.5,0.5)$ is the unique solution of the constrained problem~\eqref{eqs:ConstraintWmin}. In this case, $\left((4,4),(4,4)\right)$ would be the efficient allocation with a SCW of $18$. If we had only asked demand queries for prices $p_1=p_2\in[0,0.5)$, i.e., prices that solve the unconstrained minimization problem~\eqref{eq:clearing_objective}, but not the constrained minimization problem~\eqref{eqs:ConstraintWmin}, the WDP would end up with a SCW of only $10$ by allocating some bundle of value $10$ to one of the bidders and nothing to the other bidder, since the WDP is constrained by feasibility due to the limited capacity $c=(10,10)$.
\end{example}

\subsubsection{Details on \textsc{NextPrice} (\Cref{alg:Constrained_W_minimization})}

In \Cref{alg:Constrained_W_minimization}, we present the details of our \textsc{NextPrice} procedure, a modification of the classical GD in \Cref{thm:GD_on_W} that systematically favours under-demand over over-demand (see \Cref{sec:ML-powered Demand Query Generation}). 
Compared to classical gradient descent on $W$ (based on \Cref{thm:GD_on_W}), there are three noteworthy modifications in the \textsc{NextPrice} procedure.
\begin{enumerate}[align=left, leftmargin=*,topsep=2pt]
\item\label{itm:Nextprice:Asymmetry} First, as outlined at the end of \Cref{sec:ML-powered Demand Query Generation}, to incentivize GD on $W(p, (\mathcal{M}_i^\theta)_{i = 1}^{n})$ towards price vectors 
with no positive over-demand, the GD steps punish over- and under-demand \emph{asymmetrically}. Specifically, at each iteration step, in case of predicted over-demand for some good, the gradient step for that good is $(1+\mu)$-times larger than what it would have been in case of under-demand (\Cref{alg:constrianed_w:line_over_demand_update}). 
Finally, $\mu$ is not a constant, but it adaptively increases as long as \Cref{alg:Constrained_W_minimization} has not found a price vector with no predicted over-demand (\Cref{alg:constrianed_w:line_mu_update}).
\item\label{itm:Nextprice:AllEpochsBest} Second, as also outlined at the end of \Cref{sec:ML-powered Demand Query Generation}, once the gradient steps have terminated, \Cref{alg:Constrained_W_minimization} returns the price vector $p$ that led to the lowest value of $W$ among \emph{all} price vectors examined that led to no positive predicted over-demand (i.e., satisfying constraint~\eqref{eq:app:corollary_constraint}) (\Crefrange{alg:constrianed_w:line_FoundBetterFeasiblePrices}{alg:constrianed_w:line_UpdateBetterFeasiblePrices}).\footnote{\label{foot:NoFeasibleFound}If every gradient step resulted in positive predicted over-demand, we would pick just the one with minimal $W$ ignoring the constraint~\eqref{eq:app:corollary_constraint} for this demand query (\Crefrange{alg:constrianed_w:line_NoFeasibleAllocationFound}{alg:constrianed_w:line_UseUnconstraintSolutionInstead}), but this case never occurred in our experiments (see \Cref{fig:app_predicted_feasible_allocation}).}
\item\label{itm:Nextprice:scaledLR} Finally, the learning rate for each good in \Cref{alg:Constrained_W_minimization} is scaled to be proportional to that good's current price $p_j$ (\Cref{alg_constr_W_gd_line:scaling}). 
In theory, we do not have to do this, as \Cref{thm:GD_on_W} guarantees that even a uniform learning rate for all goods cannot get stuck in local minima of $W$.  
Using a uniform learning rate for all goods has two disadvantages. 
First, we would have to tune that learning rate parameter separately for each domain, since the goods' values are in different scales in each domain. 
Additionally, in the SRVM and MRVM domains the prices of different goods can vary by orders of magnitude. If we were to select a uniform learning rate for all goods, we would have to select one that would be suitable for the lowest-valued items 
(otherwise the GD steps would overshoot a lot for the prices of lower-valued goods, i.e. GD, would jump back and forth between large amounts of over- and under- demand for lower-valued goods), 
which would increase significantly the number of steps required until the learning rate becomes sufficiently small so that we do not have such extreme jumps.
Scaling the learning rate for each good proportionally to its current price alleviates both of these potential issues.

\end{enumerate}

\begin{algorithm}[t!]
    \DontPrintSemicolon
    \SetKwInOut{Input}{Input}
    \SetKwInOut{Output}{Output}
    \Input{Trained MVNNs $\left(\MVNNi{}\right)_{i=1}^n$, 
    last $F_{\text{init}}$ round prices $p^{\Qinit}$,
    Epochs $T\in \N$, Learning Rate Base $\lambda>0$, Learning Rate Decay $\eta \in [0,1]$, 
    Feasibility multiplier $\mu$, 
    Feasibility multiplier increment $ \nu$}
    \For{$j = 1 \text{ to } m$}{
    $p^0_j \gets  \sim U[0.75 \cdot p^{\Qinit}_j, 1.25 \cdot p^{\Qinit}_j ] $ \label{alg_constrianed_w:line_Initial_p}\;
    } 
    $W_{best} \gets \infty$ \;  
    $W^f_{best} \gets \infty$ \;  
    $\text{feasible} \gets \text{False}$ \; 
    \For{$t = 0 \text{ to } T - 1$}{
        $W \gets \iprod{p^t}{c}$  \Comment*[r]{\color{CommentColor} Seller's indirect revenue}
        \For(){$i = 1 \text{ to } n$}{
         Solve $\hat{x}^*_i(p^t) \in \argmax_{x\in \X}\mathcal{M}_i^{\theta}(x) - \iprod{p^t}{x}$\label{line:alg:MILPxhat}\;
         $U_i \gets \mathcal{M}_i^{\theta}(\hat{x}^*_i(p^t)) - \iprod{p^t}{\hat{x}^*_i(p^t)}$ \Comment*[f]{\color{CommentColor}Bidder $i$'s indirect Utility} \; 
         $W \gets W + U_i$ \; 
        }
        $d \gets \sum_{i\in N}\hat{x}^*_i(p^t)$ \Comment*[f]{\color{CommentColor}Total Predicted Demand} \;
        \If(\Comment*[f]{\color{CommentColor}Found better prices wrt. $W$}){$W < W_{best}$}
          {
          $W_{best} \gets W$ \; 
          $p_{best} \gets p^t$ \; 
        }
        \If(\Comment*[f]{\color{CommentColor}Found better feasible prices}\label{alg:constrianed_w:line_FoundBetterFeasiblePrices}){$W < W^f_{best}$  and $d\leq c$}
          {
          $W^f_{best} \gets W$ \; 
          $p^f_{best} \gets p^t$ \label{alg:constrianed_w:line_UpdateBetterFeasiblePrices}\; 
          $\text{feasible} \gets \text{True}$ \; 
        }
        \If(\Comment*[f]{\color{CommentColor}Predicted Market Clearing Prices}){$d = c$}
          {
           \textbf{break}
        }
        \For{$j = 1 \text{ to } m$}{
        $\gamma_j \gets \lambda \cdot  p^{t}_j$ \Comment*[r]{\color{CommentColor} Scale l.r. for each good} \label{alg_constr_W_gd_line:scaling}
     $p^{t+1}_j \gets p^{t}_j -  \gamma_j (c_j - d_j) $  \Comment*[r]{\color{CommentColor} \Cref{thm:GD_on_W}, \cref{eq:simplified_update_rule_GD}} 
     \If(\Comment*[f]{\color{CommentColor} Over-demand for good $j$}){$d_j > c_j$}
          {
           $p^{t+1}_j \gets p^{t}_j - \mu \gamma_j (c_j - d_j) $ \Comment*[f]{\color{CommentColor}\Cref{eqs:finalUpdateRule}} \label{alg:constrianed_w:line_over_demand_update}
        }
    $\lambda \gets \lambda \cdot (1 - \eta) $ \Comment*[f]{\color{CommentColor}Learning rate decay}
    }
    \If(\Comment*[f]{\color{CommentColor} No feasible allocation yet}){\textbf{\text{not}} \text{feasible}}
          {
           $\mu \gets \mu \cdot \nu $ \label{alg:constrianed_w:line_mu_update}
        }
    }
    \If(\Comment*[f]{\color{CommentColor}No feasible allocation found}\label{alg:constrianed_w:line_NoFeasibleAllocationFound}){\textbf{not} feasible}
          {
          $p^f_{best} \gets p_{best}$ \Comment*[f]{\color{CommentColor}\Cref{foot:NoFeasibleFound}}\label{alg:constrianed_w:line_UseUnconstraintSolutionInstead}\; 
        }
        \If(\Comment*[f]{\color{CommentColor}Do not enforce feasibility (see \Cref{rem:unconstrainedNextPrice})}){$\mu = \nu = 0$}
          {
          \Return{Prices $p_{best}$ minimizing $W(\cdot, (\mathcal{M}_i^\theta)_{i = 1}^{n})$ }  
        }
    \Return{Feasible prices $p^f_{best}$ minimizing $W(\cdot, (\mathcal{M}_i^\theta)_{i = 1}^{n})$ }
    \caption{\textsc{NextPrice}}
    \label{alg:Constrained_W_minimization}
\end{algorithm}

\begin{remark}\label{rem:unconstrainedNextPrice}
Note that by setting $\mu = \nu = 0$, \textsc{NextPrice} (\Cref{alg:Constrained_W_minimization}) performs symmetrical GD on $W$ without constraint~\eqref{eq:app:corollary_constraint} as suggested by \Cref{thm:GD_on_W} (see \Cref{subsec:app:Results_unconstrained_W} for an empirical evaluation of minimizing $W$ with $\mu = \nu = 0$, i.e., without constraint \eqref{eq:app:corollary_constraint}).
\end{remark}

\begin{remark}
    For each GD-step  we solve the inner optimization problem of \Cref{eq:indirect_utility}, i.e., $\max\limits_{x\in \X}\left(\MVNNi{x}-\iprod{p}{x}\right)$, in \Cref{line:alg:MILPxhat} for each bidder $i$, using the MILP encoding of MVNNs from \citep{weissteiner2022monotone}.\footnote{The actual MILP encoding we are using is based on the improved MILP-encoding of \citet[Section~3.2 and Appendix~F]{weissteiner2023bayesian} and slightly modified to work with mMVNNs instead of classical MVNNs.}
\end{remark}

In our experiments, we use $300$ epochs for all domains, with a good-specific learning rate of $1$\% of the price~$p_j^t$ of that good and a learning rate decay of $0.5$\%, i.e., we set 
$T = 300, \lambda = 0.01$ and $\eta = 0.005$, while we set 
$\mu = 2$ and $\nu = 1.01$.

Intuitively, this way \Cref{alg:Constrained_W_minimization} punishes over-demand at least three times as much as under-demand, which means that it can very quickly converge to a price region with no over-demand, and then it starts minimizing under-demand in the same way as suggested by \Cref{thm:GD_on_W}.
Setting $\nu$ to a number even slightly larger than $1$ ensures that \Cref{alg:Constrained_W_minimization} can converge to such a price region even in the extreme case where it starts with large amounts of over-demand. 
As shown in \Cref{fig:app_predicted_feasible_allocation}, even in the MRVM domain where no linear clearing prices exist, the modifications of \Cref{alg:Constrained_W_minimization} were sufficient for it to return a price vector with no predicted over-demand in \emph{all} ML-powered clock rounds, in  $100$\% of our instances (i.e., SATS seeds), while this number was almost $0$\% if we were to apply symmetrical GD as suggested by \Cref{thm:GD_on_W} by setting $\mu = \nu  = 0$ in our \Cref{alg:Constrained_W_minimization}. 
Those results can be found in \Cref{subsec:app:Results_unconstrained_W}.

\begin{figure}[h!]
    \begin{center}
\includegraphics[width=1\columnwidth,trim=0 0 0 0, clip]{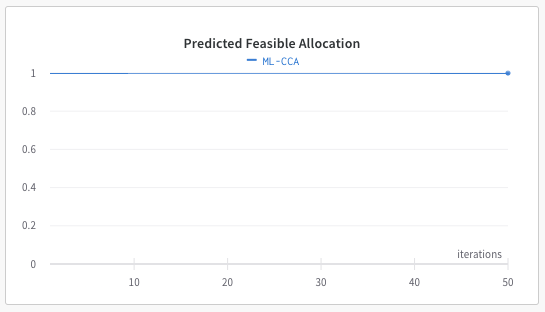}
    \caption{Fraction of instances in the MRVM domain where the price vector returned by \Cref{alg:Constrained_W_minimization} was predicted to be feasible per iteration (starting after the $\Qinit$-phase, i.e., in clock round 50). This fraction is constantly $100$\% across all rounds.}
    \label{fig:app_predicted_feasible_allocation}
    \end{center}
    \vskip -0.25cm
\end{figure}

The following \Cref{example:asymmetryAlgorihtm} shows that even in cases where the constraint~\eqref{eq:app:corollary_constraint}
is mathematically irrelevant, the asymmetry of \Cref{alg:Constrained_W_minimization} can still be very beneficial.

\begin{example}\label{example:asymmetryAlgorihtm}
    Let there be $m=1$ item with capacity $c=10$ and $n=2$ bidders with value functions $v_1(x)=6\1{x\geq6}$ and $v_2(x)=3\1{x\geq1}+2\1{x\geq5}$.\footnote{In the case of $m=1$, we do not distinguish between the 1-dimensional vector $x$ and the number $x_1$ (just as we write $p=p_1$ and $c=c_1$).} First, note that similarly as in \Cref{example:app:constrained_vs_unconstrained_W_minimization} in this example also no LCP exists.
    For every $p<0.5$ we have an over-demand of 1, and for every $p>0.5$ we have an under-demand of at least $3$ ($3$ for $p\in(0.5,1]$, $9$ for $p\in[1,3]$ or $10$ for $p\geq3$).
    For $p=0.5$, we have $\X^*_1(0.5)=\{6\}$ and $\X^*_2(0.5)=\{1,5\}$ and thus the over-demand is either $1$ or $-3$. Thus, the price $p=0.5$ is the unique minimizer of $W$ and $p=0.5$ also fulfills constraint~\eqref{eq:app:corollary_constraint} (see \Cref{foot:PrecisFormulationConstraint}). Therefore, $p=0.5$ is both the unique solution to the constraint problem~\eqref{eqs:ConstraintWmin} \emph{and} the unique solution to the unconstrained problem~\eqref{eq:clearing_objective}, which makes the two optimization problems \eqref{eqs:ConstraintWmin} and \eqref{eq:clearing_objective} mathematically equivalent in this case.
    
    However in practice, any GD-based algorithm will almost never be able to \emph{exactly} compute $p=0.5$. If we run the symmetric version of \Cref{alg:Constrained_W_minimization} (see \Cref{rem:unconstrainedNextPrice}) we would either end up with a price $p$ slightly below $0.5$ or slightly above.\footnote{%
    In this case we would even end up more likely with a price below $0.5$, because than we only have an over-demand of $1$ rather than an under-demand of $3$ items. In every step, where $p^t$ is slightly below $0.5$, GD will push up the price by only $1\gamma$, while in the steps where $p^t$ is slightly above $0.5$, GD will push down the price by $3\gamma$. The probability of \emph{exactly} reaching $p=0.5$ is zero if we initialize $p^0$ with a continuous distribution (as we do in \Cref{alg_constrianed_w:line_Initial_p}). Also \Cref{alg:constrianed_w:line_UseUnconstraintSolutionInstead} is more likely to pick a price slightly below $0.5$, since for every $\epsilon \in [0,0.5)$, $W(0.5+\epsilon,v)=W(0.5,v)+3\epsilon$ and $W(0.5-\epsilon,v)=W(0.5,v)+1\epsilon$ (see \Cref{thm:GD_on_W}). I.e., \Cref{alg:constrianed_w:line_UseUnconstraintSolutionInstead} would clearly prefer $p=0.5-\epsilon$ over $p=0.5+\epsilon$.} However, our asymmetric \Cref{alg:Constrained_W_minimization} makes sure that we do not end up with a price $p$ below $0.5$, which would result in over-demand, but rather at a price $p$ slightly above $0.5$, which results at the feasible allocation where the first bidder gets $6$ items and the second bidder gets $1$ item. In this case, this would be the efficient allocation with a SCW of $9$. If we had only asked demand queries for prices $p\in(0,0.5)$ then the WDP would end up with a SCW of only $6$ by allocating $6$ items to first bidder and $0$ items to the second bidder, since the first bidder would answer any of these demand queries with $6$ items and the second bidder would answer them all with $5$ items, which cannot be combined by the WDP for a capacity of $c=10<11$.
\end{example}

\subsection{Detailed Experimental Results}\label{subsec:Details Results}

\begin{table*}[ht]
    \renewcommand\arraystretch{1.2}
    \setlength\tabcolsep{2pt}
	\robustify\bfseries
	\centering
	\begin{sc}
	\resizebox{1\textwidth}{!}{
	\small
    \begin{tabular}{lcccccccc}
    \toprule
    &  \multicolumn{4}{c}{\textbf{GSVM}}  &\multicolumn{4}{c}{\textbf{LSVM}}\\
        \cmidrule(l{2pt}r{2pt}){2-5}
        \cmidrule(l{2pt}r{2pt}){6-9}
     \textbf{Mechanism}&\textbf{E\textsubscript{clock}} & \textbf{E\textsubscript{raise}} & \textbf{E\textsubscript{profit}} & \textbf{Clear} & \textbf{E\textsubscript{clock}} & \textbf{E\textsubscript{raise}} & \textbf{E\textsubscript{profit}} & \textbf{Clear}\\ 
    \midrule
                                  \multirow{1}{*}{\textbf{ML-CCA}} & \ccell (97.38\,,\,98.23\,,\,98.91) & \ccell(98.51\,,\,98.93\,,\,99.30) & \ccell(100.00\,,\,100.00\,,\,100.00) & \ccell56 & \ccell(89.78\,,\,91.64\,,\,93.34) &  \ccell(95.56,96.39,97.16) & \ccell (99.90\,,\,99.95\,,\,99.99) & \ccell26\\
    \midrule
                      \textbf{CCA} & (88.89\,,\,90.40\,,\,91.84) & (92.64\,,\,93.59\,,\,94.49) & \ccell(99.99\,,\,100.00\,,\,100.00) & 3 & (80.94\,,\,82.56\,,\,84.11) & (90.64\,,\,91.60\,,\,92.54) & (99.55\,,\,99.76\,,\,99.91)  & 0\\

        \midrule
                      \textbf{$p$-value}   & 4.2e-18 & 6.5e-20 & 0.14 & 3.19e-18 & 1.8e-17 & 1.6e-14 & 0.0101  & 2.57e-8\\
    \bottomrule
    \end{tabular}
}
    \end{sc}
    \vskip -0.1 in
    \caption{Detailed results for the GSVM and LSVM domains of ML-CCA vs CCA including the lower and upper 95\%-bootstrapped CI bounds over a test set of $100$ instances of the following metrics: \textsc{\tablecaptionbf{E\textsubscript{clock}}} = efficiency in \% for clock bids, \textsc{\tablecaptionbf{E\textsubscript{raise}}} = efficiency in \% for raised clock bids, \textsc{\tablecaptionbf{E\textsubscript{profit}}} = efficiency in \% for raised clock bids and $100$ profit-max demand queries, \textsc{\tablecaptionbf{Clear}} = percentage of instances where linear clearing prices were found in the clock phase. Winners based on a paired t-test with $\alpha=5\%$ are marked in grey. The $p$-value for this pairwise t-test with $\mathcal{H}_0: \mu_{\text{ML-CCA}}\leq \mu_{\text{CCA}}$ shows at which significance level we can reject the null hypothesis of CCA having a higher or equal average value in the corresponding metric than ML-CCA.}
\label{tab:details_GSVM_LSVM_efficiency_loss_mlca}
    \vskip -0.3cm
\end{table*}
\begin{table*}[ht]
    \renewcommand\arraystretch{1.2}
    \setlength\tabcolsep{2pt}
	\robustify\bfseries
	\centering
	\begin{sc}
	\resizebox{1\textwidth}{!}{
	\small
    \begin{tabular}{lcccccccc}
    \toprule
    &  \multicolumn{4}{c}{\textbf{SRVM}}  &\multicolumn{4}{c}{\textbf{MRVM}}\\
        \cmidrule(l{2pt}r{2pt}){2-5}
        \cmidrule(l{2pt}r{2pt}){6-9}
     \textbf{Mechanism}&\textbf{E\textsubscript{clock}} & \textbf{E\textsubscript{raise}} & \textbf{E\textsubscript{profit}} & \textbf{Clear} & \textbf{E\textsubscript{clock}} & \textbf{E\textsubscript{raise}} & \textbf{E\textsubscript{profit}} & \textbf{Clear}\\ 
    \midrule
                                  \multirow{1}{*}{\textbf{ML-CCA}} & \ccell (99.48\,,\,99.59\,,\,99.68) & \ccell(99.92\,,\,99.93\,,\,99.95) & \ccell(100.00\,,\,100.00\,,\,100.00) & \ccell 13 & \ccell(92.79\,,\,93.04\,,\,93.28)& \ccell(93.09\,,\,93.31\,,\,93.52) & \ccell(93.46\,,\,93.68\,,\,93.89) & \ccell0 \\
    \midrule
                      \textbf{CCA}  & \ccell(99.52\,,\,99.63\,,\,99.73) & (99.75\,,\,99.81\,,\,99.86) & \ccell(100.00\,,\,100.00\,,\,100.00) & 8 & (91.87\,,\,92.44\,,\,92.86) & 
                      (92.25\,,\,92.62\,,\,92.96) & (92.84\,,\,93.18\,,\,93.48)  & \ccell0\\
    \midrule
                      \textbf{$p$-value}  & 0.78 & 2.1e-5 & - & 0.0122 & 7.1e-3 & 1.6e-5 & 3.0e-4  & - \\
    \bottomrule
    \end{tabular}
}
    \end{sc}
    \vskip -0.1 in
    \caption{Detailed results for the SRVM and MRVM domains of ML-CCA vs CCA including the lower and upper 95\%-bootstrapped CI bounds over a test set of $100$ instances of the following metrics: \textsc{\tablecaptionbf{E\textsubscript{clock}}} = efficiency in \% for clock bids, \textsc{\tablecaptionbf{E\textsubscript{raise}}} = efficiency in \% for raised clock bids, \textsc{\tablecaptionbf{E\textsubscript{profit}}} = efficiency in \% for raised clock bids and $100$ profit-max demand queries, \textsc{\tablecaptionbf{Clear}} = percentage of instances where linear clearing prices were found in the clock phase. Winners based on a paired t-test with $\alpha=5\%$ are marked in grey. The $p$-value for this pairwise t-test with $\mathcal{H}_0: \mu_{\text{ML-CCA}}\leq \mu_{\text{CCA}}$ shows at which significance level we can reject the null hypothesis of CCA having a higher or equal average value for the corresponding metric than ML-CCA.}
\label{tab:details_SRVM_MRVM_efficiency_loss_mlca}
    \vskip -0.3cm
\end{table*}

\subsubsection{Detailed Efficiency Results}
In \Cref{tab:details_GSVM_LSVM_efficiency_loss_mlca,tab:details_SRVM_MRVM_efficiency_loss_mlca}, we provide the detailed efficiency results corresponding to \Cref{tab:efficiency_loss_mlca} including 95\%-bootstrapped CIs and $p$-values. Concretely, we now present triplets which show the lower bound of the bootstrapped 95\%-CI, the mean, and the upper bound of the bootstrapped 95\%-CI (e.g., (97.05 , 97.87 , 98.56) means that the lower bound of the bootstrapped 95\%-CI is equal to 97.05, the mean is equal to 97.87, and the upper bound of the bootstrapped 95\%-CI is equal to 98.56). Those bootstrapped CIs were created with the \href{https://docs.scipy.org/doc/scipy/reference/generated/scipy.stats.bootstrap.html}{percentile method} and 10.000 bootstrap-samples. 
It is noteworthy that in all domains other than SRVM (which is very easy, and can be solved by both mechanisms), our ML-CCA outperforms the CCA, both for the clock bids and the clock bids raised, and we can reject the null hypothesis of ML-CCA not outperforming the CCA in terms of  efficiency at great confidence levels that, based on the domain, vary from less than $2$\% all the way to less than $7\cdot10^{-18}\text{\%}=7\mathrm{e}{-20}$.

For GSVM and LSVM, for both practical bidding heuristics (clock bids and clock bids raised) the improvement of our ML-CCA over the CCA is highly significant with all $p$-values being below $2\mathrm{e}{-14}$.
For SRVM, the differences between the methods seem to be small, since both methods almost reach $100\%$ efficiency. However, for clock bids raised this improvement is clearly statistically significant with a $p$-value of $0.0021$\%, while for clock bids there is actually no statistically significant difference. Note that for 7 out of all 8 practical settings (4 domains with two practical bidding heuristics) the $p$-value is below $1$\%. For clock bids raised all four domains have a $p$-value below $0.0021$\%.
For MRVM, for all three bidding heuristics our ML-CCA is significantly better than the CCA with $p$-values $0.71$\%, $1.6\mathrm{e}{-5}$ and $3\mathrm{e}{-4}$.

In all three domains where LCPs exist, our ML-CCA found them statistically significantly more often than the CCA with $p$-values $3.2\mathrm{e}{-18}$, $2.6\mathrm{e}{-8}$ and $1.22$\%.

When reading the efficiency results in \Cref{tab:details_GSVM_LSVM_efficiency_loss_mlca,tab:details_SRVM_MRVM_efficiency_loss_mlca} one should keep in mind that the CCA has generated over $\$ 20$ Billion in revenue for spectrum allocations between $2012$ and $2014$ alone \citep{ausubel2017practical}. Since revenue is a lower bound for SCW, improving CCA's efficiency on average by $1$\% point would have improved the SCW by more than $\$200$ Million within this time range alone and CCA is still the most prominent practical mechanism for spectrum allocation. 

Instead of efficiency~$\frac{V(a)}{V(a^*)}$, one could also study the efficiency loss~$\frac{V(a^*)-V(a)}{V(a^*)}=1-\frac{V(a)}{V(a^*)}$, which corresponds to the relative cost of deviating from the efficient allocation in terms of social welfare. For GSVM with clock bids only our ML-CCA can cut down the efficiency loss of the CCA by a factor $5.4$ from $9.6$\% to $1.77$\%. Similarly for GSVM with clock bids raised our method can cut down the efficiency loss by a factor $5.9$ from $6.41$\% to $1.07$\%.

\subsubsection{Path Plots of Profit-Max Bids}
In \Cref{fig:profit_max}, we show the effect of adding up to $\QPmax=100$ bids in the supplementary round of both our ML-CCA mechanism as well as the CCA using the profit-max heuristic.
In GSVM, both mechanisms can reach $100$\% efficiency using those profit-max bids. However, ML-CCA's clock phase can do so after only $18$ profit-max bids, while the CCA requires $44$. 
In LSVM, for any number of profit max bids, our mechanism exhibits higher efficiency than the CCA, while with $100$ profit-max bids we can reach $99.95$\% efficiency as opposed to $99.76$\% for the CCA.
In SRVM, both mechanisms can reach over $99.99$\% efficiency using only $4$ profit max bids. In MRVM, we can see that for almost any number of profit-max bids, our ML-CCA outperforms the CCA and the results are statistically significant on the $95$\% CI level.\footnote{Note that for example for 100 profit-max bids, the $95$\% CIs slightly overlap, while the corresponding $p$-value of the paired test is $0.03\text{\%}=3\mathrm{e}{-4}$. This suggest that a paired test would show the statistical significance of our ML-CCA having a higher average efficiency than the CCA for probably any number of profit-max bids. Note that a paired test is the correct statistical test for such situations, since both ML-CCA and the CCA were evaluated on the \emph{same} 100 SATS-seeds.} Furthermore, for MRVM, it is interesting to note that the CCA, even with $100$ profit max bids per agent, cannot reach the clock bids raised efficiency of our ML-CCA (i.e., the efficiency with $0$ profit max bids), while it needs $38$ profit max bids to reach the efficiency that the clock phase of our ML-CCA achieves. In other words, the CCA requires up to an additional $138$ value bids per agent to achieve the same efficiency that our ML-CCA can achieve using only $100$ clock bids. 

\begin{figure}[t!]
    \begin{center}
\includegraphics[width=1\columnwidth,trim=0 0 0 0, clip]{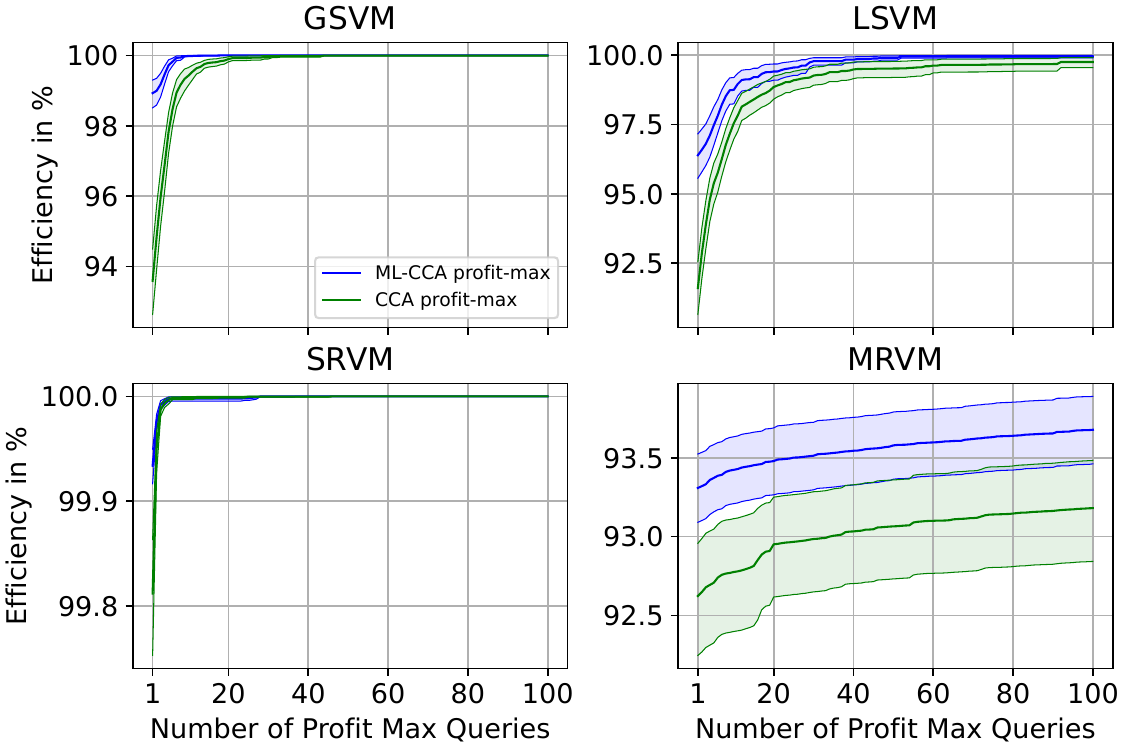}
    \caption{Efficiency of adding additionally $\QPmax=100$ profit-max bids in SATS for ML-CCA and CCA after 100 clock bids and 100 raised clock bids. Averaged over 100 runs including a bootstrapped 95\% CI.}
    \label{fig:profit_max}
    \end{center}
    \vskip -0.25cm
\end{figure}

\subsubsection{Path Plots of Clearing Error and Linear Item Prices}
In \Cref{fig:GSVM_CE_and_LPs,fig:LSVM_CE_and_LPs,fig:SRVM_CE_and_LPs_logscale,fig:MRVM_CE_and_LPs}, we present for all domains in SATS the path plots for clock rounds $r=0,\ldots,100$ of the (squared) \emph{clearing error} (CE) defined as:
\begin{align}\label{eq:clearing_error}
\sum_{j=1}^{m} \left (\left(\sum_{i=1}^{n}x_i^*(p^{r})_j - c_j \right )^2\right),
\end{align}
and the \emph{linear item prices} $p_j^r\in \Rpz$ for $j=1,\ldots,m$.
\begin{figure}[t!]
    \begin{center}
\includegraphics[width=1\columnwidth,trim=0 0 0 0, clip]{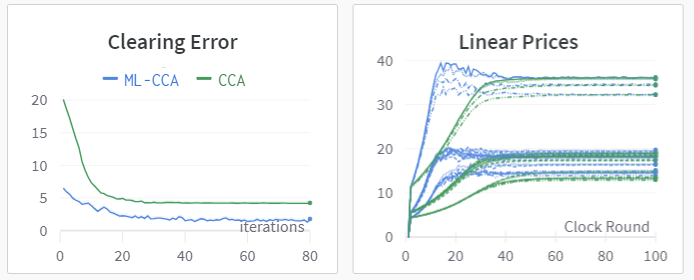}
    \caption{\textbf{Left:} CE of ML-CCA and CCA in GSVM per iteration (starting after the $\Qinit$-phase, i.e., in clock round 20) averaged over 100 runs including the standard error. \textbf{Right:} \emph{linear item  prices} $p_j^r\in \Rpz$ for $j=1,\ldots,18$ defining the demand query for each clock round $r=0,\ldots,100$ averaged over 100 runs.}
    \label{fig:GSVM_CE_and_LPs}
    \end{center}
    \vskip -0.25cm
\end{figure}

In \Cref{fig:GSVM_CE_and_LPs}, we present the results for GSVM. We observe that, for any iteration after the initial $\Qinit=20$ clock rounds, i.e., any clock round $r=20,\ldots,100$, the CE of our ML-CCA is smaller than the CE resulting from CCA. Specifically, ML-CCA has already at iteration 0 (clock round $20)$ a small CE ($\approx$ 7) whilst CCA needs approximately 20 more iterations to reach a similar CE. 

Moreover, from the path plots of the prices we can distinguish the two phases of our proposed ML-CCA: the initial CCA phase with a predefined fixed price increment in case of over-demand for the first $20$ clock rounds (or $50$ in the case of MRVM), and the ML-powered demand query generation phase starting with the $21\textsuperscript{th}$ (or $51\textsuperscript{th}$ for MRVM) clock round. We observe that our approach, for the price of each item, is immediately searching in a local neighbourhood around the plateaued CCA price of that item, and tries to clear the market by slightly increasing and decreasing certain prices. Finally, we can see that CCA properly plateaus to final CCA prices around clock rounds $80-100$, where no item is over-demanded anymore.

In \Cref{fig:LSVM_CE_and_LPs}, we present the results for LSVM. We can see that for any iteration after the initial $\Qinit=20$ clock rounds, i.e., any clock round $r=20,\ldots,100$, the CE of our ML-CCA is significantly smaller than the CE of CCA. Specifically, CCA requires $55$ iterations (i.e., 75 clock rounds) 
to reach a CE comparable to what ML-CCA can achieve in the first iterations.\footnote{Note that by increasing the price increment for CCA, one could reduce the number of iterations until it achieves a low CE, but that could result in a significant drop in the efficiency of the auction, see \Cref{subsec:app:Results_reduced_Qmax}.} 
The price path plots in LSVM display a similar picture as in GSVM. We see that ML-CCA immediately identifies the correct region after the $\Qinit=20$ initial CCA prices and then tries to locally search by increasing and decreasing prices around the plateaued CCA prices, without at the same time sacrificing efficiency, as would be the case if we were to increase the price increment of CCA to reach that price area faster.  

\begin{figure}[t!]
    \begin{center}
\includegraphics[width=1\columnwidth,trim=0 0 0 0, clip]{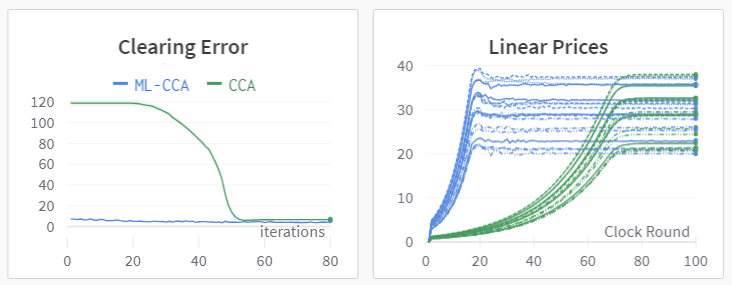}
    \caption{\textbf{Left:} CE of ML-CCA and CCA in LSVM per iteration (starting after the $\Qinit$-phase, i.e., in clock round 20) averaged over 100 runs including the standard error. \textbf{Right:} \emph{linear item  prices} $p_j^r\in \Rpz$ for $j=1,\ldots,18$ defining the demand query for each clock round $r=0,\ldots,100$ averaged over 100 runs.}
    \label{fig:LSVM_CE_and_LPs}
    \end{center}
    \vskip -0.25cm
\end{figure}

In \Cref{fig:SRVM_CE_and_LPs_logscale}, we present the results for SRVM on a log scale. We can see again that, for any iteration after the initial $\Qinit=20$ clock rounds, i.e., any clock round $r=20,\ldots,100$, the CE of our ML-CCA is smaller than that of CCA. 
Furthermore, we can see that in less than $20$ iterations, our ML-CCA is able to drop the CE down to $1.1$ (with $1$ being a lower bound on the CE when clearing prices do not exist), while the CCA can never reach those numbers.

Recall, that in SRVM there are only three distinct items (i.e., $m=3$) with quantities $c_1=6,c_2=14$, and $c_3=9$. We again see the same behaviour of ML-CCA's price discovery mechanism as in the other SATS domains: ML-CCA is able to immediately identify prices for the three items that are close to the final plateaued CCA prices.

\begin{figure}[t!]
    \begin{center}
\includegraphics[width=1\columnwidth,trim=0 0 0 0, clip]{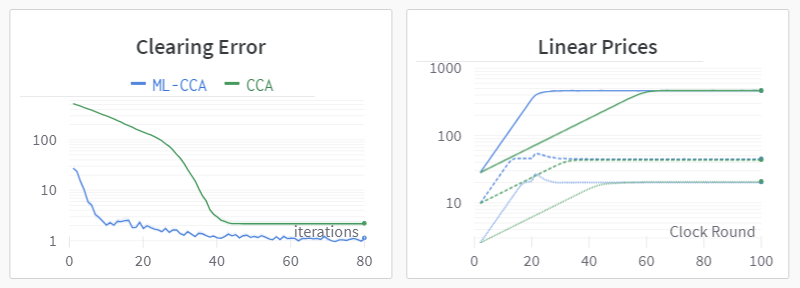}
    \caption{\textbf{Left:} CE of ML-CCA and CCA in SRVM per iteration (starting after the $\Qinit$-phase, i.e., in clock round 20) averaged over 100 runs including the standard error. \textbf{Right:} \emph{linear item  prices} $p_j^r\in \Rpz$ for $j=1,\ldots,3$ defining the demand query for each clock round $r=0,\ldots,100$ averaged over 100 runs. Both $y$-axes are on a log scale.}
    \label{fig:SRVM_CE_and_LPs_logscale}
    \end{center}
    \vskip -0.25cm
\end{figure}

In \Cref{fig:MRVM_CE_and_LPs}, we present the results for MRVM. Interestingly, in this domain the CCA reaches a lower CE than our ML-CCA. However, please note that CE, as defined in \Cref{eq:clearing_error}, penalizes all goods equally, which can be an uninformative metric in domains where the values of different goods vary significantly. 
This is most pronounced in the MRVM domain, where the prices of the most valuable items are more than $100$ times larger than the prices of the least valuable goods , as shown in the right part \Cref{fig:MRVM_CE_and_LPs}. 

\begin{figure}[h!]
    \begin{center}
\includegraphics[width=1\columnwidth,trim=0 0 0 0, clip]{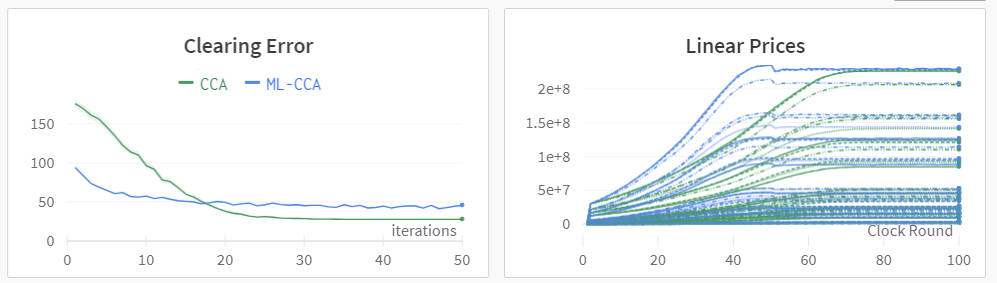}
    \caption{\textbf{Left:} CE of ML-CCA and CCA in MRVM per iteration (starting after the $\Qinit$-phase, i.e., in clock round 50) averaged over 100 runs including the standard error. \textbf{Right:} \emph{linear item  prices} $p_j^r\in \Rpz$ for $j=1,\ldots,42$ defining the demand query for each clock round $r=0,\ldots,100$ averaged over 100 runs.}
    \label{fig:MRVM_CE_and_LPs}
    \end{center}
    \vskip -0.25cm
\end{figure}

\paragraph{Compute Time}
In \Cref{tab:app:compute_time} we report, for our choice of hyperparameters given in \Cref{tab:appenedix_hpo_winners}, 
the average time per round required in each domain to train the $N$ mMVNNs of the bidders according to our \Cref{alg:train_on_dqs} and the time required to generate the next price vector, given the trained MVNNs, using our \Cref{alg:Constrained_W_minimization}. 
The sum of those two numbers is the average time per round required by our ML-CCA per domain to generate the next demand query, once the bidders have responded to the current one. 
For this experiment, we report average results over the same $100$ instances as all other experiments in this paper. 

\begin{table}[h!]
    \centering
    \begin{tabular}{rrrr}
        \toprule
Domain & Train Time & Price Gen. Time & Total  \\
\midrule 
GSVM  &  4.49 & 1.81 &  6.30\\ 
LSVM  &  2.60 & 1.19 &  3.79\\ 
SRVM  &  1.26 & 0.41 &  1.67\\ 
MRVM  &  32.22 & 11.19 &  43.41 \\ 
\bottomrule
\end{tabular}
\caption{Detailed results of the average time required (in minutes) for ML-CCA per round to train the mMVNNs of all bidders, and to generate the price vector of the next demand query, given the trained mMVNNs. Shown are average results over $100$ instances.
}
\label{tab:app:compute_time}
\end{table}

\subsection{Experimental Results for Reduced \texorpdfstring{$\Qmax=50$}{Qmax=50}}\label{subsec:app:Results_reduced_Qmax}
\begin{table*}[ht]
    \renewcommand\arraystretch{1.2}
    \setlength\tabcolsep{2pt}
	\robustify\bfseries
	\centering
	\begin{sc}
	\resizebox{1\textwidth}{!}{
	\small
    \begin{tabular}{lcccccccc}
    \toprule
    &  \multicolumn{4}{c}{\textbf{GSVM}}  &\multicolumn{4}{c}{\textbf{LSVM}}\\
        \cmidrule(l{2pt}r{2pt}){2-5}
        \cmidrule(l{2pt}r{2pt}){6-9}
     \textbf{Mechanism}&\textbf{E\textsubscript{clock}} & \textbf{E\textsubscript{raise}} & \textbf{E\textsubscript{profit}} & \textbf{Clear} & \textbf{E\textsubscript{clock}} & \textbf{E\textsubscript{raise}} & \textbf{E\textsubscript{profit}} & \textbf{Clear}\\ 
    \midrule
                                  \multirow{1}{*}{\textbf{ML-CCA} ($\Qmax = 50$)}  & \ccell (96.97\,,\,97.84\,,\,98.58) & \ccell(98.10\,,\,98.59\,,\,99.03) & \ccell(100.00\,,\,100.00\,,\,100.00) & \ccell 51 & \ccell (88.99\,,\,90.81\,,\,92.50) &  \ccell (95.07\,,\,95.91\,,\,96.70) & \ccell (99.90\,,\,99.95\,,\,99.99) & \ccell 17\\
    \midrule
    \textbf{CCA} ($\Qmax = 50$)  & (89.05\,,\,90.51\,,\,91.88) & (92.74\,,\,93.70\,,\,94.62) & \ccell(99.99\,,\,100.00\,,\,100.00) & 1 & (80.57\,,\,82.21\,,\,83.77) & (90.59\,,\,91.52\,,\,92.46) & (99.64\,,\,99.80\,,\,99.93)  & 0\\
     \textbf{CCA} ($\Qmax = 100$)  & (88.89\,,\,90.40\,,\,91.84) & (92.64\,,\,93.59\,,\,94.49) & (99.99\,,\,100.00\,,\,100.00) & 3 & (80.94\,,\,82.56\,,\,84.11) & (90.64\,,\,91.60\,,\,92.54) &  (99.55\,,\,99.76\,,\,99.91)  & 0\\  \midrule
\textbf{$p$-value} ($\Qmax = 50$)   & 4.2e-17 & 4.2e-18 & 0.1433 & 7.0e-17 & 2.5e-16 & 3.1
e-13 & 0.0145  & 9.1e-6\\ 
                      
    \bottomrule
    \end{tabular}
}
    \end{sc}
    \vskip -0.1 in
    \caption{Detailed results for the GSVM and LSVM domains of ML-CCA vs CCA for $\Qmax = 50$ including the lower and upper 95\%-bootstrapped CI bounds over a test set of $100$ instances of the following metrics: \textsc{\tablecaptionbf{E\textsubscript{clock}}} = efficiency in \% for clock bids, \textsc{\tablecaptionbf{E\textsubscript{raise}}} = efficiency in \% for raised clock bids, \textsc{\tablecaptionbf{E\textsubscript{profit}}} = efficiency in \% for raised clock bids and $100$ profit-max demand queries, \textsc{\tablecaptionbf{Clear}} = percentage of instances where clearing prices were found in the clock phase. Winners for $\Qmax=50$ based on a paired t-test with $\alpha=5\%$ are marked in grey. The $p$-value for this pairwise t-test with $\mathcal{H}_0: \mu_{\text{ML-CCA}}\leq \mu_{\text{CCA}}$ shows at which significance level we can reject the null hypothesis of CCA with $\Qmax=50$ having a higher or equal average value in the corresponding metric than ML-CCA with $\Qmax=50$. Additionally, we also reprint the CCA ($\Qmax=100$) results from \Cref{tab:details_GSVM_LSVM_efficiency_loss_mlca} without marking statistical significance.}
\label{tab:details_GSVM_LSVM_efficiency_loss_mlca_Qmax50}
    \vskip -0.3cm
\end{table*}
\begin{table*}[ht]
    \renewcommand\arraystretch{1.2}
    \setlength\tabcolsep{2pt}
	\robustify\bfseries
	\centering
	\begin{sc}
	\resizebox{1\textwidth}{!}{
	\small
    \begin{tabular}{lcccccccc}
    \toprule
    &  \multicolumn{4}{c}{\textbf{SRVM}}  &\multicolumn{4}{c}{\textbf{MRVM}}\\
        \cmidrule(l{2pt}r{2pt}){2-5}
        \cmidrule(l{2pt}r{2pt}){6-9}
     \textbf{Mechanism}&\textbf{E\textsubscript{clock}} & \textbf{E\textsubscript{raise}} & \textbf{E\textsubscript{profit}} & \textbf{Clear} & \textbf{E\textsubscript{clock}} & \textbf{E\textsubscript{raise}} & \textbf{E\textsubscript{profit}} & \textbf{Clear}\\ 
    \midrule
                                  \multirow{1}{*}{\textbf{ML-CCA} ($\Qmax = 50$)} & \ccell (99.47\,,\,99.58\,,\,99.67) & \ccell(99.92\,,\,99.93\,,\,99.95) & \ccell(100.00\,,\,100.00\,,\,100.00) & \ccell 13 & \ccell(91.63\,,\,92.11\,,\,92.50)& \ccell(91.94\,,\,92.42\,,\,92.81) & \ccell(92.73\,,\,93.04\,,\,93.34) & \ccell0 \\
    \midrule
    \textbf{CCA} ($\Qmax = 50$)  & \ccell(99.36\,,\,99.47\,,\,99.59) & (99.64\,,\,99.72\,,\,99.79) & \ccell(100.00\,,\,100.00\,,\,100.00) & 4 & (87.70\,,\,88.70\,,\,89.61) & 
                      (88.52\,,\,89.28\,,\,90.02) & (89.57\,,\,90.27\,,\,90.92)  & \ccell0\\
    \textbf{CCA} ($\Qmax = 100$)  & (99.52\,,\,99.63\,,\,99.73) & (99.75\,,\,99.81\,,\,99.86) & (100.00\,,\,100.00\,,\,100.00) & 8 & (91.87\,,\,92.44\,,\,92.86) & 
                       (92.25\,,\,92.62\,,\,92.96) & (92.84\,,\,93.18\,,\,93.48)  & 0\\
    \midrule
              \textbf{$p$-value} ($\Qmax = 50$)   & 0.088 & 2.9e-7 & - & 0.0011 & 2.4e-10 & 1.3e-12 & 1.6e-14  & -\\
                      
    \bottomrule
    \end{tabular}
}
    \end{sc}
    \vskip -0.1 in
    \caption{Detailed results for the SRVM and MRVM domains of ML-CCA vs CCA for $\Qmax = 50$ including the lower and upper 95\%-bootstrapped CI bounds over a test set of $100$ instances of the following metrics: \textsc{\tablecaptionbf{E\textsubscript{clock}}} = efficiency in \% for clock bids, \textsc{\tablecaptionbf{E\textsubscript{raise}}} = efficiency in \% for raised clock bids, \textsc{\tablecaptionbf{E\textsubscript{profit}}} = efficiency in \% for raised clock bids and $100$ profit-max demand queries, \textsc{\tablecaptionbf{Clear}} = percentage of instances where clearing prices were found in the clock phase. Winners for $\Qmax=50$ based on a paired t-test with $\alpha=5\%$ are marked in grey. The $p$-value for this pairwise t-test with $\mathcal{H}_0: \mu_{\text{ML-CCA}}\leq \mu_{\text{CCA}}$ shows at which significance level we can reject the null hypothesis of CCA with $\Qmax=50$ having a higher or equal average value in the corresponding metric than ML-CCA with $\Qmax=50$. In all domains (including MRVM) we used $\Qinit=20$ for this experiment. Additionally, we also reprint the CCA ($\Qmax=100$) results from \Cref{tab:details_SRVM_MRVM_efficiency_loss_mlca} without marking statistical significance.} 
\label{tab:details_SRVM_MRVM_efficiency_loss_mlca_Qmax50}
    \vskip -0.3cm
\end{table*}

In this section, we present all results of ML-CCA\footnote{For ML-CCA, we used the same HPs for all domains as in \Cref{sec:experiments}, other than MRVM, were we also set $\Qinit=20$.} and CCA for a reduced number of $\Qmax=50$ demand queries instead of $\Qmax=100$, which were presented in the main paper in \Cref{sec:experiments}.  For ease of exposition we also include the CCA optimized for $\Qmax=100$ as in \Cref{sec:experiments}.
We additionally present $95$\% CIs and a paired t-test with $\alpha=5\%$. The $p$-value for this pairwise t-test with $\mathcal{H}_0: \mu_{\text{ML-CCA}}\leq \mu_{\text{CCA}}$ shows at which significance level we can reject the null hypothesis of CCA with $\Qmax=50$ having a higher or equal average value in the corresponding metric than the ML-CCA with $\Qmax=50$. All results are presented in 
\Cref{tab:details_GSVM_LSVM_efficiency_loss_mlca_Qmax50,tab:details_SRVM_MRVM_efficiency_loss_mlca_Qmax50}. 

Just like for $\Qmax = 100$, our ML-CCA outperforms the CCA considerably in all domains. Specifically, the efficiency improvement after the clock phase is over $7.3$\% points for GSVM, $8.6$\% points in LSVM and $3.4$\%  points in MRVM, and all of those improvements are highly statistically significant, based on a paired t-test with $\alpha=2.4\mathrm{e}{-10}$.
The two mechanisms are statistically tied in SRVM (other than for raised clock bids, where our ML-CCA again out performs the CCA). As pointed out in \Cref{sec:experiments}, this domain is very easy, and both mechanisms can solve it even with 50 clock rounds). 
If we add the clock bids raised heuristic of the supplementary round to both mechanisms, the efficiency improvement of ML-CCA is $4.9$\% points in GSVM, $4.4$\% points in LSVM and $3.1$\% points in MRVM%
. Again all those improvements are highly significant with $\alpha=2.9\mathrm{e}{-7}$, this time also for SRVM and thus for all fur domains.
Finally, it is noteworthy that in the GSVM and LSVM domains, our ML-CCA with $50$ clock rounds can achieve significantly higher efficiency with $50$ clock rounds compared to the CCA with $100$. These results even persist if we add the clock bids raised heuristic of the supplementary round to both mechanisms, which would induce up to an additional $100$ value queries for each bidder in the CCA, and only $50$ value queries for our ML-CCA.

\begin{table*}[ht]
    \renewcommand\arraystretch{1.2}
    \setlength\tabcolsep{2pt}
	\robustify\bfseries
	\centering
	\begin{sc}
	\resizebox{1\textwidth}{!}{
	\small
    \begin{tabular}{lcccccccc}
    \toprule
    &  \multicolumn{4}{c}{\textbf{GSVM}}  &\multicolumn{4}{c}{\textbf{LSVM}}\\
        \cmidrule(l{2pt}r{2pt}){2-5}
        \cmidrule(l{2pt}r{2pt}){6-9}
     \textbf{Mechanism}&\textbf{E\textsubscript{clock}} & \textbf{E\textsubscript{raise}} & \textbf{E\textsubscript{profit}} & \textbf{Clear} & \textbf{E\textsubscript{clock}} & \textbf{E\textsubscript{raise}} & \textbf{E\textsubscript{profit}} & \textbf{Clear}\\ 
    \midrule
    \multirow{1}{*}{\textbf{ML-CCA-C}} & \ccell (97.38\,,\,98.23\,,\,98.91) & \ccell(98.51\,,\,98.93\,,\,99.30) & \ccell(100.00\,,\,100.00\,,\,100.00) & \ccell56 & \ccell(89.78\,,\,91.64\,,\,93.34) &  \ccell(95.56\,,\,96.39\,,\,97.16) & \ccell (99.90\,,\,99.95\,,\,99.99)  & \ccell26\\ \midrule
    \multirow{1}{*}{\textbf{ML-CCA-U}}  &  (97.05\,,\,97.87\,,\,98.56) & \ccell(98.22\,,\,98.67\,,\,99.07) & \ccell(100.00\,,\,100.00\,,\,100.00) & 51 & \ccell(89.84\,,\,91.60\,,\,93.28) &  (95.33\,,\,96.16\,,\,96.94) & \ccell(99.90\,,\,99.95\,,\,99.99)  &  23\\

    \midrule
              \textbf{$p$-value}   & 0.0414 & 0.0965 & - & 0.0122 & 0.3428 & 0.0189 & -  & 0.0416\\
                                  
    \bottomrule
    \end{tabular}
}
    \end{sc}
    \vskip -0.1 in
    \caption{
    ML-CCA with constrained $W$ (ML-CCA-C) and unconstrained $W$ minimization (ML-CCA-U). Shown are averages including the lower and upper 95\%-bootstrapped CI bounds over a test set of $100$ synthetic CA instances for the GSVM and LSVM domains of the following metrics: 
    efficiency in \% for clock bids (\textsc{\tablecaptionbf{E\textsubscript{clock}}}), raised clock bids (\textsc{\tablecaptionbf{E\textsubscript{raise}}}) and raised clock bids and $100$ profit-max bids (\textsc{\tablecaptionbf{E\textsubscript{profit}}}) and percentage of instances where linear clearing prices were found (\textsc{\tablecaptionbf{Clear}}).
    Winners based on a paired t-test with $\alpha=5\%$ are marked in grey. The $p$-value for this pairwise t-test with $\mathcal{H}_0: \mu_{\text{ML-CCA-C}}\leq \mu_{\text{ML-CCA-U}}$ shows at which significance level we can reject the null hypothesis of ML-CCA-U having a higher or equal average value in the corresponding metric than ML-CCA-C.}
\label{tab:details_GSVM_LSVM_efficiency_constrained_vs_unconstrained}
    \vskip -0.3cm
\end{table*}

\begin{table*}[t!]
    \renewcommand\arraystretch{1.2}
    \setlength\tabcolsep{2pt}
	\robustify\bfseries
	\centering
	\begin{sc}
	\resizebox{1\textwidth}{!}{
	\small
    \begin{tabular}{lcccccccc}
    \toprule
    &  \multicolumn{4}{c}{\textbf{SRVM}}  &\multicolumn{4}{c}{\textbf{MRVM}}\\
        \cmidrule(l{2pt}r{2pt}){2-5}
        \cmidrule(l{2pt}r{2pt}){6-9}
     \textbf{Mechanism}&\textbf{E\textsubscript{clock}} & \textbf{E\textsubscript{raise}} & \textbf{E\textsubscript{profit}} & \textbf{Clear} & \textbf{E\textsubscript{clock}} & \textbf{E\textsubscript{raise}} & \textbf{E\textsubscript{profit}} & \textbf{Clear}\\ 
    \midrule
    \multirow{1}{*}{\textbf{ML-CCA-C}} & \ccell (99.48\,,\,99.59\,,\,99.68) & \ccell(99.92\,,\,99.93\,,\,99.95) & \ccell(100.00\,,\,100.00\,,\,100.00) & \ccell13 & \ccell(92.90\,,\,93.12\,,\,93.34)& \ccell(93.03\,,\,93.25\,,\,93.47) & \ccell(93.46\,,\,93.67\,,\,93.88) & \ccell0 \\ \midrule
                                  \multirow{1}{*}{\textbf{ML-CCA-U}}  & \ccell (99.54\,,\,99.63\,,\,99.70) & (99.91\,,\,99.93\,,\,99.94) & \ccell(100.00\,,\,100.00\,,\,100.00) & \ccell13 & (91.49\,,\,92.06\,,\,92.55)& (91.88\,,\,92.31\,,\,92.70) & (92.27\,,\,92.70\,,\,93.09) & \ccell0 \\

    \midrule
              \textbf{$p$-value} & 0.7793 & 0.0431 & - & - & 9.3e-6 & 2.7e-08 & 1.0e-8  & -\\
                                  
    \bottomrule
    \end{tabular}
}
    \end{sc}
    \vskip -0.1 in
    \caption{ ML-CCA with constrained $W$ (ML-CCA-C) and unconstrained $W$ minimization (ML-CCA-U). Shown are averages including the lower and upper 95\%-bootstrapped CI bounds over a test set of $100$ synthetic CA instances for the SRVM and MRVM domains of the following metrics: 
    efficiency in \% for clock bids (\textsc{\tablecaptionbf{E\textsubscript{clock}}}), raised clock bids (\textsc{\tablecaptionbf{E\textsubscript{raise}}}) and raised clock bids plus $100$ profit-max bids (\textsc{\tablecaptionbf{E\textsubscript{profit}}}) and percentage of instances where linear clearing prices were found (\textsc{\tablecaptionbf{Clear}}).
    Winners based on a paired t-test with $\alpha=5\%$ are marked in grey. The $p$-value for this pairwise t-test with $\mathcal{H}_0: \mu_{\text{ML-CCA-C}}\leq \mu_{\text{ML-CCA-U}}$ shows at which significance level we can reject the null hypothesis of ML-CCA-U having a higher or equal average value in the corresponding metric than ML-CCA-C.}
\label{tab:details_SRVM_MRVM_efficiency_constrained_vs_unconstrained}
    \vskip -0.3cm
\end{table*}

\begin{table*}[t!]
    \renewcommand\arraystretch{1.2}
    \setlength\tabcolsep{2pt}
	\robustify\bfseries
	\centering
	\begin{sc}
	\resizebox{1\textwidth}{!}{
	\small
    \begin{tabular}{lcccccccc}
    \toprule
    &  \multicolumn{4}{c}{\textbf{SRVM}}  &\multicolumn{4}{c}{\textbf{MRVM}}\\
        \cmidrule(l{2pt}r{2pt}){2-5}
        \cmidrule(l{2pt}r{2pt}){6-9}
     \textbf{Mechanism}&\textbf{E\textsubscript{clock}} & \textbf{E\textsubscript{raise}} & \textbf{E\textsubscript{profit}} & \textbf{Clear} & \textbf{E\textsubscript{clock}} & \textbf{E\textsubscript{raise}} & \textbf{E\textsubscript{profit}} & \textbf{Clear}\\ 
    \midrule
    \multirow{1}{*}{\textbf{ML-CCA-mMVNN}} & \ccell (99.48\,,\,99.59\,,\,99.68) & \ccell(99.92\,,\,99.93\,,\,99.95) & \ccell(100.00\,,\,100.00\,,\,100.00) & \ccell13 & (92.90\,,\,93.12\,,\,93.34)& \ccell(93.03\,,\,93.25\,,\,93.47) & \ccell(93.46\,,\,93.67\,,\,93.88) & \ccell0 \\ \midrule
    \multirow{1}{*}{\textbf{ML-CCA-MVNN}}  & \ccell (99.53\,,\,99.63\,,\,99.71) & (99.89\,,\,99.91\,,\,99.93) & \ccell(100.00\,,\,100.00\,,\,100.00) & \ccell 13 & \ccell (92.98\,,\,93.18\,,\,93.37)& 
    \ccell (93.19\,,\,93.39\,,\,93.57) & \ccell (93.48\,,\,93.69\,,\,93.89) & \ccell0 \\

    \midrule
              \textbf{$p$-value} & 0.7401 & 0.0027 & - & - & 0.9700 & 0.8787 & 0.8182  & -\\
                                  
    \bottomrule
    \end{tabular}
}
    \end{sc}
    \vskip -0.1 in
    \caption{ ML-CCA with mMVNNs (ML-CCA-mMVNN) and MVNNs (ML-CCA-MVNN) as the ML model architecture. Shown are averages including the lower and upper 95\%-bootstrapped CI bounds over a test set of $100$ synthetic CA instances for the SRVM and MRVM domains of the following metrics: 
    efficiency in \% for clock bids (\textsc{\tablecaptionbf{E\textsubscript{clock}}}), raised clock bids (\textsc{\tablecaptionbf{E\textsubscript{raise}}}) and raised clock bids plus $100$ profit-max bids (\textsc{\tablecaptionbf{E\textsubscript{profit}}}) and percentage of instances where linear clearing prices were found (\textsc{\tablecaptionbf{Clear}}).
    Winners based on a paired t-test with $\alpha=5\%$ are marked in grey. The $p$-value for this pairwise t-test with $\mathcal{H}_0: \mu_{\text{ML-CCA-mMVNN}}\leq \mu_{\text{ML-CCA-MVNN}}$ shows at which significance level we can reject the null hypothesis of ML-CCA-MVNN having a higher or equal average value in the corresponding metric than ML-CCA-mMVNN.}
\label{tab:details_SRVM_MRVM_efficiency_MVNNs_vs_mMVNNs}
    \vskip -0.3cm
\end{table*}

\subsection{Experimental Results for unconstrained \texorpdfstring{$W$}{W} minimization
}
\label{subsec:app:Results_unconstrained_W}

In this section, we empirically evaluate the importance of constraint~\eqref{eq:app:corollary_constraint} in the \textsc{NextPrice}-procedure by comparing the efficiency results of ML-CCA with constraint~\eqref{eq:app:corollary_constraint} (\emph{ML-CCA-C}) and without it (\emph{ML-CCA-U}).

Constraint~\eqref{eq:app:corollary_constraint} ensures that the predicted induced demand should constitute a feasible allocation for every clock round of ML-CCA-C, as discussed in \Cref{sec:ML-powered Demand Query Generation} and \Cref{subsec:app_constrained_W_minimization}.\footnote{ML-CCA-C is our default method, which we simply call ML-CCA outside of \Cref{subsec:app:Results_unconstrained_W}. In particular, we used \textsc{NextPrice} (\Cref{alg:Constrained_W_minimization}) with default hyper-parameters as discussed in \Cref{subsec:app_constrained_W_minimization}.}
On the other side, \emph{ML-CCA-U} optimizes $W$ without any constraint (i.e., by performing unconstrained classical GD on $W$ as suggested by \Cref{thm:connection_clearing_prices_efficiency} and \Cref{thm:GD_on_W}).\footnote{For ML-CCA-U we ignore constraint~\eqref{eq:app:corollary_constraint}, which corresponds to setting the hyper-parameters $\mu=\nu=0$ in \Cref{alg:Constrained_W_minimization} as described in \Cref{rem:unconstrainedNextPrice}.}

Those results are presented in 
\Cref{tab:details_GSVM_LSVM_efficiency_constrained_vs_unconstrained,tab:details_SRVM_MRVM_efficiency_constrained_vs_unconstrained}. 
In the GSVM, LSVM and SRVM domains, the results for the two approaches are almost identical (while constrained $W$ minimization is statistically better in a few cases). A noteworthy difference is that for the GSVM and LSVM domains, the unconstrained $W$ minimization clears the market in $5$\% and $3$\% less of the cases compared to the constrained version. Overall, as suggested in \Cref{sec:ML-powered Demand Query Generation} and \Cref{subsec:app_constrained_W_minimization}, in domains where 
linear prices can achieve low clearing error (see \Cref{subsec:Details Results}, 
\Cref{fig:GSVM_CE_and_LPs,fig:LSVM_CE_and_LPs,fig:SRVM_CE_and_LPs_logscale})
minimizing $W$ by performing classical GD on it without any additional constraints suffices for significant efficiency improvements compared to the CCA. 

In the MRVM\footnote{For this test in MRVM, we used $\Qinit = 70$ for both ML-CCA-C and ML-CCA-U.} domain however, we can see that the efficiency improvement of performing constrained $W$ minimization compared to the unconstrained one is highly statistically significant across all three bidding heuristics with $p$-values ranging from $1.0\textrm{e}{-8}$ to $9.3\textrm{e}{-6}$. 
Without push bids, the efficiency improvement is approximately $1$\% point. This efficiency improvement also persists if we add the clock bids raised and profit max heuristics for the supplementary round, and is again highly statistically significant.
Thus, we can conclude that in all cases, there is no disadvantage to using the constrained $W$ minimization to generate the next price vector in ML-CCA. This clear improvements of ML-CCA-C over ML-CCA-U fits well our hypothesis stated in \Cref{rem:ConstraintImportant} that constraint~\eqref{eq:app:corollary_constraint} is especially important in cases where no LCPs exist (or when the market exhibits a high clearing error as defined in \Cref{eq:clearing_error}, 
see \Cref{fig:GSVM_CE_and_LPs,fig:LSVM_CE_and_LPs,fig:SRVM_CE_and_LPs_logscale,fig:MRVM_CE_and_LPs}).

\subsection{Experimental Results for using MVNNs on multiset domains}
\label{subsec:app:results_mvnns_multiset_domains}
In this section, we experimentally evaluate the benefits of mMVNNs over MVNNs for multiset domains  by comparing the efficiency results of ML-CCA with mMVNNs (ML-CCA-mMVNN) and MVNNs (ML-CCA-MVNN) as the neural network architecture. 
We only compare the two architectures for the two multiset domains, i.e., SRVM and MRVM, as in the other two domains, the two architectures are mathematically equivalent. 
It is important to note that, in the case of MVNNs, even though the network has not incorporated the prior knowledge that some items are identical copies of each other, the price generation algorithm \emph{does make use} of that prior information:
in the case of MVNNs, our price generation algorithm appropriately calculates the total demand of each item as the aggregate demand of that item's copies, and then sets the same price for all of those copies. 

Those results are presented in 
\Cref{tab:details_SRVM_MRVM_efficiency_MVNNs_vs_mMVNNs}. 
In the SRVM domain, the performance of the two architectures is almost identical in terms of both clearing potential and efficiency.
The only statistically significant result is that mMVNNs slightly outperform MVNNs in terms of efficiency if the raised clock bids heuristic is used in the supplementary round.
In the MRVM domain, (linear) clearing prices never exist, so no architecture ever clears the market. 
In terms of efficiency, the only statistically significant result in this domain is that MVNNs actually achieve $0.06$\% higher efficiency compared to mMVNNs after the clock phase.
We have good reason to believe that the strong performance of MVNNs is because, as explained in the previous paragraph, the price generation algorithm \textit{does make use} of the fact that some items are identical copies of each other, even though the MVNNs do not make use of that information. 
However, we are currently unable to verify this hypothesis, as the SATS simulator only supports demand queries where copies of the same item have identical prices.

\fi
\end{document}